\documentclass[preprint]{acmart}
 \renewcommand\footnotetextcopyrightpermission[1]{} 
\acmJournal{PACMPL}
\acmVolume{1}
\acmNumber{POPL} 
\acmArticle{1}
\acmYear{2021}
\acmMonth{1}
\acmDOI{} 
\startPage{1}

\setcopyright{none}

\usepackage{subcaption} 

\usepackage{amsmath}
\usepackage{stmaryrd}
\usepackage{proof}
\usepackage{graphicx}
\usepackage{xypic}
\usepackage{tikz}
\usetikzlibrary{calc}

\usepackage[disable
]{todonotes}

\def\eg{{\it e.g.}}
\def\ie{{\it i.e.}}
\def\lat{\mbox{$\lambda$-term}}

\newcommand\todom[2][]{\todo[color=blue!20,#1]{#2}} 
\newcommand\todod[2][]{\todo[color=yellow!20,#1]{#2}} 

\newtheorem{theorem}{Theorem}
\newtheorem{definition}[theorem]{Definition}
\newtheorem{lemma}[theorem]{Lemma}

\newtheorem{proposition}[theorem]{Proposition}

\newtheorem{remark}[theorem]{Remark}

\newcommand{\refth}[1]{Theorem~\ref{th:#1}}
\newcommand{\reflemma}[1]{Lemma~\ref{lemma:#1}}
\newcommand{\refsublemma}[2]{Lemma~\ref{lemma:#1}.\ref{#1:#2}}
\newcommand{\refprop}[1]{Proposition~\ref{prop:#1}}
\newcommand{\refdef}[1]{Definition~\ref{def:#1}}
\newcommand{\reffig}[1]{Fig.~\ref{fig:#1}}
\newcommand{\refsubfig}[1]{Fig.~\ref{subfig:#1}}
\newcommand{\refsect}[1]{Sect.~\ref{sect:#1}}

\newcommand{\replemma}[1]{\noindent\textsc{Lemma~\ref{lemma:#1}.}}

\newcommand{\gsep}{\mathrel{|}}

\newcommand{\ite}[3]{\mathsf{if}(#1, #2, #3)}

\newcommand{\PCF}{\ensuremath{\mathsf{PCF}_\tyR}}

\newcommand{\funsa}{\phi} 
\newcommand{\funsb}{\psi} 
\newcommand{\funsc}{\chi} 

\newcommand{\Relu}{\mathsf{ReLU}}
\newcommand{\Crelu}[1][]{\mathsf{cReLU}_{#1}}
\newcommand{\Floor}{\mathsf{Floor}}

\newcommand{\Sid}{\mathsf{SillyId}}
\newcommand{\EqProj}{\mathsf{EqProj}}

\newcommand{\logic}[2][\Gamma]{\mathrm{P}_{#1}(#2)} 

\newcommand{\fix}[3][]{\mathsf{fix}_{#1}#2.#3}
\newcommand{\tyR}{\mathsf R}
\newcommand{\numeral}[1]{#1}
\newcommand{\Real}{\mathbb R}
\newcommand{\Nat}{\mathbb N}

\newcommand{\sem}[2][\Gamma]{\left\llbracket #2\right\rrbracket_{#1}}
\newcommand{\pto}{\rightharpoonup}
\newcommand{\red}[1][]{\xrightarrow{#1}}
\newcommand{\reds}[1][]{\mathrel{\mathop{\xrightarrow{#1}}\!{}^\ast}}

\newcommand{\isub}[2]{\{#1/#2\}}
\newcommand{\ctxthole}{\{\cdot\}}
\newcommand{\Ctxt}{\mathsf C}
\newcommand{\Ctxtp}{\mathsf D}
\newcommand{\ctxt}[2][\Ctxt]{#1\{#2\}}

\newcommand{\hCtxt}{\mathsf{H}}
\newcommand{\hctxt}[2][\hCtxt]{#1\{#2\}}

\newcommand{\apx}{\sqsubset}
\newcommand{\Stab}[1]{\mathrm{S}(#1)}
\newcommand{\UStab}[1]{\mathrm{U}(#1)}
\newcommand{\st}{\mathrel{|}}
\newcommand{\seq}[1]{\mathbf{#1}}
\newcommand{\trapx}{\mathrel{\text{\raisebox{-4pt}{$\stackrel{\text{\scalebox{1.3}{$\sqsubset$}}}{\sim}$}}}}

\newcommand{\Ball}[2]{B_{#1}(#2)}

\newcommand{\Jacob}{\mathop{\mathcal J}\!}
\newcommand{\FSymb}{\overrightarrow{\mathbf{D}}}
\newcommand{\Fwd}[2][]{\FSymb_{#1}(#2)}
\newcommand{\BSymb}{\overleftarrow{\mathbf{D}}}
\newcommand{\Bkw}[2][]{\BSymb_{#1}(#2)}
\newcommand{\FBSymb}{\mathbf{D}}
\newcommand{\FB}[2][]{\FBSymb_{#1}(#2)}
\newcommand{\FBgrad}[2][]{{\mathop{grad}}_{#1}(#2)}
\newcommand{\Fgrad}[2][]{\overrightarrow{\mathop{grad}}_{#1}(#2)}
\newcommand{\Bgrad}[2][]{{\overleftarrow{\mathop{grad}}_{#1}(#2)}}
\newcommand{\FBgradSymb}[1][]{{\mathop{grad}}_{#1}}

\newcommand{\interior}[1]{\mathrm{int}(#1)}
\newcommand{\border}[1]{\mathrm{bor}(#1)}
\newcommand{\diffdom}[2][]{\mathrm{d}^{#1}(#2)}
\newcommand{\dom}[2][]{\mathop{\Downarrow^{#1}}\! #2} 
\newcommand{\pair}[1]{\left\langle #1\right\rangle}

\newcommand{\pr}[1]{\pi_{#1}} 
\newcommand{\fixapx}{\leq}
\newcommand{\expd}{\lhd}

\newcommand{\id}{\mathrm{id}}

\newcommand{\clone}{\mathbf}

\begin{document}

\title{Automatic Differentiation in PCF}         


\author{Damiano Mazza}
\affiliation{
  \institution{CNRS}            
  \streetaddress{UMR 7030, LIPN, Universit\'e Sorbonne Paris Nord}
  \country{France}                    
}
\email{Damiano.Mazza@lipn.univ-paris13.fr}          

\author{Michele Pagani}
\affiliation{
  \department{IRIF UMR CNRS 8243}             
  \institution{Universit\'e de Paris}           
  \country{France}                   
}
\email{pagani@irif.fr}         

\begin{abstract}
We study the correctness of automatic differentiation (AD) in the context of a higher-order, Turing-complete language (PCF with real numbers), both in forward and reverse mode. Our main result is that, under mild hypotheses on the primitive functions included in the language, AD is almost everywhere correct, that is, it computes the derivative or gradient of the program under consideration \emph{except} for a set of Lebesgue measure zero. Stated otherwise, there are inputs on which AD is incorrect, but the probability of randomly choosing one such input is zero. Our result is in fact more precise, in that the set of failure points admits a more explicit description: for example, in case the primitive functions are just constants, addition and multiplication, the set of points where AD fails is contained in a countable union of zero sets of non-identically-zero polynomials.

\end{abstract}



\maketitle

\section{Introduction}
\label{sect:Intro}
Automatic differentiation (AD) provides efficient methods for computing the derivative (or, more generally, the gradient or Jacobian) of a function specified by a computer program. Since computing derivatives is a key ingredient in the resolution of all sorts of optimization problems, it is not surprising that AD grew into a large field with applications to a host of scientific domains, most notably machine learning~\cite{Baydin:ADsurvey}.

Traditionally, AD focused on first-order imperative programs and, although its techniques allow the presence of flow control instructions and loops~\cite{Joss,Speelpenning,Beck1994}, its scope was often limited to straight-line programs (also know as \emph{computational graphs}~\cite{Goodfellow-et-al}), which were enough for most practical purposes, such as expressing neural networks. After the advances in deep learning of the last years, this is no longer the case: neural network architectures are now ``dynamic'', in the sense that the input may influence the shape of the net, and expressing such architectures requires resorting a priori to the full power of a modern programming language, yielding what some have called \emph{differentiable programming}~\cite{LeCun:diff}. This evolution of deep learning spurred the rapid development of differentiable programming frameworks \cite{PyTorch,TensorFlow} and, at the same time, received much attention in programming languages (PL) research for establishing its theoretical foundations~\cite{Elliott,Purdue,ShaikhhaFVJ19,AbadiP20,POPL2020,HuotStatonVakar}.

From the viewpoint of PL theory, AD methods boil down to \emph{program transformations}: writing $\Real$ and $\tyR$ for the set and type of real numbers, respectively, we receive as input a program $M:\tyR^n\to\tyR$ computing a (possibly partial) function $\sem[]{M}:\Real^n\pto\Real$ whose gradient $\nabla\!\sem[]M$ exists in a set $\diffdom{M}\subseteq\Real^n$ (called the \emph{domain of differentiability}), and we must output another program $\FBgrad M$ computing $\nabla\!\sem[]M$. The crucial features that one typically asks of such transformations are:
\begin{itemize}
	\item[(i)] \textbf{efficiency}: asymptotically, evaluating $\FBgrad M$ is not more costly than evaluating $M$;
	\item[(ii)] \textbf{soundness}: $\FBgrad{M}(\seq r)$ evaluates to $\nabla\!\sem[]M\!(\seq r)$ for all $\seq r\in\diffdom{M}$.
\end{itemize}
Notice that there is a tension between efficiency and soundness: implementing the definition of derivative as a limit gives a trivially sound (to an arbitrary degree of precision) but unacceptably inefficient method. Conversely, more efficient transformations tend to be more complex (\eg, reverse mode is more complex than forward mode, see below) and their soundness more difficult to prove. Also observe that we are only interested in the correctness of the result when $\nabla\!\sem[]{M}$ is defined; in case $\seq r\not\in\diffdom{M}$, the evaluation of $\FBgrad{M}(\seq r)$ may give anything, including (but not necessarily!) divergence.

Another highly desirable feature of $\FBgradSymb$ is \emph{modularity}: if $P$ is a subprogram of $M$ then $\FBgrad{P}$ is a subprogram of $\FBgrad{M}$ or, if this is not literally the case, the computation of the former may be reused in computing the latter. Indeed, it has been known from the early days of AD that modularity offers a path to attaining both efficiency and soundness: the program $M$ is decomposed into elementary blocks whose gradients are immediately computable, and $\FBgrad{M}$ is obtained by assembling these transformed blocks following the structure of $M$. In this way, the execution of $\FBgrad{M}$ mimics that of $M$, yielding efficiency, and soundness relies on the so-called \emph{chain rule} of calculus, which assures us that the derivative of a compound function may be expressed in terms of the derivative of its components. Furthermore, one sees that there are two ``dual'' ways of assembling the transformed blocks to form $\FBgrad{M}$: a covariant way, yielding \emph{forward mode} AD, and a contravariant way, yielding \emph{reverse mode} AD (this will be explained in Sect.~\ref{subsect:ad}).

The theory of AD transformations has by now been developed to considerable depth by several authors: Pearlmutter and Siskind first pointed out that reverse mode AD, commonly known as \emph{backpropagation}, may be naturally expressed in terms of higher-order programs, and used this idea to develop a differentiable variant of Scheme~\cite{PearlmutterSiskind}; more recently, Elliott emphasized functoriality as a systematic way of understanding the modular nature of AD transformations~\cite{Elliott}; the work \cite{Purdue} introduced Lantern, a fully general differentiable programming framework in which the notion of delimited continuation is used to correctly handle memory updates during backpropagation; finally, Brunel, Mazza and Pagani showed that the continuation-passing machinery at work in \cite{Purdue} (and, implicitly, in \cite{PearlmutterSiskind}) may be understood in terms of linear negation (in the sense of Girard's linear logic), giving a purely functional transformation for reverse mode AD and a conceptually clean analysis of its efficiency in terms of a ``linear factoring'' evaluation rule~\cite{POPL2020}. On the semantics side, Abadi and Plotkin studied denotational semantics for a first order differentiable language~\cite{AbadiP20} and Huot, Staton and V\'ak\'ar gave a uniform approach to proving soundness of AD (forward and reverse) for simply-typed programs based on a diffeology semantics~\cite{HuotStatonVakar}.

Nevertheless, an analysis of soundness of AD transformations for a fully general programming language is currently missing: the above-mentioned work is either fully general but lacks soundness proofs~\cite{PearlmutterSiskind,Purdue} or proves soundness in a restricted setting (first order~\cite{AbadiP20} or simply-typed \mbox{$\lambda$-calculi}~\cite{POPL2020,HuotStatonVakar,BartheCLG}). Filling this gap is precisely the contribution of the present paper: we study the soundness of AD transformations in the setting of the (idealized) functional programming language \PCF, a variant with real numbers of Plotkin's famous Turing-complete language~\cite{Plotkin:PCF}. For forward mode, we use the standard transformation described for instance in \cite{Purdue}. For reverse mode, since \PCF\ is purely functional, we consider an extension of the transformation introduced in \cite{POPL2020}, albeit simplified in that we leave linearity aside, since that is only needed for efficiency and here we are merely interested in soundness. The extension, which is a contribution of this paper in its own right, concerns conditional statements and fixpoints, which are not dealt with in \emph{loc.~cit}.

The first relevant observation is that, in presence of conditionals, soundness in the sense of statement (ii) above actually fails. Consider the program
$$\mathsf{SillyId} \quad:=\quad \lambda x^\tyR.\textsf{if}\,x=0\,\mathsf{then}\,0\,\mathsf{else}\,x.$$
We clearly have that $\mathsf{SillyId}:\tyR\to\tyR$ and that $\sem[]{\mathsf{SillyId}}$ is the identity function. We therefore expect $\FBgrad{\mathsf{SillyId}}$ to compute the constant function $1$. And yet, by modularity/functoriality, AD transformations will give something like
$$\FBgrad{\mathsf{SillyId}}\quad=\quad \lambda x^\tyR.\textsf{if}\,x=0\,\mathsf{then}\,0\,\mathsf{else}\,1,$$
which obviously gives the wrong result for $x=0$. This phenomenon, which is well known in the AD community~\cite{Beck1994}, is due to functoriality turning a syntactic discontinuity into a semantic one. Notice that, although the above example is indeed quite silly, similar situations may happen in a non-trivial neural network with rectified linear unit activation: if
$$\mathsf{ReLU} \quad:=\quad \lambda x^\tyR.\textsf{if}\,x\leq 0\,\mathsf{then}\,0\,\mathsf{else}\,x,$$
then $\mathsf{ReLU}(x)-\mathsf{ReLU}(-x)$ behaves exactly as $\mathsf{SillyId}(x)$. Also, using recursive definitions, it is easy to obtain programs on which AD fails on infinitely many inputs, even uncountably many in case of programs of type $\tyR^n\to\tyR$ with $n>1$ (simply consider $\lambda x^\tyR.\lambda y^\tyR.\textsf{if}\,x\cdot y= 0\,\mathsf{then}\,0\,\mathsf{else}\,x\cdot y$).

So, to each given \PCF\ program $M:\tyR^n\to\tyR$, we may assign a set $\mathrm{Fail}(M)\subseteq\diffdom{M}$ of points on which AD is unsound. Notice once again that we disregard what lies outside of $\diffdom{M}$, where $\FBgrad{M}$ is free to behave arbitrarily. For instance, $\FBgrad{\mathsf{ReLU}}$ is something like $\lambda x^\tyR.\textsf{if}\,x\leq 0\,\mathsf{then}\,0\,\mathsf{else}\,1$, which evaluates to $0$ when $x=0$, even though $\sem[]{\mathsf{ReLU}}$ is not differentiable in $0$. Morally, we cannot say that AD is ``wrong'' when there is no ``right'' value to compare it to.

After toying with more examples, one is led to conjecture that $\mathrm{Fail}(M)$ is always of measure zero (in the sense of the standard Lebesgue measure on $\Real^n$), so one may hope to establish an ``almost-everywhere'' relaxation of (ii):
\begin{itemize}
	\item[(ii')] \textbf{ae-soundness:} $\FBgrad{M}(\seq r)$ evaluates to $\nabla\!\sem[]M\!(\seq r)$ for all $\seq r\in\diffdom{M}$ \emph{except} on a set of measure zero.
\end{itemize}
\todod{Added paragraph} This is exactly the main result of our paper. Let us stress that, considering that ``full'' soundness is impossible, such a result is quite meaningful in practice because of the link between the Lebesgue measure and the standard understanding of randomness on $\Real^n$. In typical deep learning applications, weights are initialized ``at random'' and later updated via gradient descent in order to minimize a loss function. Technically, this means that the weights evolve following some standard probability distribution, which always arises from integrating a probability density function with respect to the Lebesgue measure. So, an informal way of stating (ii') is that \emph{AD almost never fails}, in the sense that, according to the standard definition of probability, the likelihood of computing wrong derivatives\todom{added} during gradient descent is zero.

The main result itself is articulated in \refth{Sound} and \refth{Unsound}. The first result takes care of soundness proper: we define, for any given \PCF\ program $M:\tyR^n\to\tyR$ whose domain (\ie, the inputs on which it converges) is $\dom M$, a set $\Stab{M}\subseteq\dom M$ of \emph{stable points}, and \refth{Sound} affirms that statement (ii) holds on $\diffdom{M}\cap\Stab{M}$. Then, we establish in \refth{Unsound} that the set $\dom M\setminus\Stab{M}$ of \emph{unstable points} of $M$ is of measure zero. Since $\diffdom{M}\subseteq\dom M$, this proves statement (ii').

The intuition behind stable points is the following. Take $M:\tyR^n\to\tyR$, $\seq r\in\Real^n$ and trace the execution of $M(\seq r)$. This means, in particular, unfolding the recursive definitions in $M$ and choosing, for each instance of a conditional statement of $M$, the ``$\mathsf{then}$'' or ``$\mathsf{else}$'' branch. One obtains thus a program with no conditionals and no fixpoints, \ie, a \emph{simply-typed \lat}, which is said to \emph{trace} $M(\seq r)$. Such a program depends of course on $\seq r$. If, however, there exists a simply-typed \lat\ $t$ and an open neighborhood $U\subseteq\Real^n$ of $\seq r$ (in the standard topology) such that $t(\seq r')$ traces $M(\seq r')$ for all $\seq r'\in U$, then $\seq r$ is \emph{stable} (\refdef{Stab}). For instance, any $r\neq 0$ is stable for the $\mathsf{ReLU}$ program given above: there always exists an open interval $I$ around $r$ such that either $\mathsf{Zero}:=\lambda x^\tyR.0$ or $\mathsf{Id}:=\lambda x^\tyR.x$ traces $\mathsf{ReLU}$ on $I$, depending on whether $r<0$ or $r>0$, respectively. On the other hand, $0$ is an unstable point of $\mathsf{ReLU}$: any open interval around $0$ must contain negative points, on which $\mathsf{ReLU}$ is traced by $\mathsf{Zero}$, and positive points, on which $\mathsf{ReLU}$ is traced by $\mathsf{Id}$, and of course $\mathsf{Zero}\neq\mathsf{Id}$. 

The proof of \refth{Sound} uses stability to reduce the soundness of AD on \PCF\ to the soundness of AD on simply-typed \lat s, which may be established in various ways~\cite{POPL2020,HuotStatonVakar}. The idea is simple: if $\seq r\in\Stab M$, then $M$ ``behaves like'' a simply-typed \lat\ $t$ in an open neighborhood of $\seq r$, and we know that AD works for $t$ everywhere, so it ``must'' work for $M$ on $\seq r$. Although intuitively clear, the actual argument is surprisingly subtle. First, the definition of trace (\refdef{TrApx}) is not obvious, \todod{Added sentence} due to non-uniformity issues introduced by higher types: two copies of the same higher-order subterm may be traced in different ways, as explained in the example given at the beginning of Sect.~\ref{subsection:traces}. Second, knowledge of the correctness of $\FBgrad{t}(\seq r)$ does not immediately imply the correctness of $\FBgrad{M}(\seq r)$ and some non-trivial work is needed to show that they behave similarly.

The measure-zero bound on unstable points (\refth{Unsound}), albeit obtained via a standard logical predicate argument, also requires non-trivial elements, most notably the notion of \emph{complete quasicontinuity}, which is needed to account for the behavior of unstable points under composition, and the related notion of \emph{quasivariety}. Although the exact meaning of these notions depends on the choice of primitives of \PCF\ (\ie, the basic real functions included in the language), under mild assumptions a quasivariety is always of measure zero, and \refth{Unsound} states precisely that $\mathrm{Fail}(M)$ is a quasivariety. For example, when the primitives are just constants, addition and multiplication, quasivarieties are arbitrary subsets of countable unions of zero sets of polynomials. Furthermore, our \reflemma{Adequacy} implies properties of \PCF-definable functions which, as far as we can tell, were previously unknown. For example, if $f:\Real^n\pto\Real$ is definable in \PCF, and if $U\subseteq\Real$ is open, then in general the border of $f^{-1}(U)$ (\ie, $f^{-1}(U)$ minus its interior) is not empty because conditionals introduce discontinuities, but \emph{it is always a quasivariety}. Similarly, \emph{the set of zeros of $f$ is always the disjoint union of an open set and a quasivariety}.

\paragraph*{Related work.}
We already mentioned some of the relevant previous work at the interface between AD and PL theory and stressed that our contribution here is to study the soundness of AD in a fully general setting (higher-order, Turing-complete language), something which, as far as we know, was lacking.

We also mentioned that the unsoundness of AD in presence of conditional statements is well known~\cite{Beck1994}. Surprisingly, though, recent PL work on the subject acknowledges this problem only sporadically, \eg~\cite{AbadiP20}\todod{Changed again}. 
The solution proposed therein is restricting to \emph{continuous} Boolean conditions, meaning that the inverse images of the two Boolean values along such conditions are open. For example, testing for zero or for non-positivity are not continuous, because the inverse image of $\mathsf{true}$ is the singleton $\{0\}$ or $]\!-\!\infty,0]$, respectively, which are not open. With this limitation, the authors prove statement (ii) for a Turing-complete first-order language. The benefit of Abadi and Plotkin's approach is allowing a denotational semantics modeling the $\FBgradSymb$ operator, but it has the drawback of representing standard total functions with programs diverging on singularities (for example, $\mathsf{ReLU}$ yields a program diverging in $0$), which is somewhat unexpected. We discuss this further at the end of Sect.~\ref{subsect:ad} and in \refsect{Conclusions}.

Our approach, in the wake of a large part of the AD literature, is to stick to the standard semantics and provide a bound on the unsoundness of AD, as precise as possible. This approach dates back to the Seventies, as far as we know to Joss's Ph.D.~thesis \cite{Joss}, who proved statement (ii') for forward mode AD in the context of an imperative language with variable assignments, basic arithmetic functions (sum, multiplication, division), conditional statements and gotos. So our result may be seen as an extension Joss's theorem in several directions: to a higher-order language; to a wider set of primitive functions (as long as they form an \emph{admissible clone}, \refdef{admissibility}); and to reverse mode AD. Additionally, \refth{Unsound} is more precise than (ii'), because it characterizes the set of failure points $\mathrm{Fail}(M)$ as a quasivariety, which is a rather special example of negligible set. Some discussion about this point is given in \refsect{Conclusions}, in particular the proof of \refth{Unsound} hints to methods for automatically computing at least some overapproximation of $\mathrm{Fail}(M)$ statically from the structure of $M$.

\todod{Added two paragraphs}
From a broader perspective, variants of PCF with real numbers similar to the one studied here have been considered in the literature, \eg\ \cite{Escardo} and \cite{DiGianantonioEdalat}. The latter actually also considers AD, but it is not about correctness and is quite different in spirit, being more focused on denotational semantics. There is also a recent line of work whose goal is to understand the non-differentiable points of program-defined functions, such as \cite{ZhouGKRYW,MakOngPaquetWagner}, including in the context of AD \cite{LeeYuRivalYang}, where it is a natural and important question~\cite{GriewankWalther}. 
Non-differentiability and unsoundness of AD have an important point in common: they are both introduced by conditionals. This explains why some notions used in our paper also crop up in the study of non-differentiability, such as zero sets of analytic functions~\cite{ZhouGKRYW}. However, let us underline that the two issues are orthogonal: from our perspective, PCF-definable functions might as well have been differentiable \emph{everywhere} (as in the $\mathsf{SillyId}$ example), what matters is that AD still makes mistakes and we wish to understand them. Whether our techniques also yield tools for describing the set of non-differentiable points of PCF-definable functions and, in that case, exactly how they relate to the above-mentioned work is an interesting question which we leave for the future.

Finally, let us mention that our notion of stable point (\refdef{Stab}), which is new as far as we know, is based on a concept of trace (\refdef{TrApx}) belonging to the same circle of ideas as Ehrhard and Regnier's Taylor expansion \cite{EhrhardRegnier:Taylor1,EhrhardRegnier:Taylor2} and the modern understanding of intersection types in the spirit of Mazza, Pellissier and Vial's work~\cite{Mazza:Habilitation,MazzaPellissierVial}. This extremely general perspective allows the definition of finitary approximations of programs at the level of the operational semantics, rather than denotational, as was the case traditionally. Our proof techniques are therefore not \emph{ad hoc} for our current purposes and may be expected to have applications beyond the present paper.

\paragraph{Contents of the paper.}
\refsect{Prelim} introduces the language \PCF\ with its rewriting relation (\reffig{pcf}) and the AD transformations (\reffig{AD} and Equations~\eqref{eq:grad_f},~\eqref{eq:grad_b}). Our results are quite general and do not depend on a specific operational semantics but apply to a wide family of them (\refprop{reduction_independence}), including the standard ones (call-by-value, call-by-name, etc.). We also introduce admissibility of primitives (\refdef{admissibility}) and the associated topological notions, among which quasivarieties (\refdef{quasivariety}). \refsect{Soundness} and \refsect{Unsoundness} are the heart of the paper, containing the proofs of \refth{Sound} and \refth{Unsound}, respectively, as described above. \refsect{Conclusions} concludes the paper, discussing the results. Most proofs of Sections~\ref{sect:Prelim}, \ref{sect:Soundness} and~\ref{sect:Unsoundness} are postponed to Appendices~\ref{app:prel},~\ref{app:sound} and~\ref{app:unsound}, respectively.

\paragraph{Notations.}
We write $f:A\pto B$ to say that $f$ is a partial function from a set $A$ to a set $B$. In that case, $\dom f$ will denote the subset of $A$ on which $f$ is defined. Given two partial functions $f,g: A\pto B$ and $a\in A$, the equality $f(a)=g(a)$ means that either both $f(a)$ and $g(a)$ are defined and equal, or that both $f(a)$ and $g(a)$ are undefined. The notation $f(a)\neq g(a)$ of course is understood as the logical negation of that. Given a subset $A'\subseteq A$, we write $f\vert_{A'}$ for the restriction of $f$ to $A'$. If $A$ is endowed with a complete measure $\lambda$ (typically $A$ is $\mathbb R$ and $\lambda$ is the Lebesgue measure), then we say that $f$ and $g$ are \emph{almost everywhere equal}, and we write $f\sim g$, if $\lambda(\{a\in A\st f(a)\neq f(b)\})=0$. This is equivalent to the existence of two subsets $A',Z$ of $A$ such that $\dom f\cup\dom g\subseteq A'\cup Z$, $\lambda(Z)=0$ and $f\vert_{A'}=g\vert_{A'}$, a fact which is implicitly used in the proof of \refprop{reduction_independence}.\footnote{Proof of the equivalence: let $D:=\{a\in A\st f(a)\neq f(b)\}$. If $\lambda(D)=0$, then we may take $A':=(\dom f\cup\dom g)\setminus D$ and $Z:=D$. Conversely, given $A'$ and $Z$ with the required properties, notice that $D\subseteq\dom f\cup\dom g$, hence $D\subseteq A'\cup Z$. But observe that $D\cap A'=\emptyset$, so $D\subseteq Z$, which gives us $\lambda(D)=0$.}

We write $\Ball\varepsilon{\seq r}$ for the open ball of $\Real^n$ of radius $\varepsilon$ centered at $\seq r\in\Real^n$. \todod{Changed following Reviewer C's suggestion.} Given $f:\Real^n\pto\Real^m$, we denote by $\diffdom{f}$ the \emph{domain of differentiability} of $f$, defined to be the set of all $\seq r\in\Real^n$ where $f$ is differentiable in the sense that the total derivative of $f$ at $\seq r$ exists, \ie, $f$ admits a best linear approximation at $\seq r$.  We denote by $\Jacob f$ the \emph{Jacobian} of $f$. We recall that, if $\partial_if_j$ denotes the partial derivative of the $j$-th component of $f$ with respect to its $i$-th parameter, then within $\diffdom f$ the Jacobian is equal to the $m\times n$ matrix $(\partial_if_j)$. If $m=1$, the Jacobian is called \emph{gradient} and denoted by $\nabla\!f$. As a special case of the above, within $\diffdom{f}$ we have $\nabla\!f=(\partial_1f,\ldots,\partial_nf)$. Although it may happen in general that all partial derivatives exist without the Jacobian/gradient being defined, it will never be the case in what follows because \emph{we will always work within $\diffdom{f}$}.

\section{PCF with Real Numbers}
\label{sect:Prelim}
\subsection{Terms and Semantics}
\label{sect:Terms}
\begin{figure*}
\begin{subfigure}{\textwidth}
	$$A,B::= \tyR \gsep A\to B \gsep A_1\times\dots\times A_n$$
	\caption{Types.}
	\label{fig:types}
\end{subfigure}

\begin{subfigure}[b]{\textwidth}
	\vspace{10pt}
	$$
	\infer{\Gamma,x^A\vdash x:A}{}
	\qquad\qquad
	\infer{\Gamma\vdash \funsa(M_1,\dots,M_k):\tyR}{\funsa:\tyR^k\to\tyR,&\Gamma\vdash M_1:\tyR, &\dots,&\Gamma\vdash M_k:\tyR}
	$$
	\medskip
	$$
	\infer{\Gamma\vdash \lambda x^{A}.M:A\to B}{\Gamma, x^{A}\vdash M:B}
	\qquad\qquad	
	\infer{\Gamma\vdash MN:B}{\Gamma\vdash M:A\rightarrow B,\quad \Gamma\vdash N:A}
	$$
	\medskip
	$$
	\infer{\Gamma\vdash \pair{M_1, \dots, M_k}:A_1\times\cdots\times A_k}{\Gamma\vdash M_1:A_1, &\dots,&\Gamma\vdash M_k:A_k}
	\qquad\qquad	
	\infer{\Gamma\vdash \pi_i^k M:A_i}{\Gamma \vdash M:A_1\times\cdots\times A_k}
	$$
	\medskip
	$$
	\infer{\Gamma\vdash \ite PMN:A}{\Gamma\vdash P:\tyR & \Gamma\vdash M:A & \Gamma\vdash N:A}
	\qquad\qquad
	\infer{\Gamma\vdash \fix {f^{A\to B}}M:A\to B}{\Gamma, f:A\to B \vdash M:A\to B}\\
	$$
\caption{Terms and typing rules.}\label{fig:terms_types}
\end{subfigure}

\medskip
\begin{subfigure}[b]{\textwidth}
	\begin{align*}
		(\lambda x.M) N &\red M\isub{N}{x}
		&
		\pi_i^k\!\pair{M_1,\dots, M_k} &\red M_i
		&
		\funsa(\numeral{r_1},\dots, \numeral{r_k}) &\red \numeral{\sem[]{\funsa}\!(r_1,\dots, r_k)}
	\end{align*}
	\begin{align*}		
		\ite{\numeral r}MN &\red 
		\begin{cases}
			M & \text{if } r\leq 0\\
			N &\text{if }  r>0
		\end{cases}
		&
		\fix f M &\red M\isub{\lambda x.(\fix f M)x}{f}
	\end{align*}
\caption{Rewriting steps.}\label{fig:reduction}
\end{subfigure}
\caption{The language \PCF{} over the ground type $\tyR$ of real numbers.}\label{fig:pcf}
\end{figure*}
The programming language we use, called \PCF, is introduced in~\reffig{pcf}. There is only one base type, $\tyR$, for real numbers. We use $n$-ary products for convenience. If one prefers, these may be seen as syntactic sugar defined from nullary and binary products. Throughout the paper, we stipulate that unary products are just identities: $\pair{M}$ and $\pr 1^1M$ both stand for $M$. We often omit the index $n$ in a projection $\pr i^n$, when inessential. The empty product type is denoted by $1$. Given a type $A$, we write $A^n$ for the $n$-fold product $A\times\dots\times A$.

The metavariables $\funsa, \funsc, \funsb$ range over a set of function symbols, each coming with a type of the form $\tyR^k\to\tyR$, where $k$ is the \emph{arity} of the symbol. We suppose that the set of function symbols contains at least all real numbers $r\in\Real$ as nullary symbols, called \emph{numerals}, as well as binary addition and multiplication, for which we use infix notation, \ie, $M+N$ and $M\cdot N$ stand for $+(M,N)$ and $\cdot(M,N)$, respectively. We also write $n$-ary sums as syntactic sugar. Each function symbol $\funsa:\tyR^k\to\tyR$ comes with a function $\sem[]\funsa:\Real^k\pto\Real$ and we assume that $\sem[]{r}$, $\sem[]{+}$ and $\sem[]{\cdot}$ are the corresponding numbers and operations on real numbers.  The functions $\sem[]\funsa$ for $\funsa$ ranging over function symbols will be referred to as the \emph{primitive functions} of \PCF.

The notation $\ite PMN$ is just a compact form of $\mathsf{if}\,P\leq 0\,\mathsf{then}\,M\,\mathsf{else}\,N$. We call $P$ the \emph{guard} of the conditional.

The primitive functions one considers are usually very regular, typically analytic (\eg\ exponential, logarithm, trigonometric functions, sigmoid maps\ldots). \refsect{cqc} details the precise conditions that primitives must enjoy in order for our results to hold. The expressive power of \PCF\ considerably enlarges the set of definable functions, in particular introducing singularities. The following examples will be useful in the sequel:
\begin{equation}
\label{ex:pcfterms}
\hspace{-10pt}
\begin{array}{r@{\,}lc@{\!\!\!}r@{\,}l}
\Relu&
	:=\lambda x^\tyR.\ite x0x
&&
\mathsf{Int}&
	:=\lambda x^\tyR.\lambda y^\tyR.\ite{y-x}{\ite{y-x+1}{1}{0}}{1}
\\
\Sid &
	:=\lambda x^\tyR.\ite x{\ite{-x}0x}x
&&
\Floor &
	:=\lambda x^\tyR.\left(\fix f{
		\lambda n^\tyR.
			\ite{\mathsf{Int}\, x\, n}
			{
				n
			}
			{
				f(\ite{x}{n-1}{n+1})
			}
	}\right)0
\end{array}
\end{equation}
$\Relu$ and $\Sid$ are the \PCF{} terms corresponding to the namesake examples discussed in the Introduction. $\Relu$ is a typical example of a continuous non-differentiable function, having a corner in $0$. We know that AD fails on $\Sid$; Sect.~\ref{subsect:ad} will elaborate on this point. The program $\mathsf{Int}$ takes two inputs $x$ and $y$ and gives $0$ if $x\in[y,y+1[$, or $1$ otherwise. It is auxiliary to the definition of the $\Floor$ function, mapping a real number to the greatest integer less than or equal to it. $\Floor$ is an example of how recursive definitions may yield maps with an infinite number of non-differentiable points, being discontinuous on the integers.   

We denote by $M\isub{N}{x}$ the capture-free substitution of a term $N$ to the free occurrences of the variable $x$ in $M$. 

Note that the presence of the additive structure on $\tyR$ turns the product $\tyR^n$ into a biproduct. Indeed, the injection $\iota_i^n$, for $i$ such that $1\leq i\leq n$, may be defined as
\begin{equation}\label{eq:inj}
	\iota_i^n := \lambda x^\tyR. \pair{0,\dots,0,x,0,\dots,0},
\end{equation}
where the variable $x$ is in the $i$-th position of the $n$-tuple. We omit the index $n$ when inessential or clear from the context. 

For every type $A\to B$, we set $\Omega_{A\to B}:=\fix{f^{A\to B}} f$. We write just $\Omega$ when the type is irrelevant. Given a term $\Gamma,f:A\to B\vdash M:A\to B$ and $n\in\Nat\cup\{\infty\}$, we define $\fix[n]{f}{M}$ of type $A\rightarrow B$ as follows:
\begin{equation}
	\label{eq:fixpoint_approx}
		\fix[0]{f}{M} := \Omega,\qquad
		\fix[n+1]{f}{M}:= (\lambda f.M)(\lambda x.(\fix[n]{f}{M})x),\qquad
		\fix[\infty]{f}{M}:= \fix f M.
\end{equation}


Throughout the paper, we use boldface metavariables to denote sequences of metavariables, \ie, $\seq x=x_1,\ldots,x_n$ is a sequence of variables, $\seq M=M_1,\ldots,M_n$ is a sequence of terms, etc. The length of the sequence is specified only when necessary.

Let us introduce two particularly important classes of terms:
\begin{definition}[program, simple term]
	A typing environment $\Gamma$ is \emph{ground} whenever all of its variables have type $\tyR$. A \PCF\ term $M$ is called a \emph{program of arity $n$ and coarity $m$} whenever $x_1^{\tyR},\dots,x_n^{\tyR} \vdash M:\tyR^m$.
	
	A term is called \emph{simple} if it does not contain conditionals or\todod{"or" is fine.} fixpoints. Small Latin letters $t,u,v$ range over simple terms. Note that the subset of the simple terms of \PCF\ corresponds to the simply typed \mbox{$\lambda$-calculus} on the ground type $\tyR$ enriched with function symbols.
\end{definition}

A \emph{context} is a term with a single occurrence of a special variable $\ctxthole$, called the \emph{hole}. We use metavariables $\Ctxt,\Ctxtp$ to range over contexts. Given a context $\Ctxt$ and a term $M$, we write $\ctxt M$ for the term obtained by replacing the hole $\ctxthole$ of $\Ctxt$ with $M$, allowing the capture of the free variables in $M$ by the binders of $\Ctxt$. 

 The reduction relation $\red$ is defined by context closure of the rewriting rules in~\reffig{reduction}:
\begin{definition}[reduction]\label{def:reduction}
\reffig{reduction} defines the set of \emph{rewriting rules} of \PCF, which are pairs $R\red P$ of terms, with $R$ called the \emph{redex} and $P$ the \emph{contractum} of the rule. 

A \emph{reduction step} $\sigma$ is a triple $(\Ctxt, R, P)$ such that $\Ctxt$ is a context, and $R\red P$ is a valid reduction rule. We also write $\sigma: \ctxt{R}\red \ctxt{P}$ or, when the context $\Ctxt$ is irrelevant, simply $\sigma:M\red N$, for $M=\ctxt{R}$ and $N=\ctxt{P}$, and say that $\sigma$ \emph{fires} the redex $R$ in $M$. A term $M$ without redexes, \ie\ such that $M\neq \ctxt{R}$ for any $\Ctxt$ and any redex $R$, is said to be \emph{normal}, or a \emph{normal form}.

Given another context $\mathsf D$, we denote by $\ctxt[\mathsf D]{\sigma}$ the step $(\ctxt[\mathsf D]{\Ctxt}, R, P)$. Similarly we write $\sigma\isub{N}{x}$ for the triple   $(\Ctxt\isub{N}{x}, R\isub{N}{x}, P\isub{N}{x})$, which is still a valid reduction step. 

A \emph{reduction sequence} $\rho$ from a term $M$ to a term $N$, in symbols $\rho:M\reds N$, is either empty, in which case $N=M$, or a sequence of reduction steps $(\Ctxt_i, R_i, P_i)_{1\leq i\leq n}$ with $n\geq 1$ such that $M=\ctxt[\Ctxt_1]{R_1}$, $N=\ctxt[\Ctxt_n]{P_n}$ and for all $i<n$, $\ctxt[\Ctxt_i]{P_i}=\ctxt[\Ctxt_{i+1}]{R_{i+1}}$. We call $M$ the \emph{source} of $\rho$, $N$ its \emph{target} and $n$ the \emph{length} of $\rho$. We often identify a single reduction step and the corresponding reduction sequence of length $1$. The notations $\ctxt[\mathsf D]{\rho}$ and $\rho\isub{N}{x}$ are extended to reduction sequences in the obvious way.

Two reductions sequences $\rho:M\reds M'$ and $\rho':M'\reds M''$ compose in the obvious way to yield a reduction sequence $\rho\rho':M\reds M''$. Empty reduction sequences are the identities of such an operation.

A sequence $\rho$ is \emph{normalizing} if there is no reduction step $\sigma$ such that $\rho\sigma$ is a valid reduction sequence, or, equivalently, if the target of $\rho$ is a normal form. A term $M$ is called \emph{normalizing} if there exists a normalizing reduction from $M$. Otherwise $M$ is said to be \emph{diverging}.
\end{definition}

\PCF\ is Turing-complete: usual PCF~\cite{Plotkin:PCF} is essentially the fragment obtained by restricting to integer (including negative) numerals and sum. We recall some results about reduction which are completely standard (see \eg\ \cite{amadio98book}).

\begin{proposition}[confluence]\label{prop:confluence}
	Whenever $M\reds N_1$ and $M\reds N_2$, there exist $N$ such that $N_1\reds N$ and $N_2\reds N$. In particular, if $N_1$ and $N_2$ are normal forms, then $N_1=N_2$. 
\end{proposition}

\begin{proposition}[subject reduction]\label{prop:subject_reduction}
If $\Gamma\vdash M:A$ and $M\red M'$, then $\Gamma\vdash M':A$. 
\end{proposition}

\begin{proposition}[strong normalization for simple terms]\label{prop:sn}
For every simple term $t$ there exists $n\in\Nat$ such that the length of every reduction sequence starting from $t$ is bounded by $n$. 
\end{proposition}


A \emph{reduction strategy} $\mathcal S$ is a relation between terms $M$ of \PCF{} and occurrences of redexes in $M$. A strategy is called \emph{deterministic} whenever it is a partial function.  We denote by $\red[\mathcal S]$ the reduction relation defined by reducing only the redexes fired by $\mathcal S$ and we write $\mathcal S$-nf for a normal form of $\red[\mathcal S]$. We write $\beta$ for the maximal strategy, giving the reduction relation $\red$. Of course this strategy is not deterministic, however \refprop{confluence} assures that the normal form associated with a term is unique if it exists.

Reduction strategies are often defined by fixing a subset of redexes in \reffig{reduction} and a set of \emph{evaluation contexts}. An example which will be useful in the sequel is \emph{head reduction}, which is defined by taking all reduction rules but restricting their application to \emph{head contexts}, generated by the following grammar:
\begin{align*}
	\hCtxt &::= 
		\ctxthole
		\gsep \funsa(M_1,\dots,\hCtxt,\dots,M_k)
		\gsep \hCtxt N
		\gsep \pair{\hCtxt,N} 
		\gsep \pair{M,\hCtxt} 
		\gsep \pr i\hCtxt
		\gsep \ite{\hCtxt}{M}{N}.
\end{align*}
A \emph{head reduction step} is of the form $\ctxt[\hCtxt]{R}\to\ctxt[\hCtxt]{P}$ with $R\to P$ a rewriting step of \reffig{reduction} and $\hCtxt$ a head context. A \emph{head reduction sequence} is a reduction whose steps are all head reduction steps.

Observe that head reduction is not deterministic. Apart from the freedom in the order of the evaluation of the arguments of a function symbol, we allow to reduce within a pair as well as to project it: for instance, the term $\pi_1\!\pair{(\lambda x.x)y,M}$ may be decomposed either as $\ctxt[\hCtxt]{\pi_1\!\pair{(\lambda x.x)y,M}}$ with the empty head context $\hCtxt=\ctxthole$, or as $\ctxt[\hCtxt']{(\lambda x.x)y}$ with the head context $\hCtxt'=\pi_1\!\pair{\ctxthole,M}$, and the two decompositions fire different redexes.

A classic result~\cite{Barendregt,amadio98book} is that head reduction is a ``winning strategy'' for finding the $\beta$-normal form of a closed program:
\begin{proposition}
	\label{prop:headnf}
	Let $M$ be a normalizing closed program (\ie, of type $\tyR^n$) whose $\beta$-normal form is $N$. Then, there is a head reduction sequence $M\reds N$.
\end{proposition}

Let $\mathcal S$ be a reduction strategy, let $\Gamma=x_1^\tyR,\ldots,x_n^\tyR$ and let $\Gamma\vdash M:\tyR^m$. We define the partial function $\sem{M}^{\mathcal S}:\Real^n\pto\Real^m$ as follows:
\[
	\sem{M}^{\mathcal S}\!(r_1,\ldots,r_n)=
		\begin{cases}
		\pair{q_1,\dots,q_m}&\text{if  } M\isub{r_1}{x_1}\ldots\isub{r_n}{x_n}\reds[\mathcal S] \pair{q_1,\dots,q_m},\\
		\bot&\text{otherwise,}
		\end{cases}
\]
where $\bot$ means undefined. By \refprop{confluence}, $\sem{M}^{\mathcal S}$ is a well-defined partial function, even in case $\mathcal S$ is not deterministic, in fact if $M\isub{r_1}{x_1}\ldots\isub{r_n}{x_n}\reds[\mathcal S] \pair{q_1,\dots,q_m}$ and $M\isub{r_1}{x_1}\ldots\isub{r_n}{x_n}\reds[\mathcal S] \pair{q_1',\dots,q_m'}$ we have that $q_i=q_i'$ for each $i$, as $\red[\mathcal S]\,\subseteq\,\red$ and this latter is confluent. 

We write $\dom[\mathcal S] M$ for $\dom\sem{M}^{\mathcal S}$ and $\diffdom[\mathcal S]{M}$ for $\diffdom{\sem{M}^{\mathcal S}}$. In case $\mathcal S=\beta$, we omit the indices from the above notations and write simply $\sem M$, $\dom M$ and $\diffdom{M}$. 

\begin{proposition}\label{prop:head}
	Let $M$, $N$ and $P$ be programs ($P$ of coarity $1$), then the following relations hold:
	\begin{align*}
	\sem{\phi (M_1,\dots, M_k)}&=\sem\phi \circ \pair{\sem{M_1},\dots,\sem{M_k}},
	&
	\sem{\pair{M_1,\dots,M_k}}&=\pair{\sem{M_1},\dots,\sem{M_k}},
	\end{align*}
	\begin{align*}
	\sem{\ite PMN}&=\seq r \mapsto
		\begin{cases}
		\sem M(\seq r)&\text{ if }\sem P(\seq r)\leq 0,\\
		\sem N(\seq r)&\text{ if }\sem P(\seq r)> 0,\\
		\bot&\text{ if }\sem P(\seq r)=\bot.
		\end{cases}
	\end{align*}
\end{proposition}
\begin{proof}
	Immediate consequence of Proposition~\ref{prop:headnf}.
\end{proof}

\begin{proposition}\label{prop:strategies}
	For every program $M$ and strategy $\mathcal S$, we have, for all $\seq r\in \dom[\mathcal S] M$, $\sem{M}^{\mathcal S}\!(\seq r)=\sem{M}\!(\seq r)$. So, in particular, $\dom[\mathcal S] M\subseteq \dom M$, $\diffdom[\mathcal S]{M}\subseteq\diffdom{M}$ and, for every $\seq r\in\diffdom[\mathcal S]{M}$,  $\Jacob\sem{M}^{\mathcal S}\!(\seq r)=\Jacob\sem{M}\!(\seq r)$.
\end{proposition}
\begin{proof}
	Immediate consequence of the definitions and the confluence property.
\end{proof}

\subsection{Automatic Differentiation}
\label{subsect:ad}
\begin{figure*}
	\begin{subfigure}[b]{\textwidth}
		\begin{align*}
			\Fwd[n]\tyR&:=\tyR\times\tyR^n
			&
			\FB[n]{A\rightarrow B}&:=\FB[n]{A}\rightarrow \FB[n]{B}
			\\
			\Bkw[n]\tyR&:=\tyR\times\tyR^{\perp_n}
			&
			\FB[n]{A_1\times\cdots\times A_k} &:= \FB[n]{A_1}\times\cdots\times\FB[n]{A_k}
		\end{align*}
		\caption{The action over types}
		\label{subfig:AD:types}
	\end{subfigure}
	\begin{subfigure}[b]{\textwidth}
		\begin{align*}
		\Fwd[n]{\funsa(\seq M)}&:=
		\left(\lambda\seq z^{\tyR\times\tyR^n}.\pair{\funsa(\pi_1\seq z)\ ,\ \sum_{i=1}^k \partial_i\funsa(\pi_1\seq z)\cdot \pi_2z_i\ ,\ldots,\ \sum_{i=1}^k \partial_i\funsa(\pi_1\seq z)\cdot \pi_{n+1}z_i }
		\right)\Fwd[n]{\seq M}
		\\
		\Bkw[n]{\funsa(\seq M)}&:=
		\left(
		\lambda\seq z^{\tyR\times\tyR^{\bot_n}}.\pair{\funsa(\pi_1\seq z)\ ,\ \lambda a^\tyR.\sum_{i=1}^k \pi_2z_i(\partial_i\funsa(\pi_1\seq z) \cdot a)}
		\right)\Bkw[n]{\seq M}
		\end{align*}
		\caption{The action over a function symbol $\funsa$ of arity $k$. We suppose that, for every $1\leq i\leq k$, there is an associated function symbol $\partial_i \funsa$ of arity $k$ such that $\partial_i\!\sem[]{\phi}\!(\seq r)=\sem[]{\partial_i\phi}\!(\seq r)$ for every $\seq r\in\Real^k$ on which $\partial_i\!\sem[]\phi$ is defined. The writing $\funsa(\seq M)$ is a shortcut for $\funsa(M_1,\dots, M_k)$, and, similarly, $\lambda \seq z$ stands for the sequence of abstractions $\lambda z_1\dots \lambda z_k$ and $\FB[n]{\seq M}$ for the sequence of applications to $\FB[n]{M_1}\dots\FB[n]{M_k}$. These sequences are supposed empty if $k=0$.}\label{subfig:fb-ground}
	\end{subfigure}
	\begin{subfigure}[b]{\textwidth}
		\begin{align*}
			\FB[n]{x^A} &:= x^{\FB[n]{A}} &
			\FB[n]{\lambda x^{A}.M} &:= \lambda x^{\FB[n]{A}}.\FB[n] M &
			\FB[n]{MN} &:= \FB[n]{M}\FB[n]{N} 
		\end{align*}
		\begin{align*}
			\FB[n]{\pair{M_1,\dots,M_k}} &:= \pair{\FB[n]{M_1},\dots,\FB[n]{M_k}} &
			\FB[n]{\pi_i{M}} &:= \pi_i{\FB[n]{M}}
		\end{align*}
		\begin{align*}
			\FB[n]{\ite PMN}&:=\ite{\pi_1\FB[n] P}{\FB[n] M}{\FB[n] N}
			&
			\FB[n]{\fix{f^A}{M}} &:= \fix{f^{\FB[n] A}}{\FB[n]{M}}
		\end{align*}
		\caption{The action over the other programming primitives.}\label{subfig:fb-ho}
	\end{subfigure}
	\caption{The forward and reverse AD transformations. The symbol $\FBSymb$ denotes either one of them. The index $n$ refers to the gradient dimension and acts only over ground type annotations. We will omit the index when inessential or clear from the context.}
	\label{fig:AD}
\end{figure*}

As mentioned in the Introduction, the two modes of AD are performed by means of the program transformations $\FSymb$ (forward) and $\BSymb$ (reverse), outlined in full detail in \reffig{AD}. The cornerstone of these transformations is the chain rule, which describes the derivative of a composition of functions:
\begin{equation}
	\label{eq:chain-rule}
	(f\circ g)'(x) = f'(g(x))\cdot g'(x).
\end{equation}
This says in particular that $(-)'$ is not functorial (\ie, modular/compositional): the right-hand side of Equation~\eqref{eq:chain-rule} does not use only $g'(x)$ but also  $g(x)$. Functoriality may be achieved by transforming a map $f:\Real\to\Real$ into a map $\Fwd f:\Real^2\to\Real^2$ acting as follows:
\begin{equation}
	\label{eq:forward-composition}
	\Fwd f: \pair{z,\dot z}\mapsto\pair{f(z),f'(z)\cdot \dot z}.
\end{equation}
Referring to Equation~\eqref{eq:chain-rule}, the input $z$, which is often called \emph{primal} in the AD literature, corresponds to the output of $g$, whereas $\dot z$, called \emph{tangent}, corresponds to the output of $g'$. The reader may easily check functoriality: $\Fwd{f\circ g}=\Fwd{f}\circ\Fwd{g}$. This generalizes to $n$-ary maps in the definition of $\FSymb_n$ in \refsubfig{fb-ground}\footnote{Modulo some syntactic bureaucracy, \ie, the pair $\pair{z,\dot z}$ of \eqref{eq:forward-composition} corresponds in \refsubfig{fb-ground} to a single variable $z$ of type $\tyR\times\tyR^n$, whose components are obtained by using the projections $\pi_i$} and is the essential part of forward mode AD. The attribute \emph{forward} refers to the fact that the computation of both the primal and the tangent follows the input-to-output flow, in particular the derivative $f'(z)$ is computed after having accumulated in $\dot z$ the derivative of its input function. 

Notice that, if we denote by $|M|$ the size of a term $M$, we have that $|\Fwd[n]{\funsa(M)}|=O(n)+|\Fwd[n]{M}|$. So, supposing that $F$ consists only of function symbols and variables, evaluating both $F$ and $\Fwd[n]{F}$ on a given input requires a number of operations roughly equal to their size. But $|\Fwd[n]{F}|=O(n|F|)$, therefore the evaluation of $\Fwd[n]{F}$ is blown up by a factor of $n$ with respect to the evaluation of $F$. In applications to deep learning, $F$ is a loss\todod{I agree that loss is better terminology} function and $n$ is the number of learning parameters, which may be huge (hundreds of millions).

Luckily, AD offers a more efficient method for applying the chain rule in these cases, called reverse mode AD, or \emph{backpropagation}, because the idea is to accumulate the tangents in the reverse order with respect to the primals. More precisely, by taking the notation of~\eqref{eq:chain-rule},  the backpropagation $\Bkw f$ of $f$ first computes $f'(g(x))$ and then waits for the derivative $g'(x)$ in order to perform the multiplication. As first observed by~\cite{PearlmutterSiskind}, this mode may be naturally expressed in a functional programming language by replacing the tangent variables $\dot z$ with \emph{backpropagators} $z^*$ representing functions (in fact, special forms of continuations):
\begin{equation}
\label{eq:backward-composition}
\Bkw f : \pair{z,z^*}\mapsto\pair{f(z),\lambda a^\tyR.z^*(f'(z)\cdot a)}
\end{equation}
A backpropagator is a map $\Real\to\Real^n$ \emph{waiting} for a real number (the derivative of the next function) in order to achieve the computation of the gradient of the whole function. In
~\eqref{eq:backward-composition}, the second component of the returned pair is the backpropagator associated with $f$, the variable $z^*$ being the awaited backpropagator associated with the input function of $f$ (called $g$ in~\eqref{eq:chain-rule}).  This transformation generalizes to the definition of $\BSymb_n$ in \refsubfig{fb-ground} for $n$-ary maps. 

We encourage the reader to check that $|\Bkw[n]{\funsa(M)}|=O(1)+|\Bkw[n]{M}|$, so if $F$ consists only of function symbols and variables, the evaluation of $\Bkw[n]{F}$ is asymptotically as costly as the evaluation of $F$. However, in more complex cases, the sole transformation~\eqref{eq:backward-composition} is not enough to guarantee efficiency: if $F$ contains sharing, \eg\ a subroutine $g$ called several times, the evaluation of $\Bkw F$ may duplicate uselessly the computation associated with the backpropagator of $g$, and this may result in an exponential blowup. This highlights a key difference between forward and reverse mode: if $F$ is a first order program (\ie, every abstraction $\lambda x^A$ in $F$ is such that $A$ has no arrows), then $\Fwd F$ is also a first order program, whereas $\Bkw F$ is a higher order term. When backpropagation is expressed in an imperative language, as is usually the case, duplication is not a problem because efficiency is automatically achieved by accumulating the tangents (in reverse order) in memory. But for functional languages, subtle techniques have been introduced to avoid this problem, \eg\ closure conversions \cite{PearlmutterSiskind} or memory references and delimited continuations \cite{Purdue}.

In this paper we follow the approach of \cite{POPL2020}, giving a purely functional solution based on linear logic types: backpropagators have type $\tyR^{\perp_n}$, which corresponds to the set of \emph{linear} maps from $\Real$ to $\Real^n$. The efficiency of the transformation is then guaranteed by a \emph{factoring rule} added to the operational semantics, which allows sharing the evaluation of different occurrences of a backpropagator $z^*$ in an expression: 
\begin{equation}
\label{eq:factoring}
z^*M+z^*N \quad\red\quad z^*(M+N).
\end{equation}
This rewriting rule is sound because backpropagators are linear maps, so they commute with sums. In particular, the normal form of a term obtained using \eqref{eq:factoring} is the same one would have obtained, perhaps in more steps, without using it. As mentioned in the Introduction, in our present setting we are concerned only with soundness, not efficiency. For this reason, we adopt the definition
$$\tyR^{\bot_n}:=\tyR\to\tyR^n$$
and do not consider the linear factoring rule \eqref{eq:factoring}. In fact, by the above remark, rule \eqref{eq:factoring} only speeds up computation without introducing any error, so our almost-everywhere soundness result is transparent to its use, and enforcing linearity would only lead to unnecessary complications induced by a more sophisticated type system. We do retain the notation $(-)^\bot$ as a reminder that this is supposed to be a \emph{linear} arrow (\ie, $\tyR^{\bot_n}$ should really be $\tyR\multimap\tyR^n$), but the transformation of \cite{POPL2020} remains well typed with the above ``non-linear'' definition of negation. Indeed, apart from the addition of conditional and fixpoints, the transformation of \reffig{AD} is exactly that of \emph{loc.\ cit.} and it is efficient as long as it is executed according to the operational semantics enriched with \eqref{eq:factoring}, so nothing is lost with respect to our previous work.

So far we have explained AD transformations only in regard to primitive functions, which are the ``elementary blocks'' of straight-line programs mentioned in the Introduction. It is an observation first formalized in \cite{Purdue} that the transformations may be extended to arbitrary programs simply by applying the functoriality principle: $\FSymb$ and $\BSymb$ are defined to commute with the programming constructs of the language, resulting in \refsubfig{fb-ho}. As a result, the abstract syntax tree of an expression is basically preserved by the two transformations,\footnote{This is the case for all constructs except the conditional, where a projection is added to the transformation of the guard. This minor technicality is due to the fact that we consider only the ground type of real numbers and not that of Booleans.} only the types of the variables are lifted so as to accommodate primals and tangents or backpropagators at the ground level. This behavior is often described by saying that AD is implemented via ``operator overloading''. We prefer the term ``functoriality'' because we find it technically more appropriate.

Notice that, in the definition of $\FB{\ite{P}{M}{N}}$, the transformation is applied also to the guard $P$ in order to preserve typability, because $P$ may share free variables with $M$ and $N$. However, the computation of the gradient of $P$ is useless and therefore $\FB P$ is projected to the first component. A possible optimization would be to define $\FB{\ite{P}{M}{N}}:=\ite{\FBSymb'(P)}{\FB M}{\FB N}$ where $\FBSymb'$ is an auxiliary transformation such that $\FBSymb'(x^A)=\pi_1 x^{\FB A}$ and which behaves homomorphically on every other term. We avoid introducing $\FBSymb'$ because it is not crucial for our results and because such an optimization is not so relevant at our level of abstraction (indeed, the evaluation of $\FB P$ is linear in the evaluation of $P$, as proved in \cite{POPL2020}, so asymptotically there is no gain).

\begin{figure*}
\begin{subfigure}{\textwidth}
\begin{align*}
\Fwd[1]{\Relu}&= 
		\lambda x^{\tyR\times\tyR}.\ite
			{\pi_1 x}
			{	\pair{0,0}}
			{x}
&
\Bkw[1]{\Relu}&= 
		\lambda x^{\tyR\times\tyR^\perp}.\ite
			{\pi_1 x}
			{	\pair{0,\lambda a.0}}
			{x}
\end{align*}
\caption{$\FBSymb_1$ of the rectified linear unit $\Relu$, defined in~\eqref{ex:pcfterms}.}\label{fig:FB_relu}
\end{subfigure}

\medskip
\begin{subfigure}{\textwidth}
\begin{align*}
\Fwd[1]{x^{\tyR}-y^{\tyR}}&
=\left(\lambda z_1^{\tyR\times\tyR}z_2^{\tyR\times\tyR}.\pair{\pi_1^2(z_1)-\pi_1^2(z_2), \pi_2^2(z_1)-\pi_2^2(z_2)}
\right)x^{\tyR\times\tyR}y^{\tyR\times\tyR}
\\
\Fwd[2]{x^{\tyR}-y^{\tyR}}&
=\left(\lambda z_1^{\tyR\times\tyR^2}z_2^{\tyR\times\tyR^2}.\pair{\pi_1^3(z_1)-\pi_1^3(z_2), \pi_2^3(z_1)-\pi_2^3(z_2),\pi_3^3(z_1)-\pi_3^3(z_2)}
\right)x^{\tyR\times\tyR^2}y^{\tyR\times\tyR^2}
\\
\Bkw[n]{x^{\tyR}-y^{\tyR}}&
=\left(\lambda z_1^{\tyR\times\tyR^{\perp_n}}\!\!z_2^{\tyR\times\tyR^{\perp_n}}\!\!\!\!.\pair{\pi_1^2(z_1)\!-\!\pi_1^2(z_2), \lambda a^\tyR\!.(\pi_2^2(z_1)1\!\cdot\! a+\pi_2^2(z_2)(-1)\!\cdot\! a)}
\right)x^{\tyR\times\tyR^{\perp_n}}\!\!y^{\tyR\times\tyR^{\perp_n}}
\end{align*}
\caption{$\FBSymb_n$ of the subtraction $x-y$, with $n=1,2$, where we suppose $\partial_1(x-y):=1$ and  $\partial_2(x-y):=-1$.}\label{fig:FB_substraction}
\end{subfigure}
%
\caption{Some examples of the $\FBSymb$ transformations. Notice that we take the liberty of using the same name for the ground variables $x^\tyR$ and $y^\tyR$ and their images $x^{\FB\tyR}$ and $y^{\FB{\tyR}}$ under $\FBSymb$. 
}
\label{fig:FB_examples}
\end{figure*}
Some examples of the two modes of $\FBSymb$ are given in \reffig{FB_examples}. In the case of subtraction (\reffig{FB_substraction}), notice how the size of the term resulting from the forward transformation $\FSymb_n$ increases with the gradient dimension $n$, while it is constant in $\BSymb_n$, as expected from the above discussion.

Given a term $\Gamma\vdash M:A$ with $\Gamma=x_1^{A_1},\dots, x_n^{A_n}$, one can check that $\FB\Gamma\vdash \FB M:\FB A$, 
where $\FB\Gamma=x_1^{\FB{A_1}},\dots, x_n^{\FB{A_n}}$. In particular, if $M$ is a program, then
\begin{align}
x_1^{\tyR\times\tyR^n},\dots,x_n^{\tyR\times\tyR^n} &\vdash \Fwd[n]{M}:\tyR\times\tyR^n&
x_1^{\tyR\times\tyR^{\perp_n}},\dots,{x_n^{\tyR\times\tyR^{\perp_n}}} &\vdash \Bkw[n]{M}:\tyR\times\tyR^{\perp_n}
\end{align}
If, furthermore, $M$ is simple, then the computational behavior of the transformations $\FSymb$ and $\BSymb$ is given by the following result, which was proved in \cite{HuotStatonVakar,POPL2020,BartheCLG},\footnote{In a personal communication, Mitchell Wand showed us that soundness of reverse mode AD may also be proved by means of ``open'' logical relations of the kind discussed in \cite{BartheCLG} and used here for the unsoundness bound (\refsect{Unsoundness}).} and in which $\iota_i^n$ are the injections of $\tyR$ into $\tyR^n$ as defined in Equation~\eqref{eq:inj}.
\begin{proposition}[soundness of AD for simple terms]
	\label{prop:SimpleSound}
	Let $\Gamma=x_1^\tyR,\ldots,x_n^\tyR$ and let $\Gamma\vdash t:\tyR$ be a simple program. Then, for all $\seq r=(r_1,\ldots,r_n)\in\diffdom{t}$, we have
	\begin{align*}
		\Fwd[n]{t}\isub{\pair{r_1,\iota_1^n1}\!}{x_1}\dots\isub{\pair{r_n,\iota_n^n1}\!}{x_n} &\reds \pair{\sem{t}\!(\seq r),\nabla\!\sem{t}\!(\seq r)} \\
		\Bkw[n]{t}
	\isub{\pair{r_1,\iota_1^n}\!}{x_1}\dots\isub{\pair{r_{n},\iota_n^n}\!}{x_n}\bigr) &\reds \pair{\sem{t}\!(\seq r),u}
	\end{align*}
	such that $u1\reds \nabla\!\sem{t}\!(\seq r)$.
\end{proposition}

Looking at \refprop{SimpleSound}, if $M$ is an arbitrary program of arity $n$ and coarity $1$, it is reasonable to believe that the following programs compute $\nabla\!\sem M$ in $\seq r=(r_1,\ldots,r_n)\in\Real^n$ whenever this is defined:
\begin{align}
	\Fgrad[n]{M}(\seq r)
	&:= 
	\pi_2^2\Fwd[n]{M}\isub{\pair{r_1,\iota_1^n1}\!}{x_1}\dots\isub{\pair{r_n,\iota_n^n1}\!}{x_n}
	\label{eq:grad_f},
	\\
	\Bgrad[n]{M}(\seq r)
	&:=
	\bigl(\pi_2^2\Bkw[n]{M}
	\isub{\pair{r_1,\iota_1^n}\!}{x_1}\dots\isub{\pair{r_{n},\iota_n^n}\!}{x_n}\bigr)1.
	\label{eq:grad_b}
\end{align}
In the sequel, we will omit the index $n$ when inessential or clear from the context. We will also write $\FBgrad M$ for either one the above terms.\footnote{Notice that the above definition of $\FBgradSymb$ is slightly different from the informal one used in the Introduction: it applies to terms with free ground variables, whereas in the Introduction we abusively applied $\FBgradSymb$ to closed terms.}

Referring to \reffig{FB_examples}, it is immediate to check that, regardless of the mode, 
$\FBgrad[1]{\Relu\, z^\tyR}(r)\reds 1$ if $r> 0$ and $\FBgrad[1]{\Relu\, z^\tyR}(r)\reds 0$ if $r\leq 0$.  
This is the expected result except for $r=0$, where the map $\sem[z^\tyR]{\Relu\, z}$ is not differentiable. As discussed in the Introduction, our soundness result concerns only the domain of differentiability $\diffdom M$ of a map $\sem[] M$ represented by a program and nothing is stated about the value  of $\FBgrad{M}$ outside $\diffdom M$. 
This situation is quite common in AD frameworks, where it may even be desirable to control the behavior of the ``non-existent derivative'' on singularities. For example, the following term $\Crelu$ implements the rectified linear unit (\ie, $\sem[z^\tyR]{\Crelu\,z}=\sem[z^\tyR]{\Relu\,z}$) so that $\FBgradSymb$ returns some arbitrarily chosen $q\in\Real$ on $0$:
\begin{equation}\label{ex:crelu}
\Crelu[q]:=  \lambda x^\tyR.\ite x{\ite{-x}{q\cdot x}0}x,
\qquad
\FBgrad[1]{\Crelu[q]\, z^\tyR}(r)\reds
	\begin{cases}
	1&\text{if $r>0$,}\\
	q&\text{if $r=0$,}\\
	0&\text{if $r<0$.}
	\end{cases}
\end{equation}
We do not consider these computations as errors because $\nabla{\sem[z^\tyR]{\Relu\, z}}$ is undefined at $0$.

We already discussed $\Sid$ (see \eqref{ex:pcfterms}) as a first example of a mismatch between AD and the gradient in the domain of differentiability of a map. Let us consider here a refined example:
\begin{align}
\EqProj&:= \lambda x^\tyR.\lambda y^\tyR.
	\ite{x-y}
	{\ite{y-x}xy}
	y.\label{ex:eqproj}
\end{align}
This term is extensionally equivalent to the binary projection $\lambda x^\tyR.\lambda y^\tyR.y$ and therefore its gradient should be $\pair{0,1}$ on the whole domain $\Real^2$. By contrast, the reader may check that:
\begin{align*}
\Fgrad[2]{\EqProj \, x_1^\tyR x_2^\tyR}(r_1,r_2) &= \pi_2^2\left(\Fwd[2]{\EqProj} \pair{r_1,\iota_1^21}\pair{r_2,\iota_2^21}\right)\\
&\reds \pi_2^2\left(\ite{r_1-r_2}{\ite{r_2-r_1}{\pair{r_1,\pair{1,0}}}{\pair{r_2,\pair{0,1}}}}{\pair{r_2,\pair{0,1}}}\right)\\
&\reds
\begin{cases}
	\pair{0,1}&\text{if $r_1\neq r_2$},\\
	\pair{1,0}&\text{if $r_1= r_2$}.
\end{cases}
\end{align*}
The diagonal of $\Real^2$ gives an  uncountable set of errors (a similar computation yields the same result also for the reverse mode). However, this set is negligible, \ie, of Lebesgue measure zero, in accordance with the claim (ii') stated in the Introduction.

Let us add a last comment on example \eqref{ex:eqproj}. Consider the unary program $\EqProj \, x_1^\tyR x_1^\tyR$, which is extensionally equivalent to the identity. One can check that  $\Fgrad[1]{\EqProj \, x_1 x_1}(r)\reds 1$ for every $r\in\Real$, so there is no error at all in this case. This is in sharp contrast with approaches based on partial conditionals, such as \cite{AbadiP20}, in which conditionals diverge when the guard evaluates to $0$: under such semantics, $\EqProj \, x_1 x_1$ diverges everywhere.

\medskip
\todod{Added first sentence.}
Different reduction strategies change the convergence and differentiability domain of a program, so a priori the soundness of AD depends on the strategy. In the above examples, we considered the maximal reduction strategy $\beta$. If we wish to specialize to a more restrictive reduction strategy $\mathcal S$, we should prove that $\FBgrad M$ evaluates, in accordance with $\mathcal S$, to the gradient of $\sem{M}^{\mathcal S}$ at almost every point where this is defined. In fact, if we succeed in proving this with respect to $\beta$, then it follows also for any other  ``reasonable''  strategy $\mathcal S$:
\begin{proposition}\label{prop:reduction_independence}
	Let $\Gamma\vdash M:\tyR$ be a program. If $\sem{\FBgrad M}\vert_{\diffdom{M}}\sim\nabla(\sem{M})$, then for any reduction strategy $\mathcal S$ such that $\diffdom[\mathcal S]{M}\subseteq\dom[\mathcal S]{(\FBgrad M)}$ we also have $\sem{\FBgrad M}^{\mathcal S}\vert_{\diffdom[\mathcal S]{M}}\sim\nabla(\sem{M}^{\mathcal S})$.
\end{proposition}
\begin{proof}
	By hypothesis we have that $\diffdom{M}=A\cup Z$ with $Z$ of measure zero and $\sem{\FBgrad M}\vert_A=\nabla(\sem{M})\vert_A$. Let $\mathcal S$ be a reduction strategy satisfying the hypothesis and let $B := A\cap \diffdom[\mathcal S]{M}$. We need to prove that $\sem{\FBgrad M}^{\mathcal S}\vert_B=\nabla(\sem{M}^{\mathcal S})\vert_B$ and that $\diffdom[\mathcal S]{M}\setminus B$ is negligible. 


	Let us start with this latter point. Notice that, by \refprop{strategies},
	$\diffdom[\mathcal S]{M}\setminus B = \diffdom[\mathcal S]{M}\setminus (A\cap \diffdom[\mathcal S]{M}) = \diffdom[\mathcal S]{M}\setminus A\subseteq \diffdom{M}\setminus A$,
	and the latter set is of measure zero by hypothesis.

	Let us now take $\seq r\in B$. We have:
	\begin{align*}
		\sem{\FBgrad M}^{\mathcal S}\!(\seq r)
		&=\sem{\FBgrad M}\!(\seq r) &\text{by } B\subseteq \diffdom[\mathcal S]{M}\subseteq\dom[\mathcal S]{(\FBgrad M)}\text{ and \refprop{strategies}}\\
		&=\nabla(\sem{M})(\seq r) &\text{because }B\subseteq A\\
		&=\nabla(\sem{M}^{\mathcal S})(\seq r) &\text{by $ B\subseteq \diffdom[\mathcal S]{M}$ and \refprop{strategies}}.
	\end{align*}
\end{proof}

The condition $\diffdom[\mathcal S]{M}\subseteq\dom[\mathcal S]{(\FBgrad M)}$ is reasonable, since $\diffdom[\mathcal S]{M}\subseteq\dom[\mathcal S]{M}$ by definition and it is very likely that $\dom[\mathcal S]{M}\subseteq \dom[\mathcal S]{\FBgrad M}$, because the convergence of $\FBgrad M$ coincides with that of $\FB M$ and the latter essentially behaves like $M$, as we will prove in \refsect{stable_points}. Notice that, when $t$ is simple, $\dom[\mathcal S]{t}=\dom[\mathcal S]{\FBgrad t}$ is trivially true because of strong normalization (\refprop{sn}), so \refprop{SimpleSound} in fact holds for any reduction strategy. Common strategies such as call-by-value, call-by-name and call-by-need are easily seen to enjoy the condition of \refprop{reduction_independence}. See Remark~\ref{rk:cbv} below for a proof sketch in the case of call-by-value.

\subsection{Primitive Functions and Complete Quasicontinuity}
\label{sect:cqc}
The only assumption we made so far about the function symbols of \PCF\ is that for every $\phi$ of arity $k$ and every $1\leq i\leq k$, there is another $k$-ary function symbol $\partial_i\phi$ corresponding to the partial derivative of $\sem[]\phi$ with respect to its $i$-th argument (\refsubfig{fb-ground}). In order to prove one of our main results (\refth{Unsound}), we will need to make some further topological and measure-theoretic assumptions, which we proceed to spell out.

In what follows, we always consider $\Real^n$ with its standard topology and the Lebesgue measure. We denote by $\interior{X}$ the interior of a set $X$ (the largest open set contained in $X$) and by $\border X$ its \emph{border}, defined as $\border X:=X\setminus\interior X$. Equivalently, $\border X=\partial X\cap X$ where $\partial X$ is the boundary of $X$ (the closure of $X$ minus $\interior{X}$). In case $X$ is closed, $\border X=\partial X$.

We recall that a \emph{clone}~\cite{Clones} on a set $A$ is a collection $\clone P$ of functions $A^n\pto A$ (for varying $n$) 
which is closed under composition\footnote{We mean that if $f:A^k\pto A$ and $g_1,\ldots,g_k:A^n\pto A$ are in $\clone P$, then so is the function $\seq a\mapsto f(g_1(\seq a),\ldots,g_k(\seq a))$.} and contains all projections (in particular, the identity on $A$). 
Notice that clones are stable under arbitrary intersections, hence every set $\clone F$ of functions $A^n\pto A$ (for possibly varying $n$) generates a clone $\pair{\clone F}$, the smallest clone containing $\clone F$.

\begin{definition}[admissible primitive functions]
\label{def:admissibility}
	We say that a clone $\clone P$ on $\Real$ is \emph{admissible} if $f\in\clone P$ implies:
	\begin{enumerate}
		\item $f$ is continuous on its domain;
		\item if $f:\Real^n\pto\Real$ is not identically zero, then $f^{-1}(0)$ is of Lebesgue measure zero in $\Real^n$.
	\end{enumerate}
	
	Fix a set $\clone F$ of function symbols together with their semantics. We abusively denote by $\clone F$ also the set of all $\sem[]\phi$ with $\phi$ ranging over the chosen function symbols. We say that $\clone F$ forms an \emph{admissible set of primitive functions} if $\pair{\clone F}$ is admissible.
\end{definition}

It is well known~\cite{Mityagin} that an example of admissible clone is provided by the collection of all real functions which are defined and analytic on some open set $U\subseteq\Real^n$, for varying $U$ and $n$. Notice that a subclone of an admissible clone is admissible. Therefore, a simple way of ensuring that the primitive functions of \PCF\ are admissible is to ask that they are all analytic where they are defined. This is of course true of our ``mandatory'' primitive functions (constants, addition and multiplication), as well as all functions usually taken as primitive, such as division, square root, exponential, logarithm, the trigonometric functions and their inverses, Gaussian functions, many sigmoid functions (\eg\ the error function), etc. Other desirable functions which are not analytic (the step function, the floor function, the rectified linear unit\ldots) are usually programmable in \PCF\ from these primitive functions.

The definitions that follow are parametric in a choice of admissible clone $\clone B$, so we should speak of $\clone B$-quasiopen set, $\clone B$-quasicontinuity, etc. However, for simplicity, we will omit the parameter $\clone B$, implicitly fixing once and for all an admissible set $\clone F$ of primitive functions and letting $\clone B:=\pair{\clone F}$. Functions in $\clone B$ will be called \emph{basic}.	

\begin{definition}[quasiopen set]
	We define the class of \emph{quasiopen} sets of $\Real^n$ to be the smallest class of subsets of $\Real^n$ which:
	\begin{enumerate}
		\item contains every open set;
		\item contains the zero set of every basic function $\Real^n\pto\Real$;
		\item is closed under countable unions and binary intersections.
	\end{enumerate}
	Inductively, the quasiopen sets of $\Real^n$ may be defined as follows:
	$$Q,Q' ::= U \gsep h^{-1}(0) \gsep \bigcup_{i\in I}Q_i \gsep Q\cap Q',$$
	where $U$ ranges over the open sets of $\Real^n$, $h:\Real^n\pto\Real$ ranges over basic functions (which may further be supposed to be not identically zero) and $I$ is countable.
\end{definition}

\begin{definition}[quasivariety]
\label{def:quasivariety}
	A set $Z\subseteq\Real^n$ is called a \emph{quasivariety} if there exists a family $\{h_i\}_{i\in I}$ of not identically zero basic functions $h_i:\Real^n\to\Real$ with $I$ countable and such that
	$$Z\subseteq\bigcup_{i\in I}h^{-1}(\{0\}).$$
	In other words, a quasivariety is an arbitrary subset of a countable union of zero sets of basic functions (except the identically zero function).
\end{definition}

The following result, which says that quasivarieties form a class of ``negligible sets'', will be frequently used in the sequel, without explicit mention:
\begin{lemma}
	\label{lemma:quasivar}
	Quasivarieties enjoy the following properties:
	\begin{enumerate}
		\item \textbf{\emph{measure zero:}} if $Z\subseteq\Real^n$ is a quasivariety, then it a has Lebesgue measure zero in $\Real^n$;
		\item \textbf{\emph{stability under countable unions:}} if $\{Z_i\}_{i\in I}$ is a countable family of quasivarieties, then $\bigcup_{i\in I} Z_i$ is a quasivariety;
		\item \textbf{\emph{stability under subsets:}} if $Z$ is a quasivariety and $Z'\subseteq Z$, then $Z'$ is a quasivariety.
	\end{enumerate}
\end{lemma}
\begin{proof}
	Immediate from the definition.
\end{proof}

\begin{lemma}
	\label{lemma:qopen}
	Let $Q\subseteq\Real^n$ be quasiopen. Then:
	\begin{enumerate}
		\item\label{qopen:decomp} there exists an open set $U$ and a quasivariety $Z$ such that $Q=U\cup Z$;
		\item\label{qopen:border} $\border Q$ is a quasivariety. 
		Hence, in the above one may always take $U=\interior{Q}$ and $Z=\border{Q}$.
	\end{enumerate}
\end{lemma}

The set of non-positive numbers $\Real_{\leq 0}$ is an example of quasiopen subset of $\Real$: to see why, simply notice that $\Real_{\leq 0}=\Real_{<0}\cup\{0\}$, the first being open and the second being the zero set of the identity, which is always a basic function. In a sense, the key property of the class of quasiopen sets is that it includes both $\Real_{\leq 0}$ and $\Real_{>0}$, a fact which will be used crucially in \reflemma{if}.

On the other hand, thick Cantor sets provide examples of non-quasiopen subsets of $\Real$: such a set $K$ is closed, of positive measure and has empty interior, so $K=\border{K}$, which would contradict \reflemma{qopen}.\ref{qopen:border} if $K$ were quasiopen.

In what follows, if $f:A\pto B$ and $g:C\pto D$ are partial functions between sets, we write $f\times g$ for the function of type $A\times C\pto B\times D$ such that $(f\times g)(a,c)=(f(a),g(c))$ whenever $f(a)$ and $g(c)$ are defined, and is undefined otherwise. 
\begin{definition}[(complete) quasicontinuity]
	\label{def:cqc}
	A function $f:\Real^n\pto\Real^m$ is \emph{quasicontinuous}\footnote{The terminology ``quasicontinuous'' already has a standard meaning, unrelated to the one defined here.  On the other hand, ``quasiopen'' and ``completely quasicontinuous'', which are the fundamental notions used in this work, do not seem to have been used in the literature.} if, for every quasiopen set $Q\subseteq\Real^m$, $f^{-1}(Q)$ is quasiopen. We say that $f$ is \emph{completely quasicontinuous} (cqc) if $\id_{\Real^k}\times f$ is quasicontinuous, for all $k\in\Nat$.
\end{definition}

Complete quasicontinuity is needed in order to have \refsublemma{qcont}{pairing} below. It is worth pointing out, however, that we have not been able to find an example of a quasicontinuous function which is not completely quasicontinuous. So, while we conjecture that complete quasicontinuity is strictly stronger than quasicontinuity, the two notions might coincide in reality.

\begin{lemma}
	\label{lemma:qcont}
	We have the following properties:
	\begin{enumerate}
		\item\label{qcont:decomp} a function $f:\Real^n\to\Real^m$ is quasicontinuous iff for every $Q$ which is either open or the zero set of a basic function, $f^{-1}(Q)$ is quasiopen.
		\item\label{qcont:comp} Identities are cqc and cqc functions are stable under composition.
		\item\label{qcont:basic} Basic functions are cqc. In particular, projections are cqc.
		\item\label{qcont:pairing} If $f:\Real^k\pto\Real^m$ and $g:\Real^k\pto\Real^n$ are cqc, then the function $\pair{f,g}:\Real^k\to\Real^{m+n}$ defined by $\pair{f,g}\!(z):=(f(z),g(z))$ if $z\in\dom f\cap\dom g$ and undefined otherwise, is also cqc.
	\end{enumerate}
\end{lemma}
Contrarily to what the name might suggest, a continuous function is \emph{not} in general quasicontinuous. In fact, it is well known that any closed subset of $\Real$ may be the zero set of a map which is smooth everywhere (in particular, continuous). So let $\funsa:\Real\to\Real$ be a smooth function whose zero set is a thick Cantor set $K$, which, as observed above, is not quasiopen. The set $\{0\}$ is quasiopen (it is the zero set of the identity, which is a basic map), and yet $\funsa^{-1}(\{0\})=K$, so $\funsa$ is not quasicontinuous.

\section{Soundness of AD}
\label{sect:Soundness}

We want to prove that $\FBgrad M$ computes the gradient of a program $M$ almost everywhere in $\diffdom M$ (\refth{Unsound}). The proof splits in two parts: \refth{Sound} states that $\FBgrad M$ is sound for the set $\Stab M$ of stable points of $\diffdom M$ (\refdef{Stab}) and \refsect{Unsoundness} shows that $\Stab M$ is actually almost all of $\diffdom M\subseteq\dom M$, in the sense that $\dom M\setminus\Stab M$ is of measure zero.

Intuitively, a point $\seq r\in\dom M$ is stable whenever there exists a simple term $t$ that ``traces'' the evaluation of $M$ over an open ball $\Ball\varepsilon{\seq r}$ of $\seq r$. Such a $t$ allows us to lift the soundness theorem for simple terms (\refprop{SimpleSound}) to $\FBgrad M$. \todom{changed sentence}This reasoning is based on the extrusion lemma (\reflemma{extrusion}), which needs a notion of ``trace'' not only at level of terms (\refdef{TrApx}), but also at the level of  the reduction sequences  (\refdef{ApxRed_sequences}).

\subsection{Traces}
\label{subsection:traces}
\begin{figure*}
 	\begin{subfigure}[b]{\textwidth}
 		$$
 		\infer{\tyR\apx\tyR}{}
		\quad
 		\infer{A'\to B'\apx A\to B}{A'\apx A &B'\apx B}
		\quad
 		\infer{A_1'\times\cdots\times A'_k\apx A_1\times\cdots\times A_k}{A_i'\apx A_i,\; \forall\,1\leq i\leq k}
		 \quad
		 \infer{A_1\times\cdots\times A_n\apx A}{A_i\apx A,\; \forall\,1\leq i\leq n}
 		$$
 		\caption{The pre-trace relation on types.}
 	\end{subfigure}
 
 	\bigskip
 	\begin{subfigure}[b]{\textwidth}
 		$$
		\infer{\Xi\vdash\ctxthole\apx\ctxthole}{}
		\qquad
		\qquad		
 		\infer[i\in\{1,\dots,n\}]{\Xi, p^{A_1\times\cdots\times A_n}\apx x^A\vdash \pi_i^n p\apx x}{A_1\times\cdots\times A_n\apx A}
 		$$
 		\medskip
 		$$
 		\infer{\Xi\vdash \lambda p^{A'}.t\apx \lambda x.M}{\Xi, p^{A'}\apx x^A\vdash t \apx M}
		\qquad
		 \infer{\Xi\vdash t\!\pair{u_1,\ldots,u_n}\apx MN}{\Xi\vdash t\apx M, & \Xi\vdash u_1\apx N &\ldots& \Xi\vdash u_n\apx N}
 		$$
 		\medskip
 		$$
 		\infer{\Xi\vdash\pair{t_1,\dots,t_k}\apx\pair{M_1,\dots,M_k}}{\Xi\vdash t_1\apx M_1,&\dots,&\Xi\vdash t_k\apx M_k}
 		\qquad\quad
 		\infer{\Xi\vdash \pi_i^kt\apx\pi_i^kM}{\Xi\vdash t\apx M}
		\qquad\quad
		 \infer{\Xi\vdash \funsa(t_1,\dots,t_k)\apx\funsa(M_1,\dots,M_k)}{\Xi\vdash t_1\apx M_1 &\ldots& \Xi\vdash t_k\apx M_k}
		$$
		\medskip
 		$$
		 \infer[i\in\{1,2\}]{\Xi\vdash \pr i\!\pair{t_i,t_i} \apx\ite P{M_1}{M_2}}{\Xi\vdash t_i\apx M_i}
 		\qquad\quad
 		\infer[n>0]{\Xi\vdash t\apx \fix f M}{\Xi\vdash t\apx\fix[n]{f}{M}}
 		$$
 		\caption{The pre-trace relation on terms. In the variable rule, if $n=1$, then $\pi_i$ is omitted. Recall the definition of $\fix[n] fM$ in \eqref{eq:fixpoint_approx}.
		}
 		\label{fig:TraceTerms}
 	\end{subfigure}
 	\caption{The pre-trace relation between simple and arbitrary \PCF\ terms.}
 	\label{fig:Trace}
\end{figure*}
 
The \emph{pre-trace relation} is defined in \reffig{Trace}. The judgments used in the definition are of the form $\Xi\vdash t\apx M$, where $t$ is a simple term or simple context, $M$ is an arbitrary term or context of type, say, $\Gamma\vdash M:A$ and 
$\Xi$ is a function mapping any variable $x^A$ of $\Gamma$ to a fresh variable $p^{A'}$ with $A'\apx A$. We usually denote this map as a list $p_1^{A_1'}\apx x_1^{A_1},\ldots,p_n^{A_n'}\apx x_n^{A_n}$, supposing the $p_i$'s and $x_i$'s to be pairwise different. 

For brevity, we omit to specify the types of the terms in the subjects of the judgments in \reffig{TraceTerms}, but we encourage the reader to verify that 
if $p_1\apx x_1,\ldots,p_n\apx x_n\vdash t\apx M$ and $x_1^{C_1},\ldots,x_n^{C_n}\vdash M:A$, then $p_1^{C'_1},\ldots p_n^{C'_n}\vdash t:A'$ with $A'\apx A$ and, for all $1\leq i\leq n$, $C'_i\apx C_i$. In particular, $t$ is a simply-typed \lat{} or context. We write $t\apx M$ when the typing environment $\Xi$ is irrelevant.

\todod{Added paragraph}
\todom{Modified paragraph (and added one more)}
Conditionals and fixpoints are the only additional features of \PCF\ with respect to the simply-typed $\lambda$-calculus, and the purpose of $\apx$ is to ``trace'' them with simply-typed terms themselves. The last two rules of \reffig{TraceTerms} ``slice out'' a conditional with the traces of its two branches and unfold a fixpoint into its finite approximations. In fact, the conditional rule is a bit more convoluted as it uses a dummy projection in order to encode the index of the chosen branch, a crucial information for the extrusion property (\reflemma{extrusion}). For example, $u_1:=\lambda x^\tyR.\pi_1\!\pair{0,0}$ and $u_2:=\lambda x^\tyR.\pi_2\!\pair{x,x}$ are traces corresponding to the ``then'' and ``else'' branch, respectively, of $\Relu$ defined in \eqref{ex:pcfterms}.

The variable rule of \reffig{Trace} also deserves an explanation. Its non-trivial shape, which is due to higher order types, may be understood as follows. Let $T:=\lambda f^{\tyR\to\tyR}.f(f0+1)$ be a term using its (higher order) argument twice and consider the program $T\,\Relu$. Recall that $\Relu$ contains a conditional controlled by its argument, and has two different traces $u_1$ and $u_2$ discussed above. However, the execution of 
$T\,\Relu$, as sketched in \reffig{ex_tracing_red}, explores both branches of $\Relu$, so we allow to trace $T$ with $t:=\lambda p.\pi_2 p((\pi_1 p)0+1)$ and $T\,\Relu$ with $t\!\pair{u_1,u_2}$. That is, we allow different instances of the same variable $f$ to be traced by different components of a tuple variable $p$, because, even if in the original program all occurrences of $f$ are replaced by copies of the same term $\Relu$, different copies of $\Relu$ might be traced by different simple terms, so the occurrences of $f$ must be ``separated'' accordingly.

\begin{lemma}\label{lemma:apx-programs}
	Let $M$ and $t$ be normal forms of type $\FB[n]\tyR$ whose free variables have type belonging to $\{\FB[n]\tyR,\tyR,\tyR^n,\tyR^{\perp_n}\}$. If $t\apx M$ then $t=M$.
\end{lemma}

\begin{lemma}\label{lemma:apx-substitution}
	If $\Xi\vdash t\apx M$, then:
	\begin{enumerate}
	\item\label{lemma:apx-substitution-D}
	$\FB{\Xi}\vdash\FB{t}\apx\FB{M}$, where $\FBSymb$ turns any assignment $p^{A'}\apx x^A$ of $\Xi$ into $p^{\FB{A'}}\apx x^{\FB{A}}$.
	\item\label{lemma:apx-substitution-sub}
	Let $\Xi=\Xi',x^A\apx x^A$. For every closed simple term $u$ of type $A$, we have $\Xi'\vdash t\isub{u}{x}\apx M\isub{u}{x}$.
	\end{enumerate}
\end{lemma}
%

\begin{lemma}
	\label{lemma:apx-substs}
	We have that $\Xi\vdash w\apx M\isub{N}{x}$ is equivalent to
	\begin{itemize}
		\item $w=t\isub{u_1}{x_1}\dots\isub{u_n}{x_n}$, for some $n\in\Nat$ and terms $t,u_1,\ldots,u_n$,
		\item such that $\Xi,p^{A_1\times\dots\times A_n}\apx x^A\vdash t\isub{\pi_1 p}{x_1}\dots\isub{\pi_n p}{x_n}\apx M$, $p$ not free in $t$,
		\item and $\Xi\vdash u_i\apx N$ for all $1\leq i\leq n$.
	\end{itemize}
	In particular, $\Xi,p^{A_1\times\dots\times A_n}\apx x^A\vdash w\apx M$ implies $w=t\isub{\pi_1 p}{x_1}\dots\isub{\pi_n p}{x_n}$ for some $t$ not containing $p$ free. 
\end{lemma}

\begin{figure}
\begin{center}
	\begin{tikzpicture}[every node/.style={circle,inner sep=1pt}]
	\node (M1) at (0,1) {$T\,\Relu$};
	\node (M2) at (6.5,1) {$\Relu(\Relu\,0+1)$};	
	\node (M3) at (10.5,1) {$\ite{\Relu\,0+1}{0}{\Relu\,0+1}$};		
	\node (t1) at (0,0) {$t\!\pair{u_1,u_2}$};
	\node (t2) at (3.25,0) {$\pi_2\!\pair{u_1,u_2}(\pi_1\!\pair{u_1,u_2}0+1)$};	
	\node (t3) at (6.5,0) {$u_2(u_10+1)$};		
	\node (t4) at (10.5,0) {$\pi_2\!\pair{u_10+1,u_10+1}$};			
	\node (M4) at (1,-1) {$\ite{\Relu\,0+1}{0}{\Relu\,0+1}$};
	\node (M5) at (5,-1) {$\ite{1}{0}{\Relu\,0+1}$};	
	\node (M6) at (7.75,-1) {$\Relu\,0+1$};	
	\node (M7) at (10,-1) {$\ite000+1$};
	\node (M8) at (11.6,-1) {$0+1$};		
	\node (M9) at (12.5,-1) {$1$};				
	\node (t5) at (1,-2) {$\pi_2\!\pair{u_10+1,u_10+1}$};
	\node (t6) at (5,-2) {$\pi_2\!\pair{u_10+1,u_10+1}$};	
	\node (t7) at (7.75,-2) {$u_10+1$};		
	\node (t8) at (10,-2) {$\pi_1\!\pair{0,0}+1$};			
	\node (t9) at (11.6,-2) {$0+1$};			
	\node (t10) at (12.5,-2) {$1$};			
	\draw[->] (M1.east) -- (M2.west);
	\draw[->] (M2.east) -- (M3.west);
	\draw[->] (M4.east) -- node[at end, above, font=\footnotesize] {$*$} (M5.west);
	\draw[->] (M5.east) -- (M6.west);		
	\draw[->] (M6.east) -- (M7.west);
	\draw[->] (M7.east) -- (M8.west);
	\draw[->] (M8.east) -- (M9.west);	
	\draw[->] (t1.east) -- (t2.west);
	\draw[->] (t2.east) -- node[at end, above, font=\footnotesize] {$*$}  (t3.west);
	\draw[->] (t3.east) -- (t4.west);
	\path (t5.east) -- (t6.west) node[midway] {$=$};		
	\draw[->] (t6.east) -- (t7.west);
	\draw[->] (t7.east) -- (t8.west);
	\draw[->] (t8.east) -- (t9.west);	
	\draw[->] (t9.east) -- (t10.west);	
	\path (t1) -- (M1) node[midway, rotate=90] {$\apx$};
	\node[rotate=90] () at (6.5,.5) {$\apx$};
	\node[rotate=90] () at (10.5,.5) {$\apx$};
	\node[rotate=90] () at (1,-1.5) {$\apx$};
	\node[rotate=90] () at (5,-1.5) {$\apx$};
	\node[rotate=90] () at (7.75,-1.5) {$\apx$};
	\node[rotate=90] () at (10,-1.5) {$\apx$};
	\node[rotate=90] () at (11.6,-1.5) {$\apx$};
	\node[rotate=90] () at (12.5,-1.5) {$\apx$};
	\end{tikzpicture}
	\end{center}
	\vspace{-1.2cm}
\caption{Tracing the head reduction $T\,\Relu\reds 1$. The term $T$ is $\lambda f.f(f0+1)$, of which $t:=\lambda p.\pi_2 p((\pi_1 p)0+1)$ is a trace; $\Relu$ is given in \eqref{ex:pcfterms}, with traces $u_1:=\lambda x.\pi_1\!\pair{0,0}$, $u_2:=\lambda x.\pi_2\!\pair{x,x}$.}\label{fig:ex_tracing_red}
\end{figure}\todod{Changed $D$ to $T$.}

\begin{definition}[Tracing rewriting steps]
	\label{def:ApxRed_rules}
	Let $\sigma:R\red P$ be a rewriting step of \reffig{reduction}, and $\upsilon$ be a reduction sequence between simple terms. We define $\upsilon\apx\sigma$ depending on $R$.
	\begin{itemize}
		\item If $R=(\lambda x.M)N$: we ask that $\upsilon$ is any reduction of the form
		$$(\lambda p.t\isub{\boldsymbol\pi p}{\seq x})\!\pair{\seq u}\red t\isub{\boldsymbol\pi\!\pair{\seq{u}}\!}{\seq x}\reds t\isub{\seq u}{\seq x}$$
		where $p\apx x\vdash t\isub{\boldsymbol\pi p}{\seq x}\apx M$, with $\isub{\boldsymbol\pi p}{\seq x}$ denoting the substitutions $\isub{\pi_1 p}{x_1}\dots\isub{\pi_n p}{x_n}$ for some $n\in\Nat$, and where $\seq u$ is a sequence $u_1,\ldots,u_n$ of simple terms such that $u_i\apx N$ for all $i$, and in which the redexes $\pr i\!\pair{\seq u}\red u_i$ are reduced in any order.
		
		\item If $R=\pi_i\!\pair{M_1,\dots, M_k}$: we ask that $\upsilon$ is any reduction of the form 
		$$\pi_i\!\pair{t_1,\dots, t_k}\red t_i$$ 
		for $t_1\apx M_1$,\dots, $t_k\apx M_k$.		
		
		\item If $R=\funsa(\seq r)$: we ask that $\upsilon=\sigma$.
		
		\item If $R=\ite{r}{M_1}{M_2}$: we ask that $\upsilon$ is any reduction of the form
		$$\pr i\!\pair{t,t}\red t$$
		with $t\apx M_i$ and $i=1$ if $r\leq 0$, otherwise $i=2$.
		
		\item If $R=\fix f N$: we ask that $\upsilon\apx\sigma'$ where $\sigma'$ is any reduction of the form $$(\lambda f.M)(\lambda x.(\fix[n]{f}{M})x)\red M\isub{\lambda x.(\fix[n]{f}{M})x}{f}$$ for some $n\in\Nat$, as defined in the first case.
	\end{itemize}
\end{definition}

\begin{lemma}[pullback]\label{lemma:apx_expansion}
Let $\sigma:R\red P$ be a rewriting step. For any $w\apx P$, there exist $t\apx R$ and $\xi:t\reds w$ such that $\xi\apx\sigma$. 
\end{lemma}
%


We now extend \refdef{ApxRed_rules} to one-step head reductions (as defined in \refsect{Terms}).
\begin{definition}[Tracing head reduction steps]
	\label{def:ApxRed_steps}
	Let $\sigma=(\hCtxt,R,P)$ be a reduction step with $\hCtxt$ a head context, and let $\sigma_0$ denote the rewriting step $R\to P$. If $\xi:t\reds u$ is a reduction sequence on simple terms, we write $\xi\apx\sigma$ whenever one of the following holds:
	\begin{itemize}
		\item there exists a simple context $\mathsf h\apx\hCtxt$ such that $\xi=\ctxt[\mathsf h]{\upsilon}$ with $\upsilon\apx\sigma_0$ in the sense of \refdef{ApxRed_rules};
		\item the hole of $\hCtxt$ is in the guard of a conditional, $\xi$ is the empty sequence and $t=u\apx\hctxt R$.
	\end{itemize}
\end{definition}	

Finally, we extend the relation $\apx$ to reduction sequences by reflexive-transitive closure.
\begin{definition}[Tracing head reduction sequences]
\label{def:ApxRed_sequences}
	Let $\rho$ be a head reduction sequence starting from a term $M$ and let $\xi$ be a reduction sequence starting from a simple term $t$. We write $\xi\apx\rho$ if either $\rho$ and $\xi$ are empty and $t\apx M$, or if $\rho=\sigma_1\cdots\sigma_n$ is of length $n>0$ and $\xi=\xi_1\cdots\xi_n$ such that $\xi_i\apx\sigma_i$ for every $1\leq i\leq n$, according to \refdef{ApxRed_steps}. 
\end{definition}

\todom{added this paragraph}
\reffig{ex_tracing_red} gives an example of tracing the head reduction of the term $T\,\Relu$ discussed above. Notice that the reduct of $t\!\pair{u_1,u_2}$ is an intermediate term not corresponding to any trace. Notice also that all the reductions in the guard of a conditional (such as the first steps in the third line of \reffig{ex_tracing_red}) share the same traces. In general, $\xi\apx\rho$ implies that $\rho$ is a head reduction but not necessary $\xi$, because the first case of \refdef{ApxRed_rules} requires $\xi$ to reduce projection redexes not in the ``head position'' of a simple term. In fact, the reduction of $t\!\pair{u_1,u_2}$ in \reffig{ex_tracing_red} is not head.

%

\begin{definition}[trace relation]
	\label{def:TrApx}
	Let $M$ be a term and $t$ a simple term. We say that $t$ \emph{traces} $M$, in symbols $t\trapx M$, whenever there exists a normalizing reduction $\xi$ starting from $t$ and a normalizing head reduction $\rho$ starting from $M$ such that $\xi\apx\rho$.
\end{definition}
As an example, let us consider the term $\Sid$ of~\eqref{ex:pcfterms}. Let us define the simple terms:
\begin{align}
\label{ex:simple_silly}
 t_{1}&:= \lambda x.\pi_1\!\pair{\pi_1\!\pair{0,0},\pi_1\!\pair{0,0}},
 &
 t_{2}&:= \lambda x.\pi_1\!\pair{\pi_2\!\pair{x,x},\pi_2\!\pair{x,x}},
 &
 t_{3}&:= \lambda x.\pi_2\!\pair{x,x}.
\end{align}
Notice that, for any $i\in\{1, 2,3\}$, we have $t_i\apx \Sid$ (see \reffig{Trace}), therefore also $t_i r\apx \Sid\, r$ for any real number $r$. By contrast, we have that $t_1 r\trapx \Sid\, r$ iff $r= 0$, $t_2 r\trapx \Sid\, r$ iff $r< 0$ and, symmetrically,  $t_3 r\trapx \Sid\, r$ iff $r>0$. This highlights a sharp difference between the relations $\apx$ and $\trapx$: the former is static whereas the latter traces the execution of a term. Indeed, the projections of the $t_i$'s reflect the different choices in the conditionals of $\Sid$. 

\subsection{AD is Sound on Stable Points}\label{sect:stable_points}

Consider a \PCF\ program $x_1^\tyR,\dots,x_n^\tyR\vdash M:\tyR$. One can easily check that for every $\seq r\in\Real^n$, $M\isub{\seq r}{\seq x}$ is normalizing if and only if there exists a simple term $t$ tracing $M\isub{\seq r}{\seq x}$.  However, this term $t$ usually depends on the chosen $\seq r$. In the following definition of stable points, we consider a situation where $t$ can be ``uniformly'' chosen in an open ball around $\seq r$. 

\begin{definition}
	\label{def:Stab}
	We define the set of \emph{stable points} of a program $x_1^{\tyR},\dots,x_n^\tyR\vdash M:\tyR$ as follows:
	$$
	\Stab M :=\left\{ 
	\begin{array}{l@{\;}l}
	\seq r\in\Real^n \st&\exists\,\varepsilon>0, \exists\, x_1^{\tyR},\dots,x_n^\tyR\vdash t:\tyR\text{ s.t. }t\apx M \text{ and}\\
	&
	\forall\,\seq r'\in\Ball\varepsilon{\seq r}\ t\isub{\seq r'}{\seq x}\trapx M\isub{\seq r'}{\seq x}
	\end{array}
	\right\}.
	$$
\end{definition}
Notice that we have restricted the tracing of $M$ to terms $t$ of the same type as $M$, in particular $t$ does not split different occurrences of a free variable $x_i$ of $M$. In fact, these variables are supposed to be replaced with numerals, which are not split by $\apx$. Also observe that, by definition, we have $\Stab{M}\subseteq\dom{M}$, as $t\isub{\seq r}{\seq x}\trapx M\isub{\seq r}{\seq x}$ implies that $M\isub{\seq r}{\seq x}$ has a normal form. Moreover, $\Stab{M}$ is open: it is easy to check that $\Stab{M}=\bigcup_{\seq r\in\Stab{M}}\Ball{\varepsilon_{\seq r}}{\seq r}$, where $\varepsilon_{\seq r}$ is the positive real whose existence is given by the very definition of stability of $\seq r$.

We already argued in the Introduction that $\Stab{\Relu\, x^\tyR} = \Real\setminus\{0\}$. By recalling the above discussion about the simple terms $t_1$, $t_2$ and $t_3$ in \eqref{ex:simple_silly}, we may infer that also $\Stab{\Sid\, x^\tyR}= \Real\setminus\{0\}$, in fact $0$ is the border where one has to swap between $t_1 r$ and either $t_2 r$ or $t_3 r$ in tracing $\Sid\, r$. Similarly, but with more involved simple terms, one can check that $\Stab{\Floor\, x^\tyR}=\Real\setminus\mathbb Z$ with $\Floor$ given in~\eqref{ex:pcfterms}.  
As a last example, let us consider the term $\EqProj$ defined in~\eqref{ex:eqproj} and the simple terms
\begin{align*}
 t_{1}&:= \lambda x.\lambda y.\pi_1\pair{\pi_1\pair{x,x},\pi_1\pair{x,x}},
 &
 t_{2}&:= \lambda x.\lambda y.\pi_1\pair{\pi_2\pair{y,y},\pi_2\pair{y,y}},
 &
 t_3&:=\lambda x.\lambda y.\pi_2\pair{y,y}.
\end{align*}
These are all such that $t_i\apx\EqProj$. However, $t_1rq\trapx\EqProj\, r q$ iff $r=q$, $t_2rq\trapx\EqProj\, r q$ iff $r<q$ and $t_3 rq \trapx\EqProj\, r q$ iff $r>q$. So the diagonal splits the plane $\Real^2$ in two open sets where either $t_2$ or $t_3$ uniformly traces the execution of $\EqProj$, whereas the execution on the diagonal is traced by $t_1$. But the diagonal contains no open set, therefore $\Stab{\EqProj\, x^\tyR y^\tyR}=\Real^2\setminus\{(r,r)\;;\;r\in\Real\}$. Notice, finally, that $\EqProj\, x^\tyR x^\tyR$ may be traced everywhere by $t_1 x^\tyR x^\tyR$, hence $\Stab{\EqProj\, x^\tyR x^\tyR}=\Real$. 

\begin{figure}
	\begin{minipage}{0.9\textwidth}
	$$
	\infer{\ctxthole\expd\ctxthole}{}
	\qquad
	\infer{x^A\expd x^{\FB A}}{}
	\qquad
	\infer{\lambda x.M\expd \lambda x.M'}{M\expd M'}
	\qquad
	\infer{MN\expd M'N'}{M\expd M' & N\expd N'}
	$$

	$$
	\infer{\pair{M_1,\dots,M_k}\expd\pair{M_1',\dots,M_k'}}{M_1\expd M_1' &\dots,&M_k\expd M_k'}
	\qquad
	\infer{\pi_i^k M\expd\pi_i^k M'}{M\expd M'}
	$$

	$$
	\infer{\funsa(M_1,\dots, M_k) \expd (\lambda z_1^{\FB \tyR}\ldots\lambda z_k^{\FB \tyR}.\pair{\funsa(\pi_1z_1,\dots,\pi_1z_k),t}) M'_1\cdots M'_k}
	{z_1^{\FB \tyR}\dots z_k^{\FB \tyR}\vdash t:\tyR^{(\perp_n)} \text{ simple normal form},&
	M_1\expd M_1',&\ldots, &M_k\expd M_k'
	}
	$$

	$$
	\infer{\ite P M N \expd \ite {\pi_1 P'} {M'} {N'}}{P\expd P'& M\expd M' & N\expd N'}
	\qquad
	\infer{\fix f M\expd\fix f {M'}}{M\expd M'}
	$$
	\end{minipage}
	\caption{The expansion relation. In the rule on 
	the third line, if $k=0$ then the right-hand term in the conclusion is just 
	$\pair{\phi,t}$.
	}
	\label{fig:Expd}
\end{figure}

The tracing relation $\trapx$ is defined over programs. However, in order to prove soundness (\refth{Sound}), we must move from a program $M$ to its transformation $\FB M$, this latter having a more complex type than $M$. The difficulty behind the proof of \refth{Sound} is then to deduce from $t\trapx M$ a link between $\FB t$ and $\FB M$. In order to do that, we define a further relation $\expd$ (\reffig{Expd}) catching an invariant between $M$ and $\FB M$ (as well as between $t$ and $\FB t$) stable under the evaluation of the two terms. Then the extrusion lemma (\reflemma{extrusion} and its iterated version \reflemma{extrusion_iterate}) will use $\expd$ for deducing the needed link between $\FB t$ and $\FB M$ from the hypothesis $t\trapx M$.

\begin{definition}[expansion]
The  relation $\expd$ on terms and contexts of \PCF{} is defined in \reffig{Expd}.
\end{definition}

As mentioned above, expansion is used to establish that $M$ and $\FB M$ ``behave similarly''. Such a link emerges from two of the main properties of $\expd$, namely that $M\expd\FB M$ holds for every $M$ (\reflemma{expd_ad}), and that it is a simulation (if $M\expd M'$, then $M'$ simulates $M$, \reflemma{expd-reduction}). These justify the definition in the case $M=\phi(M_1,\dots,M_k)$: it mimics the definition of $\FB{M}$ in the first component of the product, whereas the second component is chosen as an arbitrary closed simple normal form. The first gives us stability under reduction, in particular in the case of a conditional redex, while the second component cannot be asked to be linked with the partial derivatives of $\phi$, because $\phi(M_1,\dots,M_k)$ eventually reduces to a numeral, which has zero derivative.


\begin{lemma}
	\label{lemma:expd_ad}
	For every program $x_1^{\tyR},\dots,x_n^{\tyR}\vdash M:\tyR$, $\seq r\in\Real^n$ and sequence $\seq u=u_1,\ldots,u_n$ of simple closed normal forms of suitable type, we have $M\isub{\seq r}{\seq x}\expd \FB M\isub{\pair{\seq r,\seq u}\!}{\seq x}$, where by $\{\pair{\seq r,\seq u}\!/\seq x\}$ we mean $\{\pair{r_1,u_1}/x_1\}\cdots\{\pair{r_n,u_n}/x_n\}$.
	\end{lemma}

\begin{lemma}\label{lemma:expd-reduction}
Let $M\expd M'$ and $M\red N$, then there exists $N'$ such that $M'\reds N'$ and $N\expd N'$. 
\end{lemma}
\begin{proof}[Proof Sketch]
Let $(\Ctxt,R,P)$ be the reduction step $M\red N$. The proof is an induction on $\Ctxt$. The only non-trivial part is the base of the induction, \ie\ $\Ctxt=\ctxthole$, in which the reasoning splits following \reffig{reduction}: the case of a $\beta$-reduction is a consequence of a substitution lemma on $\expd$. 
If $R=\phi(\seq r)$, then $M'\reds \pair{\phi(\pi_1\seq L'), t\isub{\seq L'}{\seq z}}$ for some $\seq L'$ of the same length as $\seq r$ such that $r_i\expd L'_i$ for all $i$. Notice that $r_i\expd L'_i$ implies that $\pi_1L'_i\red r_i$ as well as $L'_i$ is a simple closed normal form, so in particular $t\isub{\seq L'}{\seq z}$ is normalizable by \refprop{sn}. We therefore have:  $\pair{\phi(\pi_1\seq L'), t\isub{\seq L'}{\seq z}}\reds\pair{\sem[]\phi\!(\seq r),t'}$ with $t'$ a simple closed normal form, and we conclude by taking $N':=\pair{\sem[]\phi\!(\seq r),t'}$.  

The cases of the other redexes (branching, products and fixpoints) are immediate. 
\end{proof}

\begin{remark}\label{rk:cbv}
A variant of Lemma~\ref{lemma:expd-reduction} can be used to prove that $\dom[\mathcal S] M \subseteq \dom[\mathcal S]{\FBgrad M}$, for some reduction strategy $\mathcal S$ (see discussion after Proposition~\ref{prop:reduction_independence}). For example, if $\mathcal S$ is the call-by-value strategy, then one can prove the statement of Lemma~\ref{lemma:expd-reduction} by replacing $\red$ and $\reds$ with $\red[\mathcal S]$ and $\reds[\mathcal S]$ (in the case of a $R=\phi(\seq r)$ redex, one should notice that the terms $\seq L'$ are all values, and one should replace ``normalizable'' with ``$\mathcal S$-normalizable'' and ``simple closed normal form with ``simple closed $\mathcal S$-nf''). Then, consider a program $M$ of arity $1$ and suppose $\seq r\in \dom[\mathcal S] M$, i.e.~ $M\isub{\seq r}{\seq x}\reds[\mathcal S]q$, for $q$ a numeral. 
We have by Lemma~\ref{lemma:expd_ad} and this variant of Lemma~\ref{lemma:expd-reduction} that $\FB{M}\isub{\pair{\seq r,\seq u}\!}{\seq x}\reds[\mathcal S] N$ with $q\expd N$ and each $u_i$ either the normal form of $\iota_{i}^n1$ or $\iota_{i}^n$, depending whether $\FBSymb$ is $\FSymb$ or $\BSymb$ (recall Equations~\eqref{eq:grad_f} and~\eqref{eq:grad_b}). By inspecting Fig.~\ref{fig:Expd}, we can deduce $N=\pair{q,t}$ for some closed simple $\mathcal S$-nf $t$. This means $\Bgrad M\reds[\mathcal S](\pi_2\pair{q,t})1\red[\mathcal S] t1$, this latter evaluating to a normal form because it is a simple term (Proposition~\ref{prop:sn}). We conclude $\seq r\in\dom[\mathcal S] \Bgrad M$. The case of $\Fgrad M$ is simpler. 
\end{remark}

%

\begin{lemma}[extrusion]
	\label{lemma:extrusion}
	Let $M\expd M'$, $t\expd t'$, $t'\apx M'$. Let $\sigma:M\red M_1$ be a head reduction step and moreover let $\xi:t\reds t_1$ be such that $\xi\apx\sigma$ (so in particular $t\apx M$ and $t_1\apx M_1$). Then there exist $M'_1,t'_1$ such that the following relations hold:
	\begin{center}
	\pgfmathsetmacro{\tx}{0.7}
	\pgfmathsetmacro{\ty}{0.8}
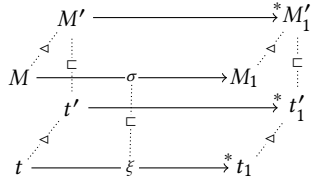
	\begin{tikzpicture}[every node/.style={circle,inner sep=1pt}]
	\node (t) at (0,0) {$t$};
	\node (M) at (0,1.2) {$M$};
	\node (t1) at (3,0) {$t_1$};
	\node[inner sep=0pt] (M1) at (3,1.2) {$M_1$};
	\draw[->] (t.east) -- node (xi) [ fill=white, inner sep=0pt, anchor=center, pos=0.5, font=\scriptsize] {$\xi$} node[at end, above, font=\footnotesize] {$*$} (t1.west);
	\draw[->] (M.east) -- node (sigma) [ fill=white, inner sep=0pt, anchor=center, pos=0.5, font=\scriptsize] {$\sigma$}  (M1.west);
	\draw[densely dotted] (sigma.south) -- node [fill=white, inner sep=0pt, anchor=center, pos=0.5, font=\scriptsize] {$\apx$} (xi.north);
	\node (t') at ($(t)+(\tx,\ty)$) {$t'$};
	\node (M') at ($(M)+(\tx,\ty)$) {$M'$};
	\node (t1') at ($(t1)+(\tx,\ty)$) {$t_1'$};
	\node[inner sep=0pt] (M1') at ($(M1)+(\tx,\ty)$) {$M_1'$};
	\draw[->] (t'.east) --  node[at end, above, font=\footnotesize] {$*$} (t1'.west);
	\draw[->] (M'.east) -- node[at end, above, font=\footnotesize] {$*$}   (M1'.west);
	\draw[densely dotted] (M'.south) -- node [fill=white, inner sep=0pt, anchor=center, pos=0.5, font=\scriptsize] {$\apx$} (t'.north);
	\draw[densely dotted] (M1'.south) -- node [fill=white, inner sep=0pt, anchor=center, pos=0.5, font=\scriptsize] {$\apx$} (t1'.north);
	\draw[densely dotted] (t.north east) -- node [fill=white, inner sep=0pt, anchor=center, pos=0.5, font=\scriptsize] {$\expd$} (t'.south west);
	\draw[densely dotted] (t1.north east) -- node [fill=white, inner sep=0pt, anchor=center, pos=0.5, font=\scriptsize] {$\expd$} (t1'.south west);
	\draw[densely dotted] (M.north east) -- node [fill=white, inner sep=0pt, anchor=center, pos=0.5, font=\scriptsize] {$\expd$} (M'.south west);
	\draw[densely dotted] (M1.north east) -- node [fill=white, inner sep=0pt, anchor=center, pos=0.5, font=\scriptsize] {$\expd$} (M1'.south west);
	\end{tikzpicture}
	\end{center}
%
%
\end{lemma}
\begin{proof}[Proof Sketch]
By \refdef{ApxRed_steps}, $\sigma=\hctxt{\sigma_0}$ for some head context $\hCtxt$ and reduction step $\sigma_0:R\red P$. The proof is by induction on $\hCtxt$. 

The case $\hCtxt=\ctxthole$ splits following \reffig{reduction}. For example, let  $\sigma$ be $M = (\lambda x.L)N\red L\isub{N}{x}=M_1$, so that $\xi$ is the reduction $ t=(\lambda p.w\isub{\boldsymbol\pi p}{\seq x})\pair{\seq u}\reds w\isub{\seq u}{\seq x}=t_1$. Then $M'=(\lambda x.L')N'$ with $L\expd L'$ and $N\expd N'$, and $t'=(\lambda p.\overline w')\pair{\seq u'}$ with $w\isub{\boldsymbol\pi p}{\seq x}\expd\overline w'$ and $\pair{\seq u}\expd \pair{\seq u'}$. Moreover, since $t'\apx M'$, by \reflemma{apx-substs}, we have $\overline w' = w'\isub{\boldsymbol\pi p}{\seq x}$, with $w'\isub{\boldsymbol\pi p}{\seq x}\apx L'$, $p$ not free in $w'$, and $u_i'\apx N'$. Moreover, by induction on $w$, one can infer from $w\isub{\boldsymbol\pi p}{\seq x}\expd\overline w'$ that actually $w\expd w'$. 

Of the induction cases, the only subtle one is $\hCtxt=\ite{\overline\hCtxt}{N_1}{N_2}$. Under this hypothesis, the reduction $\xi$ is empty and $t_1=t=\pi_i\!\pair{u,u}$ with $u\apx N_i$ for some $i\in\{1,2\}$, as well as $M'=\ite{\pi_1\overline{M}'}{N_1'}{N_2'}$ with $\ctxt[\overline\hCtxt]{R}\expd \overline M'$, $N_1\expd N'_1$, $N_2\expd N'_2$ and $t'=\pi_i\!\pair{u',u'}$ with $u\expd u'$ and $u'\apx N_i'$ (notice that the index $i$ of the projection is the same in $t$ and $t'$ because $t\expd t'$). By Lemma~\ref{lemma:expd-reduction}, $\overline M'\reds L$ such that $\ctxt[\overline\hCtxt]{P}\expd L$. We can then conclude by setting $M'_1=\ite{\pi_1L}{N_1'}{N_2'}$ and $t'_1=t'$. 

All of the remaining cases follow the same pattern. 
\end{proof}

\begin{lemma}[extrusion to normal form]
	\label{lemma:extrusion_iterate}
	Let $M\expd M'$, $t\expd t'$, $t\trapx M$ and $t'\apx M'$, for $t$ and $M$ closed terms both of type $\tyR^n$, for some $n\geq 0$. Then, $t'\reds t''$ and $M'\reds M''$ with $t''$ and $M''$ normal such that $t''\apx M''$. 
\end{lemma}
\begin{proof}[Proof Sketch]
By \refdef{TrApx}, there exist a normalizing reduction $\xi$ starting from $t$ and a normalizing reduction $\rho$ starting from $M$, such that $\xi\apx\rho$. The proof is by induction on $\rho$. 
\end{proof}
%
%

\begin{theorem}[Soundness]
	\label{th:Sound}
	 For every program $\Gamma\vdash M:\tyR$ and every $\seq r\in\Stab{M}\cap\diffdom{M}$, we have:
	$$\FBgrad M(\seq r) \reds \nabla(\sem{M})(\seq r).$$
\end{theorem}
\begin{proof}
	The assumption $\seq r\in\diffdom{M}$ tells us that there is an open ball $\Ball{\varepsilon_0}{\seq r}$ where $\nabla\!\sem{M}$ exists. By \refdef{Stab}, $\seq r\in \Stab M$ means that there is a simple program $t\apx M$ uniformly tracing $M$ in an open ball of $\seq r$, \ie, there exists $\varepsilon>0$ such that for all $\seq r'\in\Ball\varepsilon{\seq r}$ we have $t\isub{\seq r'}{\seq x}\trapx M\isub{\seq r'}{\seq x}$, and of course we may take $\varepsilon\leq\varepsilon_0$. This implies in particular that $\sem{t}$ and $\sem M$ coincide on $\Ball{\varepsilon}{\seq r}$ and, therefore, we have $\seq r\in\diffdom{t}$ as well.
	
	By \reflemma{expd_ad},  
	$t\isub{\seq r}{\seq x}\expd \FB t\isub{\pair{\seq r,\seq u}\!}{\seq x}$ and
	$M\isub{\seq r}{\seq x}\expd \FB M\isub{\pair{\seq r,\seq u}\!}{\seq x}$, where $u_i$ is either the normal form of $\iota_{i}^n1$ or $\iota_{i}^n$, depending whether $\FBSymb$ is $\FSymb$ or $\BSymb$ (recall Equations~\eqref{eq:grad_f} and~\eqref{eq:grad_b}). 
	Moreover, by \reflemma{apx-substitution} we have that $t\apx M$ gives $\FB t\isub{\pair{\seq r,\seq u}\!}{\seq x}\apx \FB M\isub{\pair{\seq r,\seq u}\!}{\seq x}$.
	We are then in position of applying \reflemma{extrusion_iterate} to $M\isub{\seq r}{\seq x}$, $t\isub{\seq r}{\seq x}$ (closed terms of ground type) and $\FB t\isub{\pair{\seq r,\seq u}\!}{\seq x}$, $\FB M\isub{\pair{\seq r,\seq u}\!}{\seq x}$. This gives us
	a normal form $t'$ of $\FB t\isub{\pair{\seq r,\seq u}\!}{\seq x}$ and $M'$ of $\FB M\isub{\pair{\seq r,\seq u}\!}{\seq x}$ such that $t'\apx M'$. By subject reduction (\refprop{subject_reduction}), $t',M'$ are closed normal forms of type $\FB\tyR$, so \reflemma{apx-programs} gives us $t'=M'$. 
	Since $\seq r\in\diffdom{t}$, we conclude by \refprop{SimpleSound}.
\end{proof}

%

\section{Unsoundness of AD}
\label{sect:Unsoundness}
\begin{definition}[unstable point]
	\label{def:unstable}
	The set of \emph{unstable points} of a program $M$, denoted by $\UStab{M}$, is the complement of $\Stab{M}$ (\refdef{Stab}) in $\dom{M}$, \ie, $\UStab{M}:=\dom{M}\setminus\Stab{M}$.
\end{definition}

We know from the remark after \refdef{Stab} that $\UStab{M}$ is closed in $\dom M$ (with respect to to the subspace topology). The goal of this section is to prove that it is a quasivariety, hence of measure zero. The main tool is the logical predicate defined in \reffig{Logic} and its adequacy (\reflemma{Adequacy}). The structure of the proof is standard, but some new notions are needed. First, our programs are first-order functions, \ie, terms with some free variables of ground type, so the logical predicate is indexed by a typing environment. Second, \reflemma{stab} states some properties of the notion of stability necessary to achieve the standard auxiliary lemmas of a logical predicate, such as closure under expansion (\reflemma{expansion_log}) or (a syntactic variant of) Scott-continuity (\reflemma{zero} and \reflemma{fix}). Third, and more important, the standard lemma of logical predicates for the conditional (\reflemma{if}) is particularly subtle. This should not come as a surprise: as discussed above
, the possibility of unsoundness of AD is due to conditionals. In particular, let us underline that \reflemma{if} uses the notion of completely quasicontinuous map introduced in \refsect{cqc}.

The logical predicate on which the proof is based is defined in \reffig{Logic}. In the rest of the section, unless otherwise stated, $\Gamma:=x_1^\tyR,\ldots,x_n^\tyR$ is a ground context. Moreover, if $M:A_1\to\cdots\to A_p\to B$ and $\seq L=L_1,\ldots,L_p$ such that $L_i:A_i$ for all $1\leq i\leq p$, then the notation $M\seq L$ stands for $ML_1\cdots L_p$, which is of course a term of type $B$.

\begin{figure*}
	\begin{minipage}{0.9\textwidth}
	\begin{align*}
		\logic{\tyR} &:= \{\Gamma\vdash M:\tyR \st \sem M \text{  is cqc and } \UStab{M} \text{ is a quasivariety} \} \\
		\logic{A\rightarrow B} &:= \{\Gamma\vdash M: A\rightarrow B \st \forall N\in\logic A, MN\in\logic B
		\} \\
		\logic{A_1\times\cdots\times A_k} &:= \{ 
			\Gamma\vdash M:A_1\times\cdots\times A_k \st \pr i M \in \logic{A_i}, \forall i\leq k
		\}
	\end{align*}
	\end{minipage}
	\caption{The definition of the logical predicate $\logic{A}$, with $\Gamma$ a ground context.}
	\label{fig:Logic}
\end{figure*}

\begin{lemma}
	\label{lemma:stab}
	We have the following inclusions, where the terms appearing in the statements are supposed to be typed under a ground context $\Gamma$.
	\begin{enumerate}
		\item\label{stab:map} Let $\funsa$ be a function symbol of arity $k$ and let $M_1,\dots,M_k$ be programs, then  $\bigcap_i\Stab{M_i}\subseteq\Stab{\funsa(M_1,\dots,M_k)}$.

		\item\label{stab:red}
		Let $R\red P$ be one of the rewriting rules in \reffig{reduction}, with $R,P$ of type $B_1\to\cdots\to B_p\to\tyR^m$. For all $1\leq i\leq p$, let $\Gamma\vdash L_i:B_i$. Then, for all $1\leq j\leq m$, $\Stab{\pi_j(P\seq L)}\subseteq\Stab{\pi_j(R\seq L)}$.

		\item\label{stab:if} Let $P:\tyR$ and $M_1,M_2$ be of type $B_1\to\cdots\to B_p\to\tyR^m$. For all $1\leq i\leq p$, let $\Gamma\vdash L_i:B_i$. Let $X_1:=\sem P^{-1}(\Real_{\leq 0})$ and $X_2:=\sem P^{-1}(\Real_{>0})$. 
		Then, for all $1\leq j\leq m$ and all $l\in\{1,2\}$, we have that $\Stab{\pi_j(M_l\seq L)}\cap\interior{X_l} \subseteq \Stab{\pi_j(\ite P{M_1}{M_2}\seq L)}$.
		
		\item\label{stab:fix} 
		Let $B=B_1\to\cdots\to B_p\to\tyR^m$, let $\Gamma\vdash L_0:A$ and $\Gamma\vdash L_i:B_i$ for all $1\leq i\leq p$. For all $k\in\Nat$ and $1\leq j\leq m$, $\Stab{\pi_j((\fix[k]{f^{A\to B}}{M})\seq L)}\subseteq\Stab{\pi_j((\fix{f^{A\to B}}M)\seq L)}$, where $\seq L:=L_0,L_1,\ldots,L_p$.
	\end{enumerate}
\end{lemma}
\begin{proof}[Proof Sketch]
We only detail the proof of the branching case, the other cases are easy variants.
		By taking the notations of point~\eqref{stab:if}, let $l\in\{1,2\}$ and let $\seq r\in\Stab{\pi_j(M_l\seq L)}\cap\interior{X_l}$. By definition, there exist $\varepsilon>0$, $u\apx \pi_j(M_l\seq L)$ such that, for all $\seq r'\in\Ball\varepsilon{\seq r}$, $u\isub{\seq r'}{\seq x}\trapx \pi_j(M_l\seq L)\isub{\seq{r'}}{\seq x}$ and $\seq r'\in\interior{X_l}$. Notice that $u=\pi_j(u'\seq{u''})$, with $u'\apx M_l$ and $\seq u''=\pair{\seq u_1''}\cdots\pair{\seq u_p}$ such that, for every $1\leq i\leq p$ and every element $u''_{i,h}$ of $\seq u''_i$, we have $u''_{i,h}\apx L_i$. Let $t := \pi_j((\pr \ell\!\pair{u',u'})\seq u'')$ and notice that $t\apx \pi_j(\ite P{M_1}{M_2}\seq L)$. Let us prove that $t\isub{\seq r'}{\seq x}\trapx \pi_j(\ite P{M_1}{M_2}\seq L)\isub{\seq r'}{\seq x}$. 

		Since $\seq r'\in\interior{X_l}\subseteq\dom P$, we have that $P\isub{\seq r'}{\seq x}$ is normalizing. Since $P$ is ground, by \refprop{headnf} there is a head reduction sequence $\rho:P\isub{\seq r'}{\seq x}\reds q$ such that $q$ is a numeral. Let now $\hCtxt=\pi_j(\ite{	\ctxthole}{M_1}{M_2}\seq u'')\isub{\seq r'}{\seq x}$. Notice that $\nu\apx\hctxt\rho$ for $\nu$ the empty reduction sequence of $t\isub{\seq r'}{\seq x}$. Moreover, by hypothesis we have normalizing reduction sequences $\nu'\apx\rho'$ from $u\isub{\seq r'}{\seq x}=\pi_j(u'\seq{u''})\isub{\seq r'}{\seq x}$ and $\pi_j(M_l\seq L)\isub{\seq r'}{\seq x}$, respectively. Furthermore, we have the head reduction steps:
		\begin{align*}
			\nu_0:&\pi_j(\pi_\ell\pair{u',u'}\seq{u''})\isub{\seq r'}{\seq x}\red\pi_j(u'\seq{u''})\isub{\seq r'}{\seq x},&
			\rho_0:&\pi_j(\ite q{M_1}{M_2}\seq L)\isub{\seq r'}{\seq x}\red \pi_j(M_l\seq L)\isub{\seq r'}{\seq x}
		\end{align*}
		such that $\nu_0\apx\rho_0$. 
		We then have 
		$\nu\nu_0\nu' \apx \hctxt\rho\rho_0\rho'$,
		which allows us to conclude.
%
\end{proof}

\begin{lemma}[function symbols]
	\label{lemma:map}
	Let $\funsa$ be a function symbol of arity $k$ and, for each $1\leq i\leq k$, $M_i\in\logic{\tyR}$. If $\sem[]{\funsa}$ is cqc, then $\funsa(M_1,\dots,M_k)\in\logic{\tyR}$.
\end{lemma}
%

\begin{lemma}[closure under expansion]
	\label{lemma:expansion_log}
	Let $R\red P$ be one of the rewriting rules in Figure~\ref{fig:reduction}. If $P\in\logic{A}$, then $R\in\logic{A}$.
\end{lemma}
%
%

\begin{lemma}[conditional]
	\label{lemma:if}
	If $P\in\logic{\tyR}$ and $M_1,M_2\in\logic{A}$, then $\ite{P}{M_1}{M_2}\in\logic{A}$.
\end{lemma}
\begin{proof}
	Let $A = A_1\rightarrow \dots \rightarrow A_p\rightarrow \tyR^m$. It is enough to prove that, given $\seq L=L_1,\ldots,L_p$ such that $L_i\in\logic{A_i}$ for every $1\leq i\leq p$, we have $\pi_j(\ite{P}{M_1}{M_2}\seq L)\in\logic{\tyR}$ for every $1\leq j\leq m$.
	
	Let $N:=\pi_j(\ite{P}{M_1}{M_2}\seq L)$, $Q_1:=\sem{P}^{-1}(\Real_{\leq 0})$ and $Q_2:=\sem P^{-1}(\Real_{>0})$. Observe that both $Q_1$ and $Q_2$ are quasiopen, because $\Real_{\leq 0}$ and $\Real_{>0}$ are quasiopen and $\sem P$ is cqc by hypothesis. Furthermore, $\Real^l\times Q_i$ is quasiopen for every $l\geq 0$ and $i\in\{1,2\}$, because it is the inverse image of $Q_i$ via a projection $\Real^l\times\Real^p\to\Real^p$, and these are cqc (\reflemma{qcont}.\ref{qcont:basic}).
	
	To prove that $\sem{N}$ is cqc we need to show that, for every $l\geq 0$ and every quasiopen set $Q\subseteq\Real^{l+1}$, the set $(\id_{\Real^l}\times\sem N)^{-1}(Q)$ is quasiopen. But \refprop{head} tells us that $(\id_{\Real^l}\times\sem N)^{-1}(Q)$ is equal to 
	$
	\left((\Real^l\times Q_1)\cap(\id_{\Real^l}\times\sem{\pi_j(M_1\seq L)})^{-1}(Q)\right)
	\cup\left((\Real^l\times Q_2)\cap(\id_{\Real^l}\times\sem{\pi_j(M_2\seq L)})^{-1}(Q)\right)
	$,
	which is quasiopen because $\sem{\pi_j(M_1\seq L)}$ and $\sem{\pi_j(M_1\seq L)}$ are cqc by hypothesis.
	
	For what concerns the unstable points, we first observe that, by \refprop{head},
	$\dom{N} = \dom{P}\cap((Q_1\cap\dom{M_1})\cup(Q_2\cap\dom{M_2}))=(Q_1\cap\dom{M_1})\cup(Q_2\cap\dom{M_2})$,
	the second equality holding because $Q_1,Q_2\subseteq\dom P$. We may therefore write
	\begin{align*}
		\UStab{N} &= \left(\bigcup_{i\in\{1,2\}}Q_i\cap\dom{M_i}\right) \setminus \Stab{N} 
		=\!\!\!
		\bigcup_{i\in\{1,2\}}\!(Q_i\cap\dom{M_i})\setminus \Stab{N} 
		\subseteq\!\!\! 
		\bigcup_{i\in\{1,2\}}\!(Q_i\cap\dom{M_i})\setminus(\interior{Q_i}\cap\Stab{M_i})
	\end{align*}
	where the inclusion is by \refsublemma{stab}{if}. Now, for all $i\in\{1,2\}$, let us write $A_i:=(Q_i\cap\dom{M_i})\setminus(\interior{Q_i}\cap\Stab{M_i})$. Notice that, by \reflemma{qopen}, $Q_i=\interior{Q_i}\cup Z_i$ with $Z_i$ a quasivariety, hence
	\[
		A_i = ((Q_i\cap\dom{M_i})\setminus\interior{Q_i})\cup((Q_i\cap\dom{M_i})\setminus\Stab{M_i})
		\subseteq (Q_i\setminus\interior{Q_i})\cup(\dom{M_i}\setminus\Stab{M_i}) = Z_i\cup\UStab{M_i},
	\]
	so 
	$A_i$ is a quasivariety, because $Z_i$ and $\UStab{M_i}$ are. Since $\UStab{N}\subseteq A_1\cup A_2$, we are done.
\end{proof}

\begin{lemma}[divergence]
	\label{lemma:zero}
	For every type $A\to B$, $\Omega_{A\to B}\in\logic{A\to B}$.
\end{lemma}

\begin{lemma}[fixpoints]
	\label{lemma:fix}
	If $\forall k\in\Nat, \fix[k] f M\in\logic{A\to B}$, then $\fix f M\in\logic{A\to B}$.
\end{lemma}
%

\begin{lemma}[adequacy]\label{lemma:Adequacy}
	Suppose that the primitive functions of \PCF\ are admissible. Let $\Gamma:=x_1^\tyR, \dots, x_n^\tyR$, $\Delta:= y_1^{A_1},\ldots,y_m^{A_m}$, let $\Gamma,\Delta\vdash M:A$ and let $\Gamma\vdash N_i\in \logic{A_i}$ for all $1\leq i\leq m$. Then, $$M\isub{N_1}{y_1}\cdots\isub{N_m}{y_m}\in\logic A.$$
\end{lemma}

\begin{theorem}
	\label{th:Unsound}
	Assuming that the primitive functions of \PCF\ are admissible, for every program $\Gamma\vdash M:\tyR$ the set
	$$
		\mathrm{Fail}(M):=\{\seq r\in\diffdom M\;;\; \FBgrad M(\seq r) \not\reds \nabla(\sem{M})(\seq r)\}
	$$
	is a quasivariety, hence of measure zero.
\end{theorem}
\begin{proof}
	By \refth{Sound}, we know that $\mathrm{Fail}(M)\subseteq\UStab M$, which is a quasivariety because $M\in\logic{\tyR}$ by \reflemma{Adequacy}.
\end{proof}

\section{Discussion and Perspectives}
\label{sect:Conclusions}
\paragraph*{On the significance of the measure zero bound}
Since the set of real numbers representable on an actual computer is of measure zero in $\Real$ (it is finite!), one may feel skeptical about the significance of \refth{Unsound}. This issue was already raised by Speelpenning~\cite{Speelpenning} while commenting on Joss's theorem~\cite{Joss}. He defines a program similar to the following:
\begin{align*}
	\mathsf{SlowId} &:= \lambda x^\tyR.\left(\fix{f^{\tyR\to\tyR}}{\lambda y^\tyR.\ite{x-y}{\ite{y-x}{y}{f\,\mathsf{Next}}}{f\,\mathsf{Next}}}\right)0,
\end{align*}
where $\vdash\mathsf{Next}:\tyR$ is a primitive which cycles through machine-representable real numbers, based on some internal state (Speelpenning uses a random number generator in his example). So, given $r\in\Real$, $\mathsf{SlowId}(r)$ will eventually  output $r$ if this is machine-representable, and diverge otherwise.  When executed on an actual computer, every $r$ is necessarily representable and $\mathsf{SlowId}$ behaves like the identity. And yet, $\FBgrad{\mathsf{SlowId}\,x}(r)\reds 0$ for every machine-representable $r$, because $\mathsf{Next}$ is treated as a constant by AD transformations (there is no sensible alternative). Speelpenning concludes that, although $\mathsf{SlowId}$ is not a counterexample to Joss's theorem (or to ours), from the practical viewpoint AD fails \emph{everywhere} on it.

We believe that Speelpenning's example is misleading. The reason why $\mathsf{SlowId}$ is not a counterexample to \refth{Unsound} is not 
that the set where AD fails on $\mathsf{SlowId}$ is of measure zero because it coincides with the set of machine-representable reals; it is because $\sem[x^\tyR]{\mathrm{SlowId}(x)}$ is nowhere differentiable! That is, in the notations of \refth{Unsound}, we actually have $\mathrm{Fail}(\mathsf{SlowId}(x))=\emptyset$ because $\diffdom{\mathsf{SlowId}(x)}=\emptyset$, and this is because $\sem[x^\tyR]{\mathsf{SlowId}(x)}$ is defined only on a discrete set.

Anyway, if the set $R:=\{r_1,\ldots,r_c\}$ of machine-representable reals is finite, there are impractically large but straightforward programs achieving the intended behavior of Speelpenning's example. For instance, with some syntactic sugar, define
$$\mathsf{SlowId}'\quad:=\quad\lambda x^\tyR.\mathsf{if}\,x=r_1\,\mathsf{then}\,r_1\,\mathsf{else}\,(\ldots \mathsf{if}\,x=r_c\,\mathsf{then}\,r_c\,\mathsf{else}\,x\ldots).$$
We have that $\sem[x^\tyR]{\mathsf{SlowId}'(x)}$ is actually the identity function, so $\mathrm{Fail}(\mathsf{SlowId}'(x))=R$ and we may legitimately say that AD is wrong ``everywhere''. But this is just a giant-sized version of the program $\mathsf{SillyId}$ of the Introduction, and speaks more of the contrivance of toying with \PCF\ as a machine-executable language (which it is not) than of the value of our result. In general, questioning the significance of \refth{Unsound} on the grounds that computers are finite is like questioning Turing machines because of their infinite tape, or objecting to the whole idea of studying the asymptotic complexity of programs because in practice we only implement finite functions, whose asymptotic complexity is $O(1)$. In our opinion, there is little point in discussing this standpoint further.

More constructively, we may argue that the significance of \refth{Unsound} lies in the fact that it gives a finer bound than just measure zero. Let us call \PCF\ with only the ``mandatory'' primitive functions (constants, addition, multiplication) \emph{minimal \PCF}.  This is already enough to express all differentiable programming architectures based on neural networks with rectified linear unit activation.  Moreover, by using Taylor series, minimal \PCF\ may also approximate every analytic function with arbitrary precision, so it has a wide range of potential applications. \refth{Unsound} tells us that, if $f:\Real^n\pto\Real$ is a function definable in minimal \PCF, then the set of points on which $\nabla\!f$ exists but AD methods fail to compute it is contained in a countable union of algebraic varieties (\ie, zeros of polynomials). In particular, when $n=1$, this set is countable.

Zero sets of polynomial equations have been studied for literally millennia as part of the vast field known as algebraic geometry. Albeit extremely complex in general, many results exist on their structure, which may be described or approximated very accurately in several cases. It is not excluded that, in the future, these results may be leveraged to develop static analysis techniques (\eg\ type systems) for establishing the absence of errors in differentiable programs.

\paragraph*{From Gradients to Jacobians}
We limited our attention to programs implementing functions $\Real^n\pto\Real^m$ with $m=1$. The case $m>1$, in which one would speak of Jacobians rather than gradients, is conceptually identical. First of all, observe that a function $f:\Real^n\pto\Real^m$ may always be decomposed into $m$ functions $f_i:=\pr if:\Real^n\pto\Real$, so restricting primitives to one output causes no loss of generality and \reffig{AD} needs no modification. When $\Gamma\vdash M:\tyR^m$ with $\Gamma$ containing $n$ variables and $m>1$, what needs to be modified are the Equations \eqref{eq:grad_f} and \eqref{eq:grad_b}, which must yield an $m\times n$ matrix whose lines are built out of $m$ expressions of the form $\FBgrad{\pr i M}$. \refth{Sound} and \refth{Unsound} 
lift to this setting because the Jacobian is just the collection of the $m$ gradients.

\paragraph*{Internalizing AD}
The transformations of \reffig{AD}, and thus the definition of $\FBgrad{M}$ for a program $M$ (Equations~\eqref{eq:grad_f} and~\eqref{eq:grad_b}) are external to \PCF: a programmer may apply them for instance via a compiler, but the transformations are not accessible from within the program itself. For practical purposes, it would be interesting to have a programming language in which $\overrightarrow{\mathop{grad}}$ and $\overleftarrow{\mathop{grad}}$ are syntactic constructs, 
typed $x_1^\tyR,\ldots,x_n^\tyR\vdash \FBgrad{M}:\tyR^n$ whenever $x_1^\tyR,\ldots,x_n^\tyR\vdash M:\tyR$,
and with $\FBgrad{M}$ being executed in such a way as to reflect the application of 
AD to $M$. A naive way of achieving this would be to turn the definition of \reffig{AD} into rewriting rules; a more sophisticated approach was provided by Pearlmutter and Siskind for 
Stalin$\nabla$~\cite{PearlmutterSiskind}.

The errors introduced by AD have the important consequence that \emph{such an internalization is impossible without breaking the expected extensional semantics of programs}. This is because, as any denotational semantics, the standard semantics defined in \refsect{Terms} is contextual, in the sense that $\sem[]{M}=\sem[]{N}$ implies $\sem[]{\ctxt{M}}=\sem[]{\ctxt{N}}$ for any context $\Ctxt$. Now, referring to \eqref{ex:pcfterms}, we have $\sem[x^\tyR]{\mathsf{SillyId}(x)}=\sem[x^\tyR]{x}$ and yet we know that $\sem[x^\tyR]{\FBgrad{\mathsf{SillyId}(x)}}\neq\sem[x^\tyR]{\FBgrad{x}}$, because the latter two functions differ at $0$. So any denotational semantics of \PCF\ with ``internal AD'' needs to interpret $\mathsf{SillyId}$ and $\mathsf{Id}:=\lambda x^\tyR.x$ differently. ``Resource-sensitive'' semantics coming from linear logic do distinguish them, but it is easy to find other examples on which these semantics too fail. An example of denotational semantics which consistently works is the one introduced by Abadi and Plotkin~\cite{AbadiP20}, whose first order language does have internal AD. In that semantics, $\mathsf{Id}$ is the identity whereas $\mathsf{SillyId}$ is a ``partial identity'', undefined at $0$. It is not clear whether this extends to higher order (and thus to \PCF, similarly to~\cite{DiGianantonioEdalat}), but assuming it does, the meaning it gives to programs is somewhat unusual: for example, using the definition given in \eqref{ex:pcfterms}, $\mathsf{Floor}(r)$ would diverge whenever $r$ is an integer. This is a further drawback of partial conditional semantics, in addition to the one pointed out in Sect.~\ref{subsect:ad} (concerning example \eqref{ex:eqproj}).

Another loosely related remark worth making at this point is that, seen as a functional on the Scott domain $\Real_\bot\to\Real_\bot$, where $\Real_\bot$ is the ``flat'' Scott domain typically used for interpreting $\tyR$ as a ground type with total conditionals, the derivative operator $\partial$ is \emph{not} Scott continuous.\footnote{Given $A\subseteq\Real$, say that $\chi$ is the \emph{indicator function} of $A$ if $\chi(x)=0$ when $x\in A$ and $\chi(x)=\bot$ otherwise. For $n>0$, let $I_n:=\ ]-\infty,0]\ \cup\ ]\frac{1}{n},+\infty[$ and let $\varphi_n$ be the indicator function of $I_n$. Notice that each $\varphi_n$ is differentiable on $J_n:=I_n\setminus\{0\}$ and $\partial\varphi_n$ is the indicator function of $J_n$. As elements of the Scott domain $\Real_\bot\to\Real_\bot$, the functions $(\varphi_n)_{n>0}$ and $(\partial\varphi_n)_{n>0}$ form two directed chains whose suprema are the identically zero function and the indicator function of $\Real\setminus\{0\}$, respectively. In particular, $\partial\!\left(\sup_{n>0}\varphi_n\right)\neq\sup_{n>0}\partial\varphi_n$, so $\partial$ is not Scott continuous.} Although the technical consequences of this observation are not entirely clear, from an intuitive point of view it means that the derivative operator is ``not computable'', and therefore no recursive procedure (like those given by AD transformations) will be error-free.

\begin{acks}                            
We would like to thank A.~Brunel, T.~Ehrhard and B.A.~Pearlmutter for useful comments and discussions. This work was partially supported by \grantsponsor{ANR}{ANR}{https://anr.fr/} PRC project PPS (\grantnum{ANR}{ANR-19-CE48-0014}).
\end{acks}

\newpage
\bibliography{Biblio.bib}

\newpage
\appendix
\section{Appendix}
%

\section{Proofs of Section~\ref{sect:Prelim}}\label{app:prel}

\replemma{qopen}
{\it	Let $Q\subseteq\Real^n$ be quasiopen. Then:
	\begin{enumerate}
		\item there exists an open set $U$ and a quasivariety $Z$ such that $Q=U\cup Z$;
		\item $\border Q$ is a quasivariety. 
		Hence, in the above one may always take $U=\interior{Q}$ and $Z=\border{Q}$.
	\end{enumerate}
}
\begin{proof}
	Point \ref{qopen:decomp} is by structural induction on $Q$. If $Q$ is open, the result trivially holds with $U:=Q$ and $Z:=\emptyset$. If $Q=h^{-1}(0)$ for a basic function $h:\Real^n\to\Real$, then we may suppose $h$ to be not identically zero, for otherwise $Q=\Real^n$ and we fall into the previous case. The result then holds by definition with $U:=\emptyset$ and $Z:=Q$. If $Q=\bigcup_{i\in I}Q_i$, then by the induction hypothesis there exist open sets $(U_i)_{i\in I}$ and quasivarieties $(Z_i)_{i\in I}$ such that $$Q=\bigcup_{i\in I}U_i\cup Z_i=\bigcup_{i\in I}U_i\cup\bigcup_{i\in I}Z_i,$$ the first union being open and the second union being a quasivariety because it is countable and each $Z_i$ is a quasivariety. If $Q=Q'\cap Q''$, then by the induction hypothesis there exist open sets $U',U''$ and quasivarieties $Z',Z''$ such that
	$$Q=(U'\cup Z')\cap(U''\cup Z'')=(U'\cap U'')\cup(U'\cap Z'')\cup(Z'\cap U'')\cup(Z'\cap Z'').$$
	We may therefore conclude by letting $U:=U'\cap U''$, which is open, and $Z:=(U'\cap Z'')\cup(Z'\cap U'')\cup(Z'\cap Z'')$, which is a quasivariety because it is a finite union of subsets of quasivarieties.
	
	For point \ref{qopen:border}, we apply point \ref{qopen:decomp} and obtain $Q=U\cup Z$ with $U$ open and $Z$ a quasivariety.  
	Now, observe that, by definition of interior, $U\subseteq\interior{Q}$, therefore $$\border Q=Q\setminus\interior{Q}\subseteq Q\setminus U=(U\cup Z)\setminus U\subseteq Z,$$
	so $\border Q$ is a quasivariety 
	because $Z$ is.
\end{proof}

For the sake of proving points (\ref{qcont:basic}) and (\ref{qcont:pairing}) below, we exploit the fact that a clone on a set $A$ may be equivalently defined as a set $\clone P$ of functions $A^n\pto A$ (for varying $n$) such that:\footnote{For the acquainted reader, a(n abstract) clone is the same as a cartesian operad.}
\begin{itemize}
	\item $\clone P$ contains the identity and is closed under \emph{operadic composition}, meaning that if $f:A^k\pto A$ and $g_1:A^{n_1}\pto A,\ldots,g_k:A^{n_k}\pto A$ are in $\clone P$, then the function of type $A^{n_1+\cdots+n_k}\pto A$ defined by $(\seq a_1,\ldots,\seq a_k)\mapsto f(g_1(\seq a_1),\ldots,g_k(\seq a_k))$ for all $\seq a_i\in A^{n_i}$, is also in $\clone P$;
	\item if $f\in\clone P$ is of arity $n$, then for any permutation $\sigma$ on $\{1,\ldots,n\}$ the function $f_\sigma$ defined by $f_\sigma(x_1,\ldots,x_n):=f(x_{\sigma(1)},\ldots,x_{\sigma(n)})$ is also in $\clone P$;
	\item $\clone P$ contains all projections and if $f\in\clone P$ is of arity $n+2$, then the function $g$ defined by $g(x_1,\ldots,x_n,y):=f(x_1,\ldots,x_n,y,y)$ is also in $\clone P$.
\end{itemize}

\replemma{qcont}
{\it
	We have the following properties:
	\begin{enumerate}
		\item a function $f:\Real^n\to\Real^m$ is quasicontinuous iff for every $Q$ which is either open or the zero set of a basic function, $f^{-1}(Q)$ is quasiopen.
		\item Identities are cqc and cqc functions are stable under composition.
		\item Basic functions are cqc. In particular, projections are cqc.
		\item If $f:\Real^k\pto\Real^m$ and $g:\Real^k\pto\Real^n$ are cqc, then the function $\pair{f,g}:\Real^k\to\Real^{m+n}$ defined by $\pair{f,g}\!(z):=(f(z),g(z))$ if $z\in\dom f\cap\dom g$ and undefined otherwise, is also cqc.
	\end{enumerate}
}
\begin{proof}
	\begin{enumerate}
		\item The implication from left to right is by definition. From right to left, given $Q\subseteq\Real^m$ quasiopen, we prove that $f^{-1}(Q)$ is quasiopen by structural induction on $Q$. The case in which $Q$ is open or the zero set of a basic function are the hypothesis. If $Q=\bigcup_{i\in I}Q_i$ with $I$ countable, we have
		$$f^{-1}(Q)=\bigcup_{i\in I}f^{-1}(Q_i),$$
		which is quasiopen because it is a countable union of sets which the induction hypothesis guarantees us to be quasiopen. If $Q=Q'\cap Q''$, then
		$$f^{-1}(Q)=f^{-1}(Q')\cap f^{-1}(Q''),$$
		which again is quasiopen because it is an intersection of two sets which the induction hypothesis guarantees us to be quasiopen.

		\item Identities are trivially cqc. Let $f,g$ be composable cqc functions. Observe that quasicontinuous functions are obviously stable under composition. Now, we have $\id\times(g\circ f)=(\id\times g)\circ(\id\times f)$, which is quasicontinuous by hypothesis and the above remark.
		\item Let $g:\Real^n\to\Real$ be basic and let $Q\subseteq\Real^{m+1}$ be quasiopen, with $m\in\Nat$ arbitrary. By point \ref{qcont:decomp}, in order to prove that $\id_{\Real^m}\times g$ is quasicontinuous it is enough to show that $(\id_{\Real^m}\times g)^{-1}(Q)$ is quasiopen for any $Q$ open or zero set of a basic function. If $Q=U$ with $U$ open, by continuity of $\id_{\Real^m}\times g$ we have that $(\id_{\Real^m}\times g)^{-1}(U)$ is open, hence quasiopen. If $Q=h^{-1}(0)$ with $h:\Real^m\to\Real$ basic, then
		$(\id_{\Real^m}\times g)^{-1}(h^{-1}(0))=(h\circ(\id_{\Real^m}\times g))^{-1}(0)$, and we conclude because $h\circ(\id_{\Real^m}\times g)$ is basic, being the composition of basic functions.

		\item We start by making the following claims:
		\begin{enumerate}
			\item every permutation $\sigma:\Real^m\to\Real^m$ is cqc;
			\item the diagonal function $\delta_k:\Real^k\to\Real^{2k}$ such that, for all $x\in\Real^k$, $\delta(x)=(x,x)$, is cqc.
		\end{enumerate}
		Both claims follow from the observation that these functions are continuous and, if $h$ is basic, then $h\circ\sigma$ and $h\circ(\id\times\delta_k)$ are basic (and $\id\times\sigma$ is still a permutation), so we conclude by point \ref{qcont:decomp}.
		
		Let now $f$ and $g$ be as in the hypothesis. We have, for all $l\in\Nat$, 
		$$\id_{\Real^l}\times\pair{f,g}=\sigma\circ(\id_{\Real^{l+n}}\times f)\circ\sigma'\circ(\id_{\Real^{l+k}}\times g)\circ(\id_{\Real^l}\times\delta_k)$$
		where $\sigma:\Real^{l+m+n}\to\Real^{l+n+m}$ is the permutation such that $\sigma(x,y,z)=(x,z,y)$ for all $x\in\Real^l$, $y\in\Real^m$ and $z\in\Real^n$, and similarly for $\sigma':\Real^{l+k+n}\to\Real^{l+n+k}$, so the result follows from the above claims and point \ref{qcont:comp}.
	\end{enumerate}
\end{proof}

\section{Proofs of Section~\ref{sect:Soundness}}\label{app:sound}

\begin{lemma}\label{lemma:apx-reflexivity}
	For any simple term $t$, $\Xi\vdash t\apx t$, where $\Xi$ is the identity map on the free variables of $t$. 
\end{lemma}
\begin{proof}
	By structural induction on $t$. In the base case, we use the variable rule with $n=1$ and $p=x$, so that there is no projection in the conclusion.
\end{proof}

\begin{lemma}\label{lemma:apx-normalisation}
Let $t\apx M$. Then:
\begin{enumerate}
	\item $t$ normal implies that $M$ is a \emph{simple} normal form;
	\item conversely, $M$ closed normal form of type $\tyR^n$ implies $t$ normal.
\end{enumerate}
\end{lemma}
\begin{proof}
	By structural induction on $t\apx M$. Point 1 is immediate once observed that $M$ cannot be a conditional or a fixpoint if $t$ is a normal form. For point 2, the hypothesis of being closed and of type $\tyR^n$ for some $n\geq 0$ assures that being normal implies being an $n$-tuple of numerals. 
\end{proof}

\replemma{apx-programs}
{\it	Let $M$ and $t$ be normal forms of type $\FB[n]\tyR$ whose free variables have type belonging to $\{\FB[n]\tyR,\tyR,\tyR^n,\tyR^{\perp_n}\}$. If $t\apx M$ then $t=M$.
}
\begin{proof}
	Let $\tyR'$ be $\tyR^n$ or $\tyR^{\bot_n}$ according to whether $\FBSymb_n$ is $\FSymb_n$ or $\BSymb_n$, respectively. Let $t$ and $M$ be as in the hypothesis, let $\Gamma$ be the typing environment of the two terms. By \reflemma{apx-normalisation}, notice that $M$ is also a simple term, so in particular it cannot be a conditional. Also, because $M$ is of type $\FB[n]\tyR=\tyR\times\tyR'$, and $\Gamma$ does not have variables of type $A\to\FB[n]\tyR$ for some $A$, we have that $M$ is either a variable of type $\FB[n]\tyR$ or a product $\pair{M_1,M_2}$, with $M_1:\tyR$ and $M_2:\tyR'$. In the first case, $t\apx M$ implies $t=M$, while in the second case it implies $t=\pair{t_1,t_2}$ with $t_i\apx M_i$ and $t_1:\tyR$, $t_2:\tyR'$.

	Let us consider the case $\tyR'=\tyR^{\bot_n}$ (the case $\tyR'=\tyR^n$ is a simpler variant). Under this hypothesis, $M_2$ is either a variable in $\Gamma$ of type $\tyR^{\perp_n}$ or it  must be equal to $\lambda a^{\tyR}.M'$ for some $M'$ of type $\tyR^n$ under the context $\Gamma':=\Gamma,  a^{\tyR}$. In the first case, we trivially have $t_2 = M_2$. Otherwise $t_2=\lambda a^{\tyR}.t'$ with $t'\apx M'$ (notice that the fact that $t_2$ has exactly the same type as $M_2$ assures that the type of its abstracted variable $a$ is $\tyR$ and not some product $\tyR^k\apx \tyR$). 

	We will now infer $t'=M'$ from $t'\apx M'$ and the fact that both terms have type $\tyR^n$ under the context $\Gamma'$ defined above. The proof is by induction on $M'$. 

	If $M'=xM''$ with $x$ of type $\tyR^{\perp_n}$ and $M''$ of type $\tyR$, then $t'=xt''$ with $t''$ of type $\tyR$ and we may conclude by induction hypothesis on $t''\apx M''$.

	Otherwise, if $M'$ is not an application and $n=1$, then $M'$ is either a numeral, a variable of type $\tyR$ or some $\funsa(M_1',\dots,M_k')$. In the first two cases we have trivially $t'=M'$. In the third case we have $t'=\funsa(t_1',\dots,t_k')$, with $t_i'\apx M_i'$ and both of type $\tyR$. We also conclude by induction hypothesis. 

	Finally, if $M'$ is not an application and $n>1$, then $M'=\pair{M_1',\dots,M_n'}$ with each $M_i'$ of type $\tyR$, then $t'=\pair{t_1',\dots,t_n'}$ with each $t_i'$ of type $\tyR$ and $t_i'\apx M_i'$. Again we conclude by induction hypothesis. 
\end{proof}

\replemma{apx-substitution}
{\it	If $\Xi\vdash t\apx M$, then:
	\begin{enumerate}
	\item
	$\FB{\Xi}\vdash\FB{t}\apx\FB{M}$, where $\FBSymb$ turns any assignment $p^{A'}\apx x^A$ of $\Xi$ into $p^{\FB{A'}}\apx x^{\FB{A}}$.
	\item
	Suppose $\Xi=\Xi',x^A\apx x^A$. Then for every closed simple term $u$ of type $A$, we have $\Xi'\vdash t\isub{u}{x}\apx M\isub{u}{x}$.
	\end{enumerate}
}
\begin{proof}
	Item \ref{lemma:apx-substitution-D} is proved by induction on a derivation of $\Xi\vdash t\apx M$. We show the cases in which the last rule is a conditional or a function symbol, all other cases being similar or immediate. Let $M=\ite P{N_1}{N_2}$ and $t=\pi_i\!\pair{u,u}$ for some $i\in\{1,2\}$ and $u\apx N_i$. Then we have 
	$\FB M=\ite{\pi_1 \FB P}{\FB{N_1}}{\FB{N_2}}$ and $\FB t= \pi_i\!\pair{\FB u, \FB u}$ and we conclude because by induction hypothesis $\FB u\apx \FB{N_i}$. Let now $M=\phi(M_1,\dots,M_k)$ and $t=\phi(t_1,\dots,t_k)$, with $t_i\apx M_i$. Then $\FB{t}=u\FB{\seq t}$ and $\FB{M}=u\FB{\seq M}$ with $u$ the closed simple term defined in \refsubfig{fb-ground} and depending only on $\phi$. By \reflemma{apx-reflexivity}, $u\apx u$ and by induction hypothesis $\FB{t_i}\apx\FB{M_i}$ for every $i\leq k$. We may then conclude $\FB t\apx \FB M$.

	Item \ref{lemma:apx-substitution-sub} is also proved by induction on a derivation of $\Xi\vdash t\apx M$, using \reflemma{apx-reflexivity} in the base case. 
\end{proof}

\replemma{apx-substs}
{\it	We have that $\Xi\vdash w\apx M\isub{N}{x}$ is equivalent to
	\begin{itemize}
		\item $w=t\isub{u_1}{x_1}\dots\isub{u_n}{x_n}$, for some $n\in\Nat$ and terms $t,u_1,\ldots,u_n$,
		\item such that $\Xi,p^{A_1\times\dots\times A_n}\apx x^A\vdash t\isub{\pi_1 p}{x_1}\dots\isub{\pi_n p}{x_n}\apx M$, $p$ not free in $t$,
		\item and $\Xi\vdash u_i\apx N$ for all $1\leq i\leq n$.
	\end{itemize}
	In particular, $\Xi,p^{A_1\times\dots\times A_n}\apx x^A\vdash w\apx M$ implies $w=t\isub{\pi_1 p}{x_1}\dots\isub{\pi_n p}{x_n}$ for some $t$ not containing $p$ free. 
}
\begin{proof}
	By induction on $M$. The second claim follows by considering $M=M\isub xx$ and remarking that $\Xi,p^{A_1\times\dots\times A_n}\apx x\vdash u_i\apx x$ implies $u_i=\pi_i p$ by the variable rule of \reffig{TraceTerms}.
\end{proof}

\replemma{apx_expansion}
{\it Let $\sigma:R\red P$ be a rewriting step. For any $w\apx P$, there exist $t\apx R$ and $\xi:t\reds w$ such that $\xi\apx\sigma$. 
}
\begin{proof}
By case inspection. If $R=(\lambda x.M)N$, then $P=M\isub{N}{x}$ and by \reflemma{apx-substs}, $w$ must be of the form $t\isub{\seq u}{\seq x}$, with $p\apx x,\Xi\vdash t\isub{\boldsymbol \pi p}{\seq x}\apx M$ and $u_i\apx N$ for all $u_i$ in $\seq u$. We may then define $\xi:(\lambda p.t\isub{\boldsymbol\pi p}{\seq x})\pair{\seq u}\reds t\isub{\seq u}{\seq x}$. Notice that
$(\lambda p.t\isub{\boldsymbol\pi p}{\seq x})\pair{\seq u}\apx (\lambda x.M)N$, as well as $\xi\apx\sigma$.

The case $R=\fix f M$ is similar to the previous one and the other cases are simpler. 
\end{proof}

\begin{lemma}[uniformity]
	\label{lemma:uniformity}
	Let $\Xi,x^\tyR\apx x^\tyR\vdash t\apx R$, let $\sigma:R\red P$ be a rewriting step and let $\xi$ be a reduction sequence starting from $t$ and such that $\xi\apx\sigma$. Then, we have $\xi\isub{r}{x}\apx\sigma\isub{r}{x}$ for any ground variable $x$ and $r\in\Real$. 
\end{lemma}
\begin{proof}
By inspecting the cases of \refdef{ApxRed_rules}, using \reflemma{apx-substitution}.\ref{lemma:apx-substitution-sub} to obtain $\Xi\vdash t\isub{r}{x}\apx R\isub{r}{x}$.
\end{proof}	

\begin{lemma}[endpoints]
	\label{lemma:endpoints}
	Let $\xi:t\reds u$ and $\rho: M\reds N$. If $\xi\apx\rho$, then $t\apx M$ and $u\apx N$. Moreover, if $u$ is a normal form, then so is $N$.
\end{lemma}
\begin{proof}
	The first implication is proved by induction on $\rho$. If the length is one, then it is a simple consequence of Definitions~\ref{def:ApxRed_rules} and~\ref{def:ApxRed_steps}. In particular, in the case $\rho=(\hCtxt,R,P)$ with the hole of $\hCtxt$ in the guard of a conditional, then notice that $t=u\apx \hctxt R$ implies $t=u\apx \hctxt P$, as $\apx$ does not depend on the guards. 

	The second implication is a consequence of \reflemma{apx-normalisation}.
\end{proof}

\begin{lemma}\label{lemma:expd-substitution}
Let $x:A,\Gamma\vdash M:B$, $\Gamma\vdash N:A$. Then $M\expd M'$ and $N\expd N'$ implies $M\isub Nx\expd M'\isub{N'}{x}$.
\end{lemma}
\begin{proof}
By structural induction on the derivation of $M\expd M'$. In the case $M=\phi(M_1,\dots,M_k)$ then $M'=u M_1'\dots M_k'$, with some \emph{closed} simple $u$. So in particular, $u\isub{N'}{x}=u$. We then conclude immediately by the induction hypothesis on the various $M_i\expd M_i'$.
\end{proof}

\begin{lemma}
\label{lemma:expd_closed_nf}
If $\Gamma \vdash M:A$ and $M\expd M'$, then $\FB \Gamma \vdash M':\FB A$. In particular, if $M$ is closed, then $M'$ is also closed. Furthermore, if $M$ is a closed normal form of type $\tyR^n$, then $M'$ is also a normal form.
\end{lemma}
\begin{proof}
By induction on  $M\expd M'$ one can check that  $\FB \Gamma \vdash M':\FB A$. 

The last statement follows because the closed normal form of type $\tyR^n$ are tuples of numerals. In this case $M'$ must be a tuple of simple normal forms of the form $\pair{r,t}$ (notice that if $k=0$ in the first rule of \reffig{Expd}, then there is no $\beta$-redex at the right-hand side of $\expd$ in the conclusion and $t$ is a closed term, although it might be not a numeral, in case of type $\tyR^{\perp}$).
\end{proof}

\replemma{expd_ad}
{\it	For every program $x_1^{\tyR},\dots,x_n^{\tyR}\vdash M:\tyR$, $\seq r\in\Real^n$ and sequence $\seq u=u_1,\ldots,u_n$ of simple closed normal forms of suitable type, we have $M\isub{\seq r}{\seq x}\expd \FB M\isub{\pair{\seq r,\seq u}\!}{\seq x}$, where by $\{\pair{\seq r,\seq u}\!/\seq x\}$ we mean $\{\pair{r_1,u_1}/x_1\}\cdots\{\pair{r_n,u_n}/x_n\}$.
}
\begin{proof}
	First of all, a straightforward induction establishes that $M\expd \FB M$. Second, notice that $r_i\expd\pair{r_i,u_i}$ for any numeral $r_i$ and any simple closed normal form $u_i$ of suitable type, by the first rule of \reffig{Expd} with $k=0$. We then apply \reflemma{expd-substitution} to conclude.
\end{proof}

\replemma{expd-reduction}
{\it
Let $M\expd M'$ and $M\red N$, then there exists $N'$ such that $M'\reds N'$ and $N\expd N'$. 
}
\begin{proof}
Let $(\Ctxt,R,P)$ be the reduction step $M\red N$. The proof is an easy induction on $\Ctxt$. The only non-trivial part is the base of the induction, \ie\ $\Ctxt=\ctxthole$, in which the reasoning splits following \reffig{reduction}:
\begin{itemize}
\item the case of a $\beta$-reduction is a consequence of \reflemma{expd-substitution}. 

\item If $R=\phi(\seq r)$, then $M'\reds \pair{\phi(\pi_1\seq L'), t\isub{\seq L'}{\seq z}}$ for some $\seq L'$ of the same length as $\seq r$ such that $r_i\expd L'_i$ for all $i$. Notice that $r_i\expd L'_i$ implies that $\pi_1L'_i\red r_i$ as well as $L'_i$ is a simple closed normal form, so in particular $t\isub{\seq L'}{\seq z}$ is normalizable by \refprop{sn}. We therefore have:  $\pair{\phi(\pi_1\seq L'), t\isub{\seq L'}{\seq z}}\reds\pair{\sem[]\phi\!(\seq r),t'}$ with $t'$ a simple closed normal form, and we conclude by taking $N':=\pair{\sem[]\phi\!(\seq r),t'}$.

\item If $R=\ite{r}{L_1}{L_2}$ and $N=L_j$ for some $j\in\{1,2\}$, then we have $M'=\ite{\pi_1\pair{r,t}}{L_1'}{L_2'}$ for some closed simple normal form $t$, and $L_j\expd L_j'$. We conclude as $M'\reds L_j'$. 

\item The other redexes (products and fixpoints) are immediate. 
\end{itemize}
\end{proof}

\replemma{extrusion}
{\it
	Let $M\expd M'$, $t\expd t'$, $t'\apx M'$. Let $\sigma:M\red M_1$ be a head reduction step and moreover let $\xi:t\reds t_1$ be such that $\xi\apx\sigma$ (so in particular $t\apx M$ and $t_1\apx M_1$). Then there exist $M'_1,t'_1$ such that the following relations hold:
	\begin{center}
	\pgfmathsetmacro{\tx}{0.7}
	\pgfmathsetmacro{\ty}{0.8}
	\begin{tikzpicture}[every node/.style={circle,inner sep=1pt}]
	\node (t) at (0,0) {$t$};
	\node (M) at (0,1.2) {$M$};
	\node (t1) at (3,0) {$t_1$};
	\node[inner sep=0pt] (M1) at (3,1.2) {$M_1$};
	\draw[->] (t.east) -- node (xi) [ fill=white, inner sep=0pt, anchor=center, pos=0.5, font=\scriptsize] {$\xi$} node[at end, above, font=\footnotesize] {$*$} (t1.west);
	\draw[->] (M.east) -- node (sigma) [ fill=white, inner sep=0pt, anchor=center, pos=0.5, font=\scriptsize] {$\sigma$}  (M1.west);
	\draw[densely dotted] (sigma.south) -- node [fill=white, inner sep=0pt, anchor=center, pos=0.5, font=\scriptsize] {$\apx$} (xi.north);
	\node (t') at ($(t)+(\tx,\ty)$) {$t'$};
	\node (M') at ($(M)+(\tx,\ty)$) {$M'$};
	\node (t1') at ($(t1)+(\tx,\ty)$) {$t_1'$};
	\node[inner sep=0pt] (M1') at ($(M1)+(\tx,\ty)$) {$M_1'$};
	\draw[->] (t'.east) --  node[at end, above, font=\footnotesize] {$*$} (t1'.west);
	\draw[->] (M'.east) -- node[at end, above, font=\footnotesize] {$*$}   (M1'.west);
	\draw[densely dotted] (M'.south) -- node [fill=white, inner sep=0pt, anchor=center, pos=0.5, font=\scriptsize] {$\apx$} (t'.north);
	\draw[densely dotted] (M1'.south) -- node [fill=white, inner sep=0pt, anchor=center, pos=0.5, font=\scriptsize] {$\apx$} (t1'.north);
	\draw[densely dotted] (t.north east) -- node [fill=white, inner sep=0pt, anchor=center, pos=0.5, font=\scriptsize] {$\expd$} (t'.south west);
	\draw[densely dotted] (t1.north east) -- node [fill=white, inner sep=0pt, anchor=center, pos=0.5, font=\scriptsize] {$\expd$} (t1'.south west);
	\draw[densely dotted] (M.north east) -- node [fill=white, inner sep=0pt, anchor=center, pos=0.5, font=\scriptsize] {$\expd$} (M'.south west);
	\draw[densely dotted] (M1.north east) -- node [fill=white, inner sep=0pt, anchor=center, pos=0.5, font=\scriptsize] {$\expd$} (M1'.south west);
	\end{tikzpicture}
	\end{center}
}
\begin{proof}
By \refdef{ApxRed_steps}, $\sigma=\hctxt{\sigma_0}$ for some head context $\hCtxt$ and reduction step $\sigma_0:R\red P$. The proof is by induction on $\hCtxt$. The case $\hCtxt=\ctxthole$ splits following \reffig{reduction}. 
\begin{itemize}
\item Let $\sigma$ be $M = (\lambda x.L)N\red L\isub{N}{x}=M_1$, so that $\xi$ is the reduction $ t=(\lambda p.w\isub{\boldsymbol\pi p}{\seq x})\pair{\seq u}\reds w\isub{\seq u}{\seq x}=t_1$. Then $M'=(\lambda x.L')N'$ with $L\expd L'$ and $N\expd N'$, and $t'=(\lambda p.\overline w')\pair{\seq u'}$ with $w\isub{\boldsymbol\pi p}{\seq x}\expd\overline w'$ and $\pair{\seq u}\expd \pair{\seq u'}$. Moreover, since $t'\apx M'$, by \reflemma{apx-substs}, we have $\overline w' = w'\isub{\boldsymbol\pi p}{\seq x}$, with $w'\isub{\boldsymbol\pi p}{\seq x}\apx L'$, $p$ not free in $w'$, and $u_i'\apx N'$. Moreover, by induction on $w$, one can infer from $w\isub{\boldsymbol\pi p}{\seq x}\expd\overline w'$ that actually $w\expd w'$. 

We can then define: $M'_1 = L'\isub{N'}{x}$ and $t'_1 = w'\isub{\seq u'}{\seq x}$. Clearly $M'\red M_1'$, as well as $t'\reds t_1'$. Moreover, since $L\expd L'$ and $N\expd N'$, we have by \reflemma{expd-substitution} that $M_1\expd M_1'$. Similarly, from $\pair{\seq u}\expd \pair{\seq u'}$ and $w\expd w'$, we have $t'\expd t'_1$. Finally, \reflemma{apx-substs} gives us $t'_1\apx M_1'$.

\item Let $\sigma$ be the step $M=\funsa(\seq r)\red \sem[]{\funsa}\!(\seq r)=M_1$, so that $\xi=\sigma$, $t=M$, $t_1=M_1$, and we also have $M'=(\lambda \seq z^{\FB\tyR}.\pair{\funsa(\pi_1\seq z),w})\pair{\seq r, \seq u}$ and $t'=(\lambda \seq z^{\FB\tyR}.\pair{\funsa(\pi_1\seq z),\overline w})\pair{\seq r, \seq{\overline u}}$ with $ \overline w\apx w$ and $ \overline u_i\apx u_i$. Furthermore, notice that $w, \overline w$ $u_i,\overline u_i$ are normal forms of type $\tyR^{(\perp)}$ having only free variables of type $\FB\tyR$, so we may apply \reflemma{apx-programs} and infer $w=\overline w$ and $u_i=\overline u_i$. 

Let $w'$ be the normal form of $w\isub{\pair{\seq r,\seq u}\!}{\seq z}$ (which exists by \refprop{sn} and is unique by \refprop{confluence}) and define $M'_1:=t_1':=\pair{\sem[]{\funsa}\!(\seq r),w'}$. Clearly, $M'\reds M'_1$ and $t'\reds t'_1$, as well as $t_1'\apx M'_1$ by Lemma~\ref{lemma:apx-reflexivity}. Moreover, since $w'$ is a closed simple normal form of type $\tyR^{(\perp)}$, we have $M_1\expd M_1'$ as well as $t_1\expd t_1'$. 

\item Let $\sigma$ be $M=\ite r{L_1}{L_2}\red L_i=M_1$ with $i\in\{1,2\}$ depending on whether $r\leq 0$ or $r>0$. Then $\xi$ is the reduction $t=\pi_i\!\pair{u,u}\red u=t_1$ with $u\apx L_i$. Moreover, $M'=\ite{\pi_1\!\pair{r,w}}{L_{1}'}{L_2'}$ with $L_j\expd L_j'$ for all $j\in\{1,2\}$ and $t\expd t'$ gives $t'=\pi_i\!\pair{u',u'}$ with $u\expd u'$, while $t'\apx M'$ gives $u'\apx L_i'$. 

Let $M'_1:= L_i'$ and $t'_1:=u'$. Clearly $M'\reds M'_1$ and $t'\red u'$. We have also $t'_1\apx M'_1$ and $t_1\expd t_1'$.

\item The case of $\sigma$ being a projection reduction step is immediate.
\item The case of $\sigma$ being a fixpoint reduction step is analogous to the $\beta$-step.
\end{itemize}	

Of the other induction cases, the only subtle one is $\hCtxt=\ite{\overline\hCtxt}{N_1}{N_2}$. Under this hypothesis, the reduction $\xi$ is empty and $t_1=t=\pi_i\!\pair{u,u}$ with $u\apx N_i$ for some $i\in\{1,2\}$, as well as $M'=\ite{\pi_1\overline{M}'}{N_1'}{N_2'}$ with $\ctxt[\overline\hCtxt]{R}\expd \overline M'$, $N_1\expd N'_1$, $N_2\expd N'_2$ and $t'=\pi_i\!\pair{u',u'}$ with $u\expd u'$ and $u'\apx N_i'$ (notice that the index $i$ of the projection is the same in $t$ and $t'$ because $t\expd t'$). By Lemma~\ref{lemma:expd-reduction}, $\overline M'\reds L$ such that $\ctxt[\overline\hCtxt]{P}\expd L$. We can then conclude by setting $M'_1=\ite{\pi_1L}{N_1'}{N_2'}$ and $t'_1=t'$. 

All of the remaining cases follow the same pattern. For instance, let $\hCtxt=\funsa(L_1,\ldots,\overline\hCtxt,\ldots,L_k)$. Then, $\sigma=\funsa(L_1,\ldots,\overline\sigma,\ldots,L_k)$, for $\overline\sigma$ the head reduction step $\ctxt[\overline\hCtxt]{\sigma_0}$, and $\xi=\funsa(w_1,\ldots,\overline\xi,\ldots,w_k)$, with $\overline\sigma\apx\overline\xi$ and $\funsa(L_1,\ldots,\ctxthole,\ldots,L_k)\apx\funsa(w_i,\ldots,\ctxthole,\ldots,w_k)$. Let us denote by $\overline M:=\ctxt[\overline\hCtxt]{R}$ and 
$\overline M_1:=\ctxt[\overline\hCtxt]{P}$ the source and the target of $\overline\sigma$, respectively. Similarly, let us denote by $\overline t$ and $\overline t_1$ the source and the target of $\overline\xi$. We have $M'=u L_1'\cdots \overline M'\cdots L_k'$ with $u$ a closed simple normal form and $L_i\expd L'_i$ and $\overline M\expd \overline M'$, and similarly $t'=u' w_1'\cdots \overline t'\cdots w_k'$ with $u'$ a closed simple normal form and $w_i\expd w'_i$ and $\overline t\expd \overline t'$. Moreover, we have also $u\apx u'$ and $w_i'\apx L_i'$ and $\overline t'\apx \overline M'$. We can then apply the induction hypothesis on the quadruple $\overline\sigma$, $\overline\xi$, $\overline M'$, $\overline t'$, obtaining $\overline M'_1$ and $\overline t_1'$. We conclude by setting $M_1':=u L_1'\cdots \overline M'_1\cdots L_k'$ and $t'_1:=u' w_1'\cdots \overline t'_1\cdots w_k'$. Notice that $M'\reds M_1'$ (as well as $t'\reds t_1'$) even if this reduction is not under a head context.
\end{proof}

\replemma{extrusion_iterate}
{\it
	Let $M\expd M'$, $t\expd t'$, $t\trapx M$ and $t'\apx M'$, for $t$ and $M$ closed terms both of type $\tyR^n$, for some $n\geq 0$. Then, $t'\reds t''$ and $M'\reds M''$ with $t''$ and $M''$ normal such that $t''\apx M''$. 
}
\begin{proof}
By \refdef{TrApx}, there exist a normalizing reduction $\xi$ starting from $t$ and a normalizing reduction $\rho$ starting from $M$, such that $\xi\apx\rho$. The proof is by induction on the length of $\rho$. 

If the length is $0$, then $M$ is a closed normal form of type $\tyR^n$, so by \reflemma{apx-normalisation} also $t$ is a ground normal form of type $\tyR^n$. We can apply \reflemma{expd_closed_nf} and conclude that $t'$ and $M'$ are normal.

Otherwise, let $\xi=\upsilon\xi'$ and $\rho=\sigma\rho'$ such that $\upsilon:t\reds t_1$, $\sigma:M\red M_1$ and $\upsilon\apx\sigma$, $\xi'\apx\rho'$. In particular, $t_1\trapx M_1$. By \reflemma{extrusion}, there exist $M_1',t_1'$ such that $M_1\expd M_1'$, $t_1\expd t_1'$, $t_1'\apx M_1'$ and $M'\reds M_1'$, $t'\reds t_1'$. We can thus conclude by induction on the quadruple $M_1,M_1',t_1,t_1'$.
\end{proof}

\section{Proofs of Section~\ref{sect:Unsoundness}}\label{app:unsound}

The standard semantics of \PCF\ deals with recursive definitions (\ie, fixpoints) by considering them as suprema of finitary (\ie, fixpoint-free) approximations. In order to apply this idea to our setting we need to follow a more syntactic approach than the usual one, based on Scott domains. Indeed, our definition of stable point fundamentally uses traces (\refdef{TrApx}), and the latter are defined in terms of reduction sequences, which are abstracted away in Scott domains. We therefore introduce a further relation $\fixapx$ which approximates fixpoints within the syntax (\refprop{appr}) and which interacts well with the trace relation (Propositions~\ref{prop:FixApxComp-terms} and~\ref{prop:FixApxComp}).

\begin{definition}
	\label{def:fixapprox}
	The \emph{approximation relation $\fixapx$} between terms is defined by the following rules:
	$$
	\infer{x\fixapx x}{}
	\qquad
	\infer{\funsa(M_1,\dots,M_k)\fixapx \funsa(M_1',\dots,M_k')}{M_1\fixapx M_1'&\dots& M_k\fixapx M_k'}
	\qquad
	\infer[n\leq m]{\fix[n]{f}{M}\fixapx\fix[m]{f}{M'}}{M\fixapx M'}
	$$
(where $m,n\in\Nat\cup\{\infty\}$)
and all other rules lifting the relation homomorphically.
\end{definition}

\begin{lemma}[substitution]
	\label{lemma:fixapx_subst}
	We have:
	\begin{enumerate}
		\item\label{fixapx_subst:red} if $M\fixapx M'$ and $Q\fixapx Q'$, then $M\isub{Q}{x}\fixapx M'\isub{Q'}x$; 
		\item\label{fixapx_subst:exp} if $M\fixapx N'\isub{Q'}{x}$, then $M=N\isub{Q}{x}$ for some $N\fixapx N'$ and $Q\fixapx Q'$.
	\end{enumerate}
\end{lemma}
\begin{proof}
	Both points are by structural induction, on $M$ for point \ref{fixapx_subst:red} and on $N'$ for \ref{fixapx_subst:exp}.
\end{proof}

\begin{lemma}[monotonicity]
	\label{lemma:FixApxMono}
	If $M\fixapx M'$ and $M\red N$, then $M'\red N'$ such that $N\fixapx N'$.
\end{lemma}
\begin{proof}
By structural induction on $M$. In case $M=(\lambda x.P) Q\red P\isub{Q}{x}$, we have that $M'=(\lambda x.P')Q'$, with $P\leq P'$ and $Q\fixapx Q'$. So $M'\red P'\isub{Q'}{x}$ and we may conclude by \refsublemma{fixapx_subst}{red}.

Let $M=\fix[n+1]{f}{P}$, with $n\in\Nat\cup\{\infty\}$ and $\infty+1=\infty$. Then, $M\red P\isub{\lambda x.(\fix[n] fP)x}f$. Notice that $M'=\fix[m+1]{f}{P'}$, with some $m\geq n$ and $P'$ such that $P\fixapx P'$. We have $M'\red P'\isub{\lambda x.(\fix[m] fP')x}f$. Since $\lambda x.(\fix[n] fP)x\fixapx \lambda x.(\fix[m] fP')x$, we may conclude by \refsublemma{fixapx_subst}{red}.

The other cases are immediate. 
\end{proof}

\begin{lemma}[continuity]
	\label{lemma:FixApxCont}
	If $M'\red N'$ and $N\fixapx N'$, then there exists $M\fixapx M'$ such that $M\red N$. Moreover, if $N$ has no occurrence of $\mathrm{fix}_\infty$ apart from $\Omega$, then neither does $M$.
\end{lemma}
\begin{proof}
	By structural induction on $M'$. The cases $M'=(\lambda x.P')Q'$ or $M'=\fix[k+1]{f}{P'}$ are similar to the analogous cases in the proof of Lemma~\ref{lemma:FixApxMono}, using \refsublemma{fixapx_subst}{exp} instead of \refsublemma{fixapx_subst}{red}. The fact that $M$ has no occurrence of $\mathrm{fix}_\infty$ apart from $\Omega$ is a trivial consequence of supposing this for $N$. 

	If $M'=\ite{r}{L_1'}{L_2'}$, then $N'=L'_i$ for some $i\in\{1,2\}$. Then we can write $L_i:=N$ and chose an arbitrary $L_{3-i}\fixapx L_{3-i}'$ with no occurrence of $\mathrm{fix}_\infty$ apart from $\Omega$ and set $M:=\ite r{L_1}{L_2}$. 
\end{proof}

\begin{proposition}[fixpoints are suprema of approximations]
	\label{prop:appr}
	Given a program $\Gamma\vdash\fix fL:\tyR$, we have 
	\[
		\sem{\fix fL}=\sup_{k<\infty}\sem{\fix[k] fL},
	\]
	\ie, for every $\seq r$, $\sem{\fix fL}(\seq r)= q$ iff there is $k<\infty$, $\sem{\fix[k] fL}(\seq r)=q$.
\end{proposition}
\begin{proof}
	The right-to-left implication is an immediate consequence of a stronger statement:
	\begin{itemize}
	\item[($\star$)] given two programs $M,M'$ such that $M\fixapx M'$, $\sem M(\seq r)=q$ implies $\sem{M'}(\seq r)=q$.
	\end{itemize}
	In fact, suppose that there is a reduction $M\isub{\seq r}{\seq x}\reds q$, then by iterating \ref{lemma:FixApxMono}, we get $M'\isub{\seq r}{\seq x}\reds N'$ with $q\fixapx N'$, which implies $N'=q$. 

	Let us now prove the left-to-right implication. Suppose $\sem{\fix fL}(\seq r)= q$, \ie, there is a reduction $\fix fL\isub{\seq r}{\seq x}\reds q$. Since $q\fixapx q$, by iterating \ref{lemma:FixApxCont} we get $M\fixapx\fix fL\isub{\seq r}{\seq x}$ such that $M\reds q$. Moreover, since $q$ has no occurrence of $\mathrm{fix}_\infty$, then $M$ has no occurrence of $\mathrm{fix}_\infty$ apart from $\Omega$, so that $M=\fix[k] fL'$ for some $k<\infty$ and $L'\fixapx L$. We then conclude by claim ($\star$), as $M\fixapx \fix[k] fL$. 
\end{proof}

\begin{proposition}[composition with pre-trace]
	\label{prop:FixApxComp-terms}
	If $t\apx M\fixapx M'$, then $t\apx M'$.
\end{proposition}
\begin{proof}
	By induction on a derivation of $t\apx M$.
\end{proof}

\begin{lemma}
	\label{lemma:onesteptraceapprox}
	Let $M\fixapx M'$ and $\sigma: M\red N$ be a reduction step. If $\xi\apx\sigma$, then there exists $\sigma' :M' \red N'$ s.t.  $\xi\apx\sigma'$ and $N\fixapx N'$.
\end{lemma}
\begin{proof}
	Let $\sigma=(\hCtxt,R,P)$, so $M=\ctxt[\hCtxt] R$ and $N=\ctxt[\hCtxt] P$. Notice that $\xi\apx\sigma$ implies that $\hCtxt$ is a head context. The proof is by induction on $\hCtxt$. 

	If $\hCtxt=\ctxthole$, then we follow the cases of \reffig{reduction}.
	If $M=R=(\lambda x.L_1)L_2$ and $N=P=L_1\isub{L_2}{x}$, then $M\fixapx M'$ gives us $M'=(\lambda x.L_1')L_2'$ with $L_i\fixapx L_i'$. 
	Define $N':=L_1'\isub{L_2'}{x}$ and $\sigma':N\red N'$. \refsublemma{fixapx_subst}{red} gives us $N\fixapx N'$. By definition, $\xi$ must be of the form:
	\[
		(\lambda p.t\isub{\boldsymbol\pi p}{\seq x})\!\pair{\seq u}\red t\isub{\boldsymbol\pi\!\pair{\seq{u}}}{\seq x}\reds t\isub{\seq u}{\seq x} 
	\]
	where $p\apx x\vdash t\isub{\boldsymbol\pi p}{\seq x}\apx L_1$ and $u_i\apx L_2$ for all $u_i$ in $\seq u$. By Proposition~\ref{prop:FixApxComp-terms}, we have $t\isub{\boldsymbol\pi p}{\seq x}\apx L_1'$ and $u_i\apx L_2'$, so $\xi\apx\sigma'$.

	The case $M=\fix[k+1]{f}{P'}$ is similar to the previous one. All other cases are simpler and do not need \reflemma{fixapx_subst}.

	If $\hCtxt=\ite{\overline\hCtxt}{L_1}{L_2}$, then $\xi$ is the empty reduction. Moreover, $M\fixapx M'$ gives us $M'=
	\ite{\overline M'}{L_1'}{L_2'}$ with $\ctxt[\overline\hCtxt]{R}\fixapx \overline M'$ and $L_i\fixapx L_i'$.  We apply Lemma \reflemma{FixApxMono} to $\ctxt[\overline\hCtxt]{R}\fixapx \overline M'$ and $\ctxt[\overline\hCtxt]{R}\red \ctxt[\overline\hCtxt]{P}$, thus getting $\overline\sigma':\overline M'\red \overline N'$ with $\ctxt[\overline\hCtxt]{P}\fixapx \overline N'$. We then define $\sigma' := \ite{\overline\sigma'}{L_1'}{L_2'}$. Notice that $\xi\apx\sigma'$ as this latter fires a redex in the guard of a conditional. 

	If $\hCtxt=\overline\hCtxt L$, then $\xi=\overline\xi\!\pair{\seq u}$ with each $u_i\apx L$ and $\overline\xi\apx (\overline\hCtxt, R,P)$. Moreover, $M\fixapx M'$ gives us $M'= \overline M' L'$ with $\ctxt[\overline\hCtxt]{R}\fixapx \overline M'$ and $L\fixapx L'$. So we apply the induction hypothesis and get $\overline\sigma':\overline M'\reds \overline N'$ with $\ctxt[\overline \hCtxt]{P}\fixapx \overline N'$. We then define $\sigma':=\sigma L$ and $N':=\overline N'L$.

	All other induction cases are similar to the previous one. 
\end{proof}

\begin{proposition}[composition with trace]
	\label{prop:FixApxComp}
	Let $M\fixapx M'$ and $\rho: M\reds N$, then $\xi\apx\rho$ implies that there exists $\rho' :M' \reds N'$ s.t.  $\xi\apx\rho'$ and $N\fixapx N'$.
\end{proposition}
\begin{proof}
	By induction on the length of $\rho$. The base of the induction is given by \reflemma{onesteptraceapprox}.
\end{proof}


\replemma{stab}
{\it	We have the following inclusions, where the terms appearing in the statements are supposed to be typed under a ground context $\Gamma$.
	\begin{enumerate}
		\item Let $\funsa$ be a function symbol of arity $k$ and let $M_1,\dots,M_k$ be programs, then  $\bigcap_i\Stab{M_i}\subseteq\Stab{\funsa(M_1,\dots,M_k)}$.

		\item
		Let $R\red P$ be one of the rewriting rules in \reffig{reduction}, with $R,P$ of type $B_1\to\cdots\to B_p\to\tyR^m$. For all $1\leq i\leq p$, let $\Gamma\vdash L_i:B_i$. Then, for all $1\leq j\leq m$, $\Stab{\pi_j(P\seq L)}\subseteq\Stab{\pi_j(R\seq L)}$.

		\item Let $P:\tyR$ and $M_1,M_2$ be of type $B_1\to\cdots\to B_p\to\tyR^m$. For all $1\leq i\leq p$, let $\Gamma\vdash L_i:B_i$. Let $X_1:=\sem P^{-1}(\Real_{\leq 0})$ and $X_2:=\sem P^{-1}(\Real_{>0})$. 
		Then, for all $1\leq j\leq m$ and all $l\in\{1,2\}$, we have that $\Stab{\pi_j(M_l\seq L)}\cap\interior{X_l} \subseteq \Stab{\pi_j(\ite P{M_1}{M_2}\seq L)}$.
		
		\item 
		Let $B=B_1\to\cdots\to B_p\to\tyR^m$, let $\Gamma\vdash L_0:A$ and $\Gamma\vdash L_i:B_i$ for all $1\leq i\leq p$. For all $k\in\Nat$ and $1\leq j\leq m$, $\Stab{\pi_j((\fix[k]{f^{A\to B}}{M})\seq L)}\subseteq\Stab{\pi_j((\fix{f^{A\to B}}M)\seq L)}$, where $\seq L:=L_0,L_1,\ldots,L_p$.
	\end{enumerate}
}
\begin{proof}
	\begin{itemize}
		\item[\ref{stab:map}.]
		Let $\seq r\in\bigcap_i\Stab{M_i}$. By definition, we have for each $i$, some $\varepsilon_i>0$ and $t_i\apx M_i$ such that $\forall \seq r'\in \Ball{\varepsilon_i}{\seq r}$, $t_i\isub{\seq r'}{\seq x}\trapx M_i\isub{\seq r'}{\seq x}$. Define $\varepsilon:=\min_i(\varepsilon_i)$, $u:=\phi(t_1,\dots,t_k)$ and notice that $u\apx\phi(M_1,\dots,M_k)$. Let us prove $u\isub{\seq r'}{\seq x}\trapx \phi(M_1,\dots,M_k)\isub{\seq r'}{\seq x}$, for every $\seq r'\in \Ball{\varepsilon}{\seq r}$.

		For each $i$, by hypothesis there exist two normalizing reductions $\xi_i$ and $\rho_i$ from respectively $t_i\isub{\seq r'}{\seq x}$ and $M_i\isub{\seq r'}{\seq x}$ such that $\xi_i\apx\rho_i$ (notice that the choice of these reductions may depend on $\seq r'$). In particular, this means that both $\xi_i$, $\rho_i$ end into a numeral $p_i$. Let  $\sigma: \phi(p_1,\dots,p_k)\red \sem\phi(p_1,\dots,p_k)$ and notice that
		\begin{align*}
			& \phi(\xi_1,t_{2}\isub{\seq r'}{\seq x},\dots,t_k\isub{\seq r'}{\seq x})\cdots\phi(p_1,\dots,p_{k-1},\xi_k)\sigma\\
			\apx\ &
			\phi(\rho_1,M_{2}\isub{\seq r'}{\seq x},\dots,M_k\isub{\seq r'}{\seq x})\cdots\phi(p_1,\dots,p_{k-1},\rho_k)\sigma.
		\end{align*}
		This proves $u\isub{\seq r'}{\seq x}\trapx \phi(M_1,\dots,M_k)\isub{\seq r'}{\seq x}$. 

		\item[\ref{stab:red}.]
		Let $\sigma:R\red P$ be one of the rewriting rules of \reffig{reduction} and suppose $\seq r\in\Stab{\pi_j(P\seq L)}$.  By definition, there exist $\varepsilon>0$ and $u\apx \pi_j(P\seq L)$ such that $\forall  \seq r'\in\Ball{\varepsilon}{\seq r}$, we have $u\isub{\seq r'}{\seq x}\trapx \pi_j(P\seq L)\isub{\seq r'}{\seq x}=\pi_j(P\isub{\seq r'}{\seq x}\seq L\isub{\seq r'}{\seq x})$. 

		This means that $u=\pi_j(u'\!\pair{\seq u''_1}\cdots\pair{\seq u''_p})$ with $u'\apx P$ and, for every $1\leq i\leq p$ and every element $u''_{i,h}$ of $\seq u''_i$, we have $u''_{i,h}\apx L_i$. Moreover, fixing  $\seq r'\in\Ball{\varepsilon}{\seq r}$, there are two normalizing reduction sequences $\nu\apx\rho$ from $u\isub{\seq r'}{\seq x}$ and $\pi_j(P\isub{\seq r'}{\seq x}\seq L\isub{\seq r'}{\seq x})$ respectively. In what follows, we will abbreviate the sequence of successive applications $\pair{\seq u_1''}\cdots\!\pair{\seq u_p''}$ as $\seq u''$.

		By \reflemma{apx_expansion}, there is $t'\apx R$ and $\xi:t'\reds u'$ such that $\xi\apx\sigma$. Notice that the definition of $\xi$ and $t'$ does not depend on $\seq r'$. In fact, by \reflemma{uniformity} we have $\xi\isub{\seq r'}{\seq x}\apx\sigma\isub{\seq r'}{\seq x}$. We then have, by Definitions~\ref{def:ApxRed_steps} and~\ref{def:ApxRed_sequences},
		$$
			\pi_j(\xi\isub{\seq r'}{\seq x}\seq u''\isub{\seq r'}{\seq x})\nu
			\apx
			\pi_j(\sigma\isub{\seq r'}{\seq x}\seq L\isub{\seq r'}{\seq x})\rho.
		$$
		This means that $\forall \seq r'\in\Ball{\varepsilon}{\seq r}, \pi_j(t'\seq{u''})\isub{\seq r'}{\seq x}\trapx \pi_j(R\seq L)\isub{\seq r'}{\seq x}$, so $\seq r\in\Stab{\pi_j(R\seq L)}$.

		\item[\ref{stab:if}.]
		Let $l\in\{1,2\}$ and let $\seq r\in\Stab{\pi_j(M_l\seq L)}\cap\interior{X_l}$. By definition, there exist $\varepsilon>0$, $u\apx \pi_j(M_l\seq L)$ such that, for all $\seq r'\in\Ball\varepsilon{\seq r}$, $u\isub{\seq r'}{\seq x}\trapx \pi_j(M_l\seq L)\isub{\seq{r'}}{\seq x}$ and $\seq r'\in\interior{X_l}$. Notice that $u=\pi_j(u'\seq{u''})$, with $u'\apx M_l$ and $\seq u''=\pair{\seq u_1''}\cdots\pair{\seq u_p}$ such that, for every $1\leq i\leq p$ and every element $u''_{i,h}$ of $\seq u''_i$, we have $u''_{i,h}\apx L_i$. Let $t := \pi_j((\pr \ell\!\pair{u',u'})\seq u'')$ and notice that $t\apx \pi_j(\ite P{M_1}{M_2}\seq L)$. Let us prove that $t\isub{\seq r'}{\seq x}\trapx \pi_j(\ite P{M_1}{M_2}\seq L)\isub{\seq r'}{\seq x}$. 

		Since $\seq r'\in\interior{X_l}\subseteq\dom P$, we have that $P\isub{\seq r'}{\seq x}$ is normalizing. Since $P$ is ground, by \refprop{headnf} there is a head reduction sequence $\rho:P\isub{\seq r'}{\seq x}\reds q$ such that $q$ is a numeral. Let now $\hCtxt=\pi_j(\ite{	\ctxthole}{M_1}{M_2}\seq u'')\isub{\seq r'}{\seq x}$. Notice that $\nu\apx\hctxt\rho$ for $\nu$ the empty reduction sequence of $t\isub{\seq r'}{\seq x}$. Moreover, by hypothesis we have normalizing reduction sequences $\nu'\apx\rho'$ from $u\isub{\seq r'}{\seq x}=\pi_j(u'\seq{u''})\isub{\seq r'}{\seq x}$ and $\pi_j(M_l\seq L)\isub{\seq r'}{\seq x}$, respectively. Furthermore, we have the head reduction steps:
		\begin{align*}
			\nu_0:&\pi_j(\pi_\ell\pair{u',u'}\seq{u''})\isub{\seq r'}{\seq x}\red\pi_j(u'\seq{u''})\isub{\seq r'}{\seq x} \\
			\rho_0:&\pi_j(\ite q{M_1}{M_2}\seq L)\isub{\seq r'}{\seq x}\red \pi_j(M_l\seq L)\isub{\seq r'}{\seq x}
		\end{align*}
		such that $\nu_0\apx\rho_0$. 
		We then have : 
		$$\nu\nu_0\nu' \apx \hctxt\rho\rho_0\rho',$$
		which allows us to conclude.

		\item[\ref{stab:fix}.]
		Let $P:=\pi_j((\fix[k]{f}{M})\seq L)$, $P':=\pi_j((\fix{f}{M})\seq L)$ and $\seq r\in\Stab{P}$. By definition, we have $t\apx P$ such that  for some $\varepsilon>0$ and all $\seq r'\in\Ball\varepsilon{\seq r}$, there are normalizing reduction sequences $\xi\apx\rho$ starting from $t\isub{\seq r'}{\seq x}$ and $P\isub{\seq r'}{\seq x}$, respectively.

		Notice that $P\fixapx P'$, so by \refprop{FixApxComp} we have $t\apx P'$ as well as $\rho':P'\isub{\seq r'}{\seq x}\reds P'_0$ with $\xi\apx\rho'$ and $q\fixapx P'_0$, with $q$ the target of $\rho$ (a normal form). By an inspection of the rules defining $\fixapx$, we infer $P_0'=q$, so $\rho'$ is normalizing. We conclude $t\isub{\seq r'}{\seq x}\trapx P'\isub{\seq r'}{\seq x}$, which proves that $\seq r\in\Stab{P'}$.
	\end{itemize}
\end{proof}

\replemma{map}
{\it
	Let $\funsa$ be a function symbol of arity $k$ and, for each $1\leq i\leq k$, $M_i\in\logic{\tyR}$. If $\sem[]{\funsa}$ is cqc, then $\funsa(M_1,\dots,M_k)\in\logic{\tyR}$.
}
\begin{proof}
	Let us write $\seq M:=M_1,\ldots,M_k$. By \refprop{head}, $\sem{\funsa(\seq M)}=\sem[]{\funsa}\circ\pair{\sem{M_1},\dots,\sem{M_k}}$ which is cqc by \reflemma{qcont}, items \ref{qcont:comp} and \ref{qcont:pairing}.

	For what concerns the unstable points, using \refsublemma{stab}{map} and the fact that $\dom{\funsa(\seq M)}\subseteq\dom{M_i}$ for all $1\leq i\leq k$, we have
	$$\UStab{\funsa(\seq M)}\subseteq\dom{\funsa(\seq M)}\setminus\bigcap_{i=1}^k\Stab{M_i}=\bigcup_{i=1}^k\dom{\funsa(\seq M)}\setminus\Stab{M_i}\subseteq\bigcup_{i=1}^k\dom{M_i}\setminus\Stab{M_i}=\bigcup_{i=1}^k\UStab{M_i},$$
	and the latter set is a quasivariety by hypothesis.
\end{proof}

\replemma{expansion_log}
{\it	Let $R\red P$ be one of the rewriting rules in Figure~\ref{fig:reduction}. If $P\in\logic{A}$, then $R\in\logic{A}$.}
\begin{proof}
	Let $A = A_1\rightarrow \dots \rightarrow A_p\rightarrow \tyR^m$.  It is enough to prove that, given $\seq L=L_1,\ldots,L_p$ such that $L_i\in\logic{A_i}$ for all $1\leq i\leq p$, we have $\pi_j(R\seq L)\in\logic{\tyR}$ for every $1\leq j\leq m$.

	By the confluence property, we have $\sem{\pi_j(R\seq L)}=\sem{\pi_j(P\seq L)}$, so the former is cqc as the latter is. 

	Concerning the set of unstable points, since $\dom{\pi_j(R\seq L)}=\dom{\pi_j(P\seq L)}$, by \refsublemma{stab}{red} $\UStab{\pi_j(R\seq L)}\subseteq \UStab{\pi_j(P \seq L)}$, and the latter is a quasivariety by hypothesis.
\end{proof}

\replemma{zero}
{\it	For every type $A\to B$, $\Omega_{A\to B}\in\logic{A\to B}$.}
\begin{proof}
	Recall that, by definition, $\Omega_{A\to B}=\fix{f^{A\to B}}f$. Let $B=B_1\rightarrow\dots B_p\rightarrow \tyR^m$. It is enough to prove that, for all $\seq L=L_0,L_1,\ldots,L_p$ such that $L_0\in\logic A$ and $L_i\in\logic{A_i}$ for all $1\leq i\leq p$, we have $N:=\pi_j((\fix ff)\seq L) \in\logic {\tyR}$ for all $1\leq j\leq m$. This follows immediately from
	\begin{itemize}
	\item[$(\star)$] for any sequences (including empty) $\seq P$ and $\seq z$ of terms and fresh variables, respectively, $\sem{\pi_j((\lambda \seq z.(\fix ff)\overline{\seq z})\seq P)}$ is the nowhere-defined function, where by $\overline{\seq z}$ we mean the application to all variables in $\seq z$, but in reverse order.
	\end{itemize}
	Indeed, $(\star)$ implies that $\sem{N}$ is the nowhere-defined function, which is trivially cqc, as well as that $N$ diverges, which means that $\dom N=\emptyset$, hence $\UStab{N}=\emptyset$, and the empty set is a quasivariety.
	
	Claim $(\star)$ may be proved by contradiction: we suppose that  $((\lambda \seq z.(\fix ff)\overline{\seq z})\seq P)\isub{\seq r}{\seq x}$ has a normalizing reduction sequence $\rho$, which is necessary for the semantics to be defined in $\seq r\in\Real^n$, and we derive a contradiction by induction on its length. If $\rho$ is empty, then we have a contradiction since the term is not normal. Otherwise, let $\sigma$ be the first reduction step of $\rho$. If the redex fired by $\sigma$ is in any subterm of $\seq P$, then the target of $\sigma$ is of the form $((\lambda \seq z.(\fix ff)\overline{\seq z})\seq P')\isub{\seq r}{\seq x}$ and the induction hypothesis gives us a contradiction. In case $\seq z=z\seq z'$ and $\seq P=P\seq P'$, $\sigma$ may be the step $((\lambda z\seq z'.(\fix ff)\overline{\seq z'}z)P\seq P)\isub{\seq r}{\seq x}\red ((\lambda\seq z'.(\fix ff)\overline{\seq z'})P\seq P')\isub{\seq r}{\seq x}$, which is still if the form to which the induction hypothesis applies. The last possibility is that $\sigma$ is the step $((\lambda \seq z.(\fix ff)\overline{\seq z})\seq P)\isub{\seq r}{\seq x}\red ((\lambda \seq zz'.(\fix ff)z'\overline{\seq z})\seq P)\isub{\seq r}{\seq x}$, and also in this case we obtain a contradiction by applying the induction hypothesis. 
\end{proof}

\replemma{fix}
{\it	If $\forall k\in\Nat, \fix[k] f M\in\logic{A\to B}$, then $\fix f M\in\logic{A\to B}$.}
\begin{proof}
	Let $B=B_1\to\cdots\to B_p\to\tyR^m$. It is enough to show that, given $\seq N=N_0,N_1,\ldots,N_p$ such that $N_0\in\logic{A}$ and $N_i\in\logic{B_i}$ for all $1\leq i\leq p$, we have $P:=\pi_j((\fix f M)\seq N)\in\logic{\tyR}$ for all $1\leq j\leq m$. In what follows, we let, for arbitrary $k\in\Nat$, $P_k:=\pi_j((\fix[k]{f}{M})\seq N)$.	
	First, let us prove that $\sem P$ is cqc. Let $l\geq 0$ and let $Q\subseteq\Real^{l+1}$ be quasiopen. We need to prove that $(\id_{\Real^l}\times\sem{P})^{-1}(Q)$ is quasiopen. By \refprop{appr}, we have
	$$(\id_{\Real^l}\times\sem{P})^{-1}(Q)=\bigcup_{k<\infty}(\id_{\Real^l}\times\sem{P_k})^{-1}(Q),$$
	and the latter union is quasiopen because by hypothesis each $\sem{P_k}$ is cqc.
		
	Let us now prove that $\UStab{P}$ is a quasivariety. Notice that the above equality gives us (with $l=0$) $\dom P=\bigcup_{k<\infty}\dom{P_k}$. Then, using \refsublemma{stab}{fix} and the latter observation, we may write
	\begin{align*}
		\UStab{P} &\subseteq \dom{P}\setminus\bigcup_{m<\infty}\Stab{P_m} 
		= \left(\bigcup_{k<\infty}\dom{P_k}\right)\setminus\bigcup_{m<\infty}\Stab{P_m} \\
		&\subseteq \bigcup_{k<\infty}\left(\dom{P_k}\setminus\bigcup_{m<\infty}\Stab{P_m}\right) 
		\subseteq \bigcup_{k<\infty} \dom{P_k}\setminus\Stab{P_k}=\bigcup_{k<\infty}\UStab{P_k},
	\end{align*}
	and the latter union is a quasivariety because each $\UStab{P_k}$ is a quasivariety by hypothesis.
\end{proof}

\replemma{Adequacy}
{\it	Suppose that the primitive functions of \PCF\ are admissible. Let $\Gamma:=x_1^\tyR, \dots, x_n^\tyR$, $\Delta:= y_1^{A_1},\ldots,y_m^{A_m}$, let $\Gamma,\Delta\vdash M:A$ and let $\Gamma\vdash N_i\in \logic{A_i}$ for all $1\leq i\leq m$. Then, $$M\isub{N_1}{y_1}\cdots\isub{N_m}{y_m}\in\logic A.$$
}
\begin{proof}
	We reason by induction on a derivation of $\Gamma,\Delta\vdash M:A$. Throughout the proof, we write $\overline P:=P\isub{N_1}{y_1}\cdots\isub{N_m}{y_m}$ for a generic term $P$.

	If $M$ is a variable in $\Gamma$, then $A=\tyR$ and $\sem{\overline M}=\sem{x_i}$ is a projection, so cqc by \refsublemma{qcont}{basic}. Notice also that $\UStab{x_i}=\emptyset$. 

	If $M$ is a variable in $\Delta$, then  $\overline M = N_i$ for some $i\leq m$ and we conclude by the hypothesis  $N_i\in \logic{A_i}$.

	If $M=\lambda z^B.L$, then $A=B\rightarrow C$. We need to prove that for every $N\in\logic{B}$, $\overline MN\in\logic C$. Notice that $\overline MN\red \overline L\isub{N}{z}$, and $\overline L\isub{N}{z}\in\logic C$ by induction hypothesis, so we conclude by \reflemma{expansion_log}.

	Let $M=LP$, with $\Gamma,\Delta\vdash L:B\rightarrow A$ and $\Gamma,\Delta\vdash P:B$. By induction hypothesis, $\overline L\in\logic{B\rightarrow A}$ and $\overline P\in\logic B$, so by definition $\overline M=\overline L\overline P\in\logic A$. 

	If $M=\funsa(M_1,\dots,M_k)$, we apply \reflemma{map}, via \refsublemma{qcont}{basic} and the admissibility assumption.

	If $M=\pair{N_1,\dots,N_k}:A_1\times\cdots\times A_k$, we need to prove $\pi_i \overline M\in\logic{A_i}$ for any $i\leq k$. Notice that $\pi_i \overline M\red  \overline N_i\in\logic{A_i}$ so we conclude by~\reflemma{expansion_log}.

	If $M=\ite{P}{L}{N}$, we just apply the induction hypothesis and \reflemma{if}.

	Let $M=\fix f L$, with $\Gamma,\Delta, f:A\rightarrow B\vdash L:A\rightarrow B$.  By \reflemma{zero}, $\fix[0]{f}{L}\in\logic{A\rightarrow B}$ and by the induction hypothesis, $\lambda f.L\in\logic{(A\rightarrow B)\rightarrow A\rightarrow B}$, so that, for every $k\in\mathbb N$, $\fix[k]{f}{L}\in\logic{A\rightarrow B}$. We then conclude by \reflemma{fix}.
\end{proof}

\end{document}